\documentclass[10pt]{article}

\usepackage{url}
\usepackage{mathtools}
\DeclarePairedDelimiter{\ceil}{\lceil}{\rceil}

\usepackage{amssymb}
\usepackage{amsthm}
\usepackage{empheq}
\usepackage{latexsym}
\usepackage{enumitem}
\usepackage{eurosym}
\usepackage{dsfont}
\usepackage{appendix}
\usepackage{color} 
\usepackage[unicode]{hyperref}
\usepackage{frcursive}
\usepackage[utf8]{inputenc}
\usepackage[T1]{fontenc}
\usepackage{geometry}
\usepackage{multirow}
\usepackage{todonotes}
\usepackage{lmodern}
\usepackage{anyfontsize}
\usepackage{pgfplots}
\usepackage{stmaryrd}
\usepackage{cleveref}
\usepackage{natbib}
\usepackage{upgreek}
\usepackage{comment}
\usepackage{subfig}
\usepackage{graphicx}

\usepackage{floatrow}
\newfloatcommand{capbtabbox}{table}[][\FBwidth]
\input Acorn.fd

\pgfplotsset{compat=1.17}

\setcitestyle{numbers,open={[},close={]}}

\definecolor{red}{rgb}{0.7,0.15,0.15}
\definecolor{green}{rgb}{0,0.5,0}
\definecolor{blue}{rgb}{0,0,0.7}
\hypersetup{colorlinks, linkcolor={red},citecolor={green}, urlcolor={blue}}
			
\makeatletter \@addtoreset{equation}{section}

\newtheorem{theorem}{Theorem}[section]
\newtheorem{assumption}[theorem]{Assumption}

\newtheorem{lemma}[theorem]{Lemma}
\newtheorem{proposition}[theorem]{Proposition}

\newtheorem{definition}[theorem]{Definition}
\newtheorem{remark}[theorem]{Remark}

\DeclareUnicodeCharacter{014D}{\=o}
\newcommand{\smallfont}[1]{\text{\fontsize{4}{4}\selectfont$#1$}}
\newcommand{\tinyfont}[1]{\text{\fontsize{3}{3}\selectfont$#1$}}

\def \E{\mathbb{E}}
\def \F{\mathbb{F}}

\def \H{\mathbb{H}}
\def \I{\mathbb{I}}

\def \P{\mathbb{P}}

\def \R{\mathbb{R}}
\def \S{\mathbb{S}}

\def \Ho{\overline{\H}^{_{\raisebox{-1pt}{$ \scriptstyle 2,2$}}}}

\def\Ac{{\cal A}}
\def\Bc{{\cal B}}
\def\Cc{{\cal C}}

\def\Ec{{\cal E}}
\def\Fc{{\cal F}}

\def\Hc{{\cal H}}
\def\Ic{{\cal I}}

\def\Kc{{\cal K}}
\def\Lc{{\cal L}}

\def\Pc{{\cal P}}

\def\Tc{{\cal T}}

\def\Xc{{\cal X}}
\def\Yc{{\cal Y}}
\def\Zc{{\cal Z}}

\def\Cf{\mathfrak{C}}

\def\Zf{\mathfrak{Z}}

\def\zf{\mathfrak{z}}

\def\Ar{{\rm A}}

\def\Ur{{\rm U}}
\def\Vr{{\rm V}}
\def\Pr{{\rm P}}

\def\eps{\varepsilon}
\def\d{{\mathrm{d}}}
\def\1{\mathbf{1}}
\def\e{{\mathrm{e}}}

\def\as{\text{\rm--a.s.}}
\def\ae{\text{\rm--a.e.}}

\renewcommand{\t}{\top}

\DeclareMathOperator*{\argmax}{arg\,max} 
\DeclareMathOperator*{\es}{ess\,sup^{\P}}

\def\eps{\varepsilon}

\allowdisplaybreaks
\begin{document}

\title{Time-inconsistent contract theory\footnote{Camilo Hern\'andez acknowledges the support of a Presidential Postdoctoral Fellowship at Princeton University, a Chapman Fellowship at Imperial College London and a CKGSB fellowship at Columbia University.}}

\author{Camilo {\sc Hern\'andez} \footnote{Princeton University, ORFE department, USA, camilohernandez@princeton.edu.} \and Dylan {\sc Possama\"{i}} \footnote{ETH Z\"urich, Mathematics department, Switzerland, dylan.possamai@math.ethz.ch.}}

\date{\today}

\maketitle

\begin{abstract}
This paper investigates the moral hazard problem in finite horizon with both continuous and lump-sum payments, involving a time-inconsistent sophisticated agent and a standard utility maximiser principal. Building upon the so-called dynamic programming approach in \citet*{cvitanic2015dynamic} and the recently available results in \citet*{hernandez2020me}, we present a methodology that covers the previous contracting problem. Our main contribution consists in a characterisation of the moral hazard problem faced by the principal. In particular, it shows that under relatively mild technical conditions on the data of the problem, the supremum of the principal's expected utility over a smaller restricted family of contracts is equal to the supremum over all feasible contracts. Nevertheless, this characterisation yields, as far as we know, a novel class of control problems that involve the control of a forward Volterra equation via Volterra-type controls, and infinite-dimensional stochastic target constraints. Despite the inherent challenges associated to such a problem, we study the solution under three different specifications of utility functions for both the agent and the principal, and draw qualitative implications from the form of the optimal contract. The general case remains the subject of future research.

\vspace{5mm}
 \vspace{5mm}

\end{abstract}

In this paper, we are interested in the moral hazard contracting problem between a principal and an agent with time-inconsistent preferences. A principal--agent problem pertains to the optimal contracting between two parties: the principal, who is interested in hiring the agent, offers a contract; provided the agent accepts, he can influence a random process, the outcome, via his actions. A key feature in these models is the amount of information available to the principal when designing the contract. There are three classical cases studied in the literature: risk-sharing with symmetric information, hidden action, and hidden type.  We are only concerned with the first two in this work.\medskip

In the risk-sharing scenario, also referred to as the first-best, both parties have the same information and have to agree on how to share the underlying risk. The principal thus has all the bargaining power, \emph{i.e.} she offers the contract and dictates the agent's actions---the agent is compelled to follow or else he would be severely penalised. In the case of hidden actions, the principal is imperfectly informed about the agent's actions. Either they are too costly to be monitored or simply unobservable. Consequently, the principal expects to receive a second-best utility compared to the risk-sharing case. As the agent is allowed to take actions that are not in the principal's best interest, this situation is also referred to as moral hazard, and incentives play a crucial role. Indeed, the principal hopes to influence the agent's actions by offering an appropriate contract. 
\medskip

In the case of a traditional (time-consistent) agent, a common feature of these models is that their resolution boils down to standard stochastic control theory. Indeed, in light of the principal's bargaining power, the first-best case is always cast as a stochastic control problem for a single individual---the principal---who chooses both the contract and the actions under the participation constraint. On the other hand, in the second-best problem, it being a two-stage Stackelberg game, one has to solve the agent's problem for any given fixed contract before moving to study the principal's problem. In principle, this creates a much more complicated structure on the problem. Since the introduction of the continuous-time model, it took time for the literature to present a general approach that arrived at the same conclusion for the second-best problem.\medskip

The study of moral hazard problems in continuous time has its roots in the seminal paper of \citet*{holmstrom1987aggregation}. In this model, the principal and the agent have CARA utility functions, and the agent's effort influences the drift of the output process, the solution to a controlled diffusion, but not the volatility. The resulting optimal contract is a linear function of the aggregate output. The model in \cite{holmstrom1987aggregation} drew great attention as the resolution of the, seemingly more complicated, continuous-time formulation was actually much more tractable, could be rigorously justified, and provided useful explicit solutions for the economic analysis. These were typically harder to reach in most of the discrete-time models that dominated the existing literature, see \citet*{laffont2002theory} for an overview. Following upon \cite{holmstrom1987aggregation}, \citet*{schattler1993first,schattler1997optimal} studied the validity of the so-called first-order approach, while \citet*{sung1995linearity,sung1997corporate} provided extensions to the case of diffusion control and hierarchical structures. The linearity of the optimal contract, a feature also present in \cite{sung1995linearity}, is further studied in \citet*{muller1998first, muller2000asymptotic}, \citet*{hellwig2002discrete}, \citet*{hellwig2007role} and \citet*{sung2005optimal,sung2021optimal} for the first-best problem, the interplay between the discrete-time and continuous-time models, and for a robust setting, respectively. Notably, \citet*{williams2015solvable} and \citet*{cvitanic2009optimal} characterise the optimal contract for general utilities by means of the so-called stochastic maximum principle and forward--backward stochastic differential equations (FBSDEs for short)\footnote{We refer to the monograph \citet*{cvitanic2012contract} for a general framework that systematically surveys a great portion of the literature exploiting the maximum principle, in models driven by Brownian motion.}. 
\medskip

Nevertheless, it was not until the approach in \citet*{sannikov2008continuous,sannikov2012contracts} was available that the study of the moral hazard problem was, once again, reinvigorated and arrived finally at the methodical program presented in \citet*{cvitanic2014moral,cvitanic2015dynamic}. In a nutshell, this method leverages the dynamic programming principle and the theory of backward stochastic differential equations (BSDEs) to reformulate the principal's problem as a standard optimal stochastic control problem with an additional state variable, namely, the agent's continuation utility.
This methodology has been extended to several scenarii including random horizon contracting \citet*{lin2020random}, ambiguity features from the point of view of the principal, as in \citet*{mastrolia2015moral} and \citet*{hernandez2019moral}, a principal contracting a finite number of agents \citet*{elie2019contracting}, several principals contracting a common agent \citet*{mastrolia2018principal}, a principal contracting a mean-field of agents \cite{elie2019tale}, and applications in optimal electricity demand response contracting \citet*{aid2018optimal}, or \citet*{elie2021mean}. The road map suggested by this approach is quite clear: $(i)$ identify the generic dynamic programming representation of the agent's value process, $(ii)$ express the contract payment in terms of the value process, $(iii)$ optimise the principal's objective over such payments. \medskip

All in all, the previous literature is particular to the contracting problem between two (or more) standard utility maximisers, while there is a growing need for the development of models able to explain the behaviour of agents that fail to comply with classical rationality assumptions. Indeed, there is clear evidence of such attitudes in a number of applications, from consumption problems to finance, from crime to voting, and from charitable giving to labour supply, see \citet*{rabin1998psychology} and \citet*{dellavigna2009psychology} for detailed reviews. The distinctive feature in these situations is that human beings do not necessarily behave as perfectly rational decision-makers. In reality, their criteria for evaluating their well-being are, in many cases, a lot more involved than the ones considered in the classic literature. In light of the methodology introduced in \cite{cvitanic2015dynamic}, the recently available results in \citet*{hernandez2020me} unveil the possibility of extending this blueprint to cover the moral hazard problem between a principal and a \emph{sophisticated} time-inconsistent agent. This is the task we seek to accomplish in this paper. \medskip

Time-inconsistency is, in general terms, the fact that marginal rates of substitution between goods consumed at different dates change over time, see \citet*{strotz1955myopia}, \citet*{laibson1997golden}, \citet*{odonoghue1999doing,odonoghue1999incentives}. For example, the marginal rate of substitution between immediate consumption and some later consumption is different from when these two dates were seen from a remote prior date. In many applications, this introduces a conflict between `an impatient \emph{present self} and a patient \emph{future self}', see \citet*{brutscher2011payment}. In mathematical terms, this translates into stochastic control problems in which the classic dynamic programming principle, or in other words, the Bellman optimality principle is not satisfied.\medskip

Time-inconsistency was first mentioned in \cite{strotz1955myopia} where three different types of agents are described: the pre-committed agent does not revise his initially decided strategy; the naive agent revises his strategy without taking future revisions into account; the sophisticated agent revises his strategy taking possible future revisions into account, and by avoiding such makes his strategy time-consistent. The comprehensive study of sophisticated agents started with \citet*{ekeland2008investment}, see also \citet*{ekeland2006being, ekeland2010golden}, which later became the starting point of the general Markovian theory developed by \citet*{bjork2017time}. Nonetheless, none of these approaches could handle the typical non-Markovian problems that would necessarily arise in contracting problems involving a principal and a time-inconsistent agent. \citet*{hernandez2020me} provided a probabilistic formulation able to accommodate a general non-Markovian structure and provided an extended dynamic programming principle (DPP) for a refinement of the notion of equilibria first introduced in \cite{ekeland2010golden}. In turn, the extended DPP leads to the introduction of a system of BSDEs analogous to the classical HJB equation. This system is fundamental in the sense that its well-posedness is both necessary and sufficient to characterise the value function and equilibria, which are identified as maximisers of the Hamiltonian.
\medskip

When it comes to incorporating time-inconsistent features into contract theory models, the economic literature is abundant in discrete-time models with two and up to three periods. A common feature in this literature is adopting quasi-hyperbolic discounting structures to draw conclusions in different mechanism design problems. Yet, the method of resolution in each problem remained limited to a case-by-case analysis. For instance, \citet*{amador2006commitment} and \citet*{bond2017commitment} study the feasibility of commitment in models of consumption and savings, whereas \citet*{galperti2015commitment} considers the optimal provision of commitment devices to people who value both commitment and flexibility. \citet*{bisin2015government} examines policymakers' responses to the political demands of agents with self-control problems, \citet*{halac2014fiscal} looks into a fiscal policy model in which the government has time-inconsistent preferences, while \citet*{lim2018dynamic} assesses the effects of time-inconsistency on monopoly regulation of electricity distribution. \citet*{heidhues2010exploiting} and \citet*{karaivanov2018markov} integrate time-inconsistent preferences into credit, mortgage, and insurance contract design problems, respectively. \citet*{englmaier2020long}, \citet*{gottlieb2008competition} and \citet*{gottlieb2021long} study contracting problems between firms and sophisticated, partially naive, and naive present-biased consumers. \citet*{yilmaz2013repeated, yilmaz2015contracting} considers a repeated moral hazard problem involving a sophisticated and naive agent, respectively. \citet*{ma1991adverse} studies a multi-period model in which contracts are subject to renegotiations, and the agent's action has a long-term effect. \citet*{balbus2020time} shows the existence of time-consistent equilibria for dynamic models with generalised discounting. A survey of some of the state of behavioural economics research in contract theory was provided in \citet*{koszegi2014behavioral}.

%\todo[inline]{Need to explain why there is a need to move to continuous-time: is there a general method outlined in these papers? Are they missing something by not having a dynamic model?}
\medskip

In continuous-time, where the dynamic models are sometimes more tractable and the solutions enjoy better interpretability, the literature becomes rather scarce. Models dealing with a pre-committed agent have been considered in \citet*{li2018solvable}, in which a non-constant exponential discount factor is the source of time-inconsistency, and \citet*{djehiche2015principal}, where the agent is allowed to have mean-variance utility functions. The case of a sophisticated agent was considered in \citet*{li2015optimal}, \citet*{liu2018dynamic}, \citet*{liu2019optimal} and \citet*{wang2019optimal} in the case of hyperbolic discounting. However, the time-inconsistency is restricted in the sense that it manifests only at discrete random times that are exponentially distributed. Lastly, \citet*{cetemen2021renegotiation} considers a Markovian continuous-time contracting problem and dynamic inconsistency arising from non-exponential discounting. The authors' examples are limited to the case of a principal having time-inconsistent preferences, and the agent having standard time-consistent preferences. Altogether, a thorough analysis of the general non-Markovian continuous-time contracting problem between a standard utility maximiser principal and a sophisticated time-inconsistent agent is still missing in the literature. This is because, in our opinion, the crux of the problem lies in identifying a proper description of the problem of the principal. In the case of a classic time-consistent agent and a time-inconsistent principal, following {\rm \cite{cvitanic2015dynamic}}, one expects the problem of the principal to boil down to a non-Markovian time-inconsistent control problem with an additional state variable. As studied in {\rm \cite{hernandez2020me}}, these problems are characterised by an infinite family of {\rm BSDEs}, equivalent to a so-called type-I extended BSVIE \cite{hernandez2020unified}. As such, we expect that the problem considered in this document will open the door to a complete analysis of the problem in which both the principal and the agent are time-inconsistent.

%\todo[inline]{You need to explain why this is limited and why our approach is new and brings something to the table. That's what I mean by comparing what we do with the literature}

\subsection*{Our results}

Our problem is cast in the context of a standard utility maximiser principal and an agent with time-inconsistent preferences. Indeed, the agent's reward is given by the value of a so-called backward stochastic Volterra integral equation. This choice of preferences for the agent allows us to cover classic separable and non-separable utilities simultaneously. 
As is standard in the literature, we consider the weak formulation of the problem. The state process $X$ is fixed, and the agent's actions influence the drift of $X$ through its distribution over the interval $[0,T]$. The principal chooses a contract, \emph{i.e.} \textcolor{black}{a process and} a random variable adapted to the filtration generated by the path $X_{\cdot \wedge T}$ of the state process, which specifies \textcolor{black}{the continuous payments,} the terminal payment, and satisfies the agent's participation constraint at time $t=0$. As mentioned above, our approach is inspired by that of \cite{cvitanic2015dynamic} and the recent results for non-Markovian time-inconsistent control problems from a game-theoretic point of view of \cite{hernandez2020me}. Indeed, \cite{hernandez2020me} established an extended dynamic programming principle for the agent's value process associated with any equilibrium action. In turn, this result was used to establish a direct link between the agent's problem and an infinite family of BSDEs. Following \cite{hernandez2020unified}, such a system is actually equivalent to a so-called type-I extended BSVIE. 

\medskip
At this point, we notice the first stark difference between the classic time-consistent case and ours: the problem of the agent is, in general, linked to the solution of an infinite family of equations, namely the BSVIE, as opposed to one, a BSDE. Nevertheless, the agent's preferences elucidate a connection at the terminal time $t=T$ between the terminal values of the BSVIE and the terminal payment offered by any admissible contract. This is the crucial insight in order to restrict our attention from the family of admissible contracts to a carefully tailored family of contracts for which the agent's value process allows a dynamic programming representation capturing the Volterra nature of the agent's reward. Extrapolating from the time-consistent case, the restricted family of contracts is defined in terms of a family of first-order sensitivities of the agent's value process to the output. For this family of contracts, we show that the principal identifies the equilibrium action for the agent as the maximisers of the associated Hamiltonian. Nevertheless, echoing the agent's time-inconsistent preferences, the resulting principal's problem is, in general, far from being a standard stochastic control problem.\medskip

Our main contribution, namely \Cref{thm:repcontractgeneral}, consists in a characterisation of the moral hazard problem faced by the principal and a sophisticated time-inconsistent agent. 
	In particular, it shows that under relatively mild technical conditions on the data of the problem, the supremum of the principal's expected utility over the restricted family of contracts is equal to the supremum over all feasible contracts. 
	Nevertheless, this characterisation yields, as far as we know, a novel class of control problems. 
	These problems involve the control of a forward Volterra equation via Volterra-type controls, and stochastic target constraints. 
	One of the novel features of our result is that the dynamics of this process involves the diagonal value of both the forward Volterra process and the Volterra control, see \Cref{def:indexedcontracts}. 
	In addition, the stochastic target constraint arises due to the time-inconsistent preferences of the agent, see \eqref{eq:constraintgen} and \Cref{rmk:targetconstraint}.\medskip

\textcolor{black}{Despite the inherent challenges of this class of problems, we study the solution to moral hazard problem under three different specifications of utility functions for both the agent and the principal. 
	For instance, for non-separable reward functionals, we find that if both the agent and the principal have exponential utilities functions and the agent's reward is given by the discounted value of his utility, the problem reduces to a standard control problem, see \Cref{sec:ra1} and \Cref{proposition:ra1principalreduction}. 
	This is a feature that we also see in the risk-sharing (or first-best) contracting examples between the principal and a time-inconsistent agent that we present in \Cref{sec:firstbestexamples}.
	The second example considers a risk-neutral principal and a risk-neutral agent with separable reward functional. In this case, our analysis shows that it is possible to reduce the complexity of the problem. 
	Indeed, we can exploit the structure of the problem to formulate an \emph{ansatz} to the principal's problem, for which we present a result in the spirit of a verification theorem, see \Cref{sec:separable} and \Cref{prop:sol2ndbest0}.
	In the last example, we go back to the first setting, but in this case the agent's (exponential) utility is taken on the discounted income. This simple modification highlights the intrinsic difficulties of the general case, and we are able to solve the problem of the principal for a class of contracts smaller than the one prescribed by the restricted family of contracts in \Cref{thm:repcontractgeneral}, see \Cref{sec:ra2} and \Cref{prop:sol2Bra2}.
	The general case remains the subject of future research.}

%\todo[inline]{I think it is interesting to say a bit more: with discounted exponential utility, the problem is easy and degenerates to standard control for instance. Basically describe a bit more the results we get under the scenarii you mention above}
\medskip
Regarding the qualitative implications of our results we can mention the following:
\begin{enumerate}[label=$(\roman*)$, ref=.$(\roman*)$,wide,  labelindent=0pt]
\item from a methodological point of view, unlike in the time-consistent case, the solution to the moral hazard problem does not reduce, in general, to a standard stochastic control problem. Nevertheless, the solution to the risk-sharing problem between a utility maximiser principal and a time-inconsistent sophisticated agent does, see \Cref{sec:firstbestexamples}. This suggest a dire difference between the first-best and second-best problems as soon as the agent is allowed to have time-inconsistent preferences;

\item a second takeaway from our analysis is associated with the so-called optimality of linear contracts. These are contracts consisting of a constant part and a term proportional to the terminal value of the state process as in the seminal work of \cite{holmstrom1987aggregation}. This was also the conclusion of \citet*{carroll2015robustness} in a two-stage time-consistent model in which the principal demands robustness, in the sense of evaluating admissible contracts by their worst-case performance, over unknown actions the agent might take. Similar results we obtained by \cite{sung2005optimal,sung2021optimal} and \cite{mastrolia2015moral} in the continuous-time setting. Moreover, the results in \citet*{abijaber2021linear} show that the optimal contract remains linear when the output is driven by a Gaussian Volterra process (instead of Brownian motion). We study two examples that can be regarded as (time-inconsistent) variations of \cite{holmstrom1987aggregation}, which we refer to as discounted utility, see \Cref{sec:ra1}, and utility of discounted income, see \Cref{sec:ra2}. In the former case, by virtue of the simplicity of the source of time-inconsistency, we find that the optimal contract is linear. In the latter case, we find that the optimal contract is no longer linear unless there is no discounting (as in \cite{holmstrom1987aggregation}). Our point here is that slight deviations of the model in \cite{holmstrom1987aggregation} seem to challenge the virtues attributed to linear contracts, and this suggests that they would typically cease to be optimal in general for time-inconsistent agents;

\item lastly, we comment on the non-Markovian nature of the optimal contract. It is known that, beyond the realm of the model in \cite{holmstrom1987aggregation}, the optimal contract in the time-consistent scenario is, in general, non-Markovian in the state process $X$, see \cite{cvitanic2015dynamic}. Indeed, we find the same result, see \Cref{prop:sol2ndbest0}, in the case of an agent with separable time-inconsistent preferences, see \Cref{sec:separable}. As such, we believe this is a manifestation of the agent's time-inconsistent preferences.
\end{enumerate}
%\todo[inline]{I do not agree with this completely: the point about getting linear contracts is interesting, but my point of view is different: granted when agent discounts future utility, this happens, so it's in line with the literature on robustness of linear contracts. However, we SHOW that if we perturb HM by having discounting appear inside the utility, this is no longer true. I think this gives more insight and shows that robustness to time-inconsistency is more subtle.}

\textcolor{black}{Let us illustrate some of our results, see \Cref{sec:probstate} for precise definitions. Fix a time horizon $T>0$ and consider a sophisticated time-inconsistent agent with risk-neutral preferences. This is, if the agent enters into a terminal payment contract $\xi$ with the principal, the agent seeks an equilibrium strategy $\alpha^\star\in \Ec(\Cc)$ according to}
\begin{align*}
{\rm V}^{\rm A}_t(\xi,\alpha)=\E^{\P^\alpha}\bigg[f(T-t)\xi-\frac{1}2\int_t^Tf(s-t)\alpha_s^2 \mathrm{d}s\bigg|\mathcal F_t\bigg], \text{ and, } X_t=x_0+\int_0^t \alpha_r \d r+B^\alpha_t,\;  \P^\alpha\text{\rm --a.s.}
\end{align*}

\textcolor{black}{The principal is a risk-neutral utility maximiser. This is, among all the admissible contracts $\xi\in \Cc$ of $\Fc^X_T$-measurable random variables satisfying the agent's participation constraint ${\rm V}^{\rm A}_0(\xi,\alpha) \geq R_0$, she maximises}
\[
{\rm V^P_{FB}}=\sup_{\xi \in\Cc}\E^{\P^{\alpha^\star}}\big[ X_T-\xi \big].\]

The agent is time-inconsistent in light of the general discounting function $f$ appearing in his reward. We summarise and illustrate some typical discounting models next. In all the illustrations in this section we take $T=50$.
%\vspace{-1em}
\begin{figure}[h]
    \centering
    \begin{minipage}[t]{.5\textwidth}
        \centering
%\bgroup
\def\arraystretch{2}
\begin{tabular}{|c|c|}
\hline
 \textbf{$f(t)$}                 & IDR \\ \hline
 $f_{\e}(t):=\e^{-\gamma t}  $ &      $\gamma$                                                                                                                       \\ \hline
                $f_{\rm h} (t):=(1+\alpha t)^{-\frac{\gamma}{\alpha}}  $                   &              $\frac{\gamma}{1+\alpha t}  $                                                                                                             \\ \hline
          $f_{\rm q} (t):= (1-\beta) \e^{-t(\lambda +\gamma)}+ \beta   \e^{-t\gamma}  $         &         $\gamma + \frac{\lambda (1-\beta) } {(1-\beta) + \beta   \e^{t\lambda}  }     $                                                                                                                   \\ \hline        
\end{tabular}
%\vspace{0.5em}
%\egroup
%\caption{caption 1}
        \label{fig:sub1}
    \end{minipage}%
    \begin{minipage}[t]{.45\textwidth}
        \centering
        \includegraphics[width=0.85\textwidth, trim = 0mm 47mm 0mm 10mm]{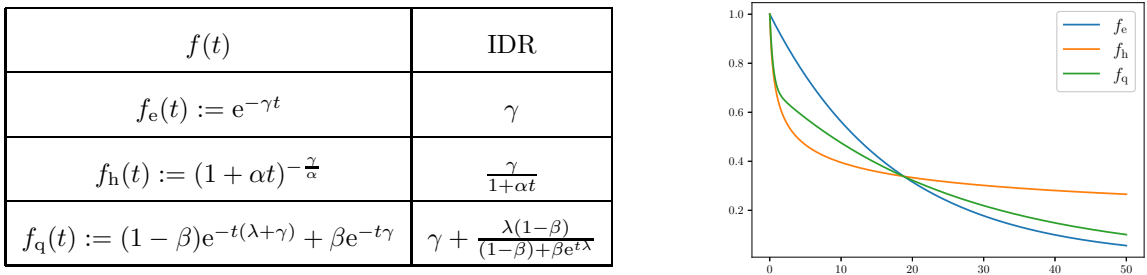}
       %\caption{caption 2}
       % \label{fig:sub2}
    \end{minipage}
\caption{\small $(\gamma,\alpha,\lambda,\beta)\in  (0,\infty)^3\times [0,1]$. On right graph, $\gamma_{\rm e}= 0.0576,\alpha=4, \gamma_{\rm h}=1, \beta=0.7,\lambda = 2.197, \gamma_{\rm q} = 0.0387$.}
\end{figure}

\textcolor{black}{The most widely used discounting model in classic economics is the exponential discount function $f_\e$. This model captures the empirical evidence that future utils are worth less than current utils, yet the instantaneous discount rate (IDR), given by $-f^\prime (t)/f(t) $, is constant over time. To be able to accommodate the fact that consumers have both a short-run preference for instantaneous gratification and a long-run preference to act patiently, \citet*{ainslei1992pico} introduced the hyperbolic discounting model. Hyperbolic discounting generates the so-called self-control problem, \emph{e.g.} it declines at a faster rate in the short run than in the long run, depending on the value of $\alpha$, whereas $\gamma$ plays the role of baseline discounting rate. This qualitative property is even more evident in the so-called quasi-hyperbolic discounting model $f_{\rm q}$ introduced by \cite{laibson1997golden}. This model exhibits the short-run impatience of the hyperbolic model, but for long time horizons, its instantaneous discount rate resembles that of the exponential model. For this discounting, $\beta$ measures the value the agent gives to future periods, whereas $\lambda$ measures the agent's additional valuation for present/current periods. Once again, $\gamma$ is the baseline discounting rate. We can see these observations in both the IDR column of the table to the left and the plot of the three models on the right.}\medskip

\textcolor{black}{Let us recall that the problem faced by an agent seeking to maximise his reward is time-consistent if and only if he discounts future utils/rewards with the exponential model. Thus, taking $f_{\rm h}$ or $f_{\rm q}$ in the above formulation leads to genuine time-inconsistent problems for which sophistication would have an inherent impact in the equilibrium actions followed by the agent under the optimal contract. In addition, it is easy to see that:}
\begin{itemize}
\item $f_{\rm h}(t)\longrightarrow f_{\rm e}(t)$, as $\alpha\longrightarrow 0$,
\item $f_{\rm q}(t)\longrightarrow f_{\rm e}(t)$, as either $\beta \longrightarrow 1$ or $\lambda\longrightarrow 0$.
\end{itemize}

\textcolor{black}{For any sufficiently regular discounting model, including the ones just discussed, we find in \Cref{prop:sol2ndbest0} that the associated optimal contract is given by}
\[
\xi^\star=C(R_0)+\int_0^T\frac{z^\star(t)}{f(T-t)} \mathrm{d} X_t,
\]
\textcolor{black}{where $C(R_0)$ is a constant depending on the agent's reservation utility. The second term, however, reflects the non-Markovian nature of the terminal payment mentioned in $(iii)$ above, and it is inherently related to the non-exponential discounting structure. Conversely, whenever $f=f_{\rm e}$ the term $\frac{z^\star(t)}{f(T-t)}$ becomes constant, leading to the well-known optimal contract that is linear in the terminal value of the output process.}\medskip

We now turn our attention to the agent's equilibrium action under the optimal contract. We illustrate this in Figure 2 below. The three columns study the above model with $f_{\rm h}$, $f_{\rm q}$ and $f_{\rm q}$ for different values of $\alpha$, $\beta$ and $\lambda$, respectively. In light of the connection with the exponential discounting mentioned above, all columns include $f_{\rm e}(t)$ in blue, which serves as a true time-consistent baseline comparison model. The first row presents the value of the discounting functions and confirms the limits pictorially. The second row presents the associated instant discounting rates, IDR, and the last row does so for the equilibrium actions. Let us first note that the shape of the equilibrium action under exponential discounting is intuitively expected. It is convex and increasing with a shape that is inversely proportional to the exponential discounting term, this reflects that the discounted optimal effort should remain constant under the optimal strategy.
\medskip

Let us first look at the case of hyperbolic discounting agents. We see that as $\alpha$ increases, along the equilibrium effort, the sophisticated agent increases its level of effort during the initial stages. In addition, the rate at which the equilibrium effort changes over time (convexity/concavity) is positive for small values of $\alpha$ and negative for large values, \emph{e.g.} $\alpha_4$ and $\alpha_2$, respectively. This means that as the time-inconsistency intensifies, sophistication causes the agent to exert larger levels of effort at the initial stages of the game. This reflects how sophistication can help overcome procrastination.
\medskip

We now look at the quasi-hyperbolic agent in the centre and right column. As $\beta$ decreases, the agent gives less weight to the future periods and values more present over future gratification. In other words, the time-inconsistency intensifies, and the sophisticated agent decides to postpone some effort to the future period. In this scenario, despite sophistication, the agent cannot overcome procrastination. Lastly, as $\lambda$ decreases, the agent weighs less the present period, where his time-inconsistency is more acute so that the inconsistency lessens. We nevertheless find that even though the initial effort decreases, and procrastination dominates for these decreasing values of $\lambda$, namely $\lambda_1$ and $\lambda_2$, as the agent valuation of the present gets significantly small, $\lambda_3$ and $\lambda_4$ respectively, his effort increases overcoming procrastination and reaching the time-consistent level of effort. We believe that the different behaviours on the equilibrium effort for the quasi-hyperbolic discounting can be reconciled when looking at the associated IDR plots. For $t$ fixed, the IDR is monotonically decreasing in $\beta$, whereas there are values of $t$ for which the IDR oscillates when $\lambda$ decreases. 
\medskip

We leave the comprehensive study of these behaviours as the subject of future research. In particular, it would be interesting to study the extension of our results to the so-called instant gratification model in \citet*{harris2013instantaneous}, which implements $f_{\rm ig}$ given by $f_{\rm q}\longrightarrow f_{\rm ig}$, as $\lambda\longrightarrow \infty$, and which are beyond the scope of this document.

\begin{figure}[H]
    \centering
    \subfloat{
        %\label{ref_label1}
        \includegraphics[width=0.30\textwidth, trim = 10mm 5mm 9mm 10mm]{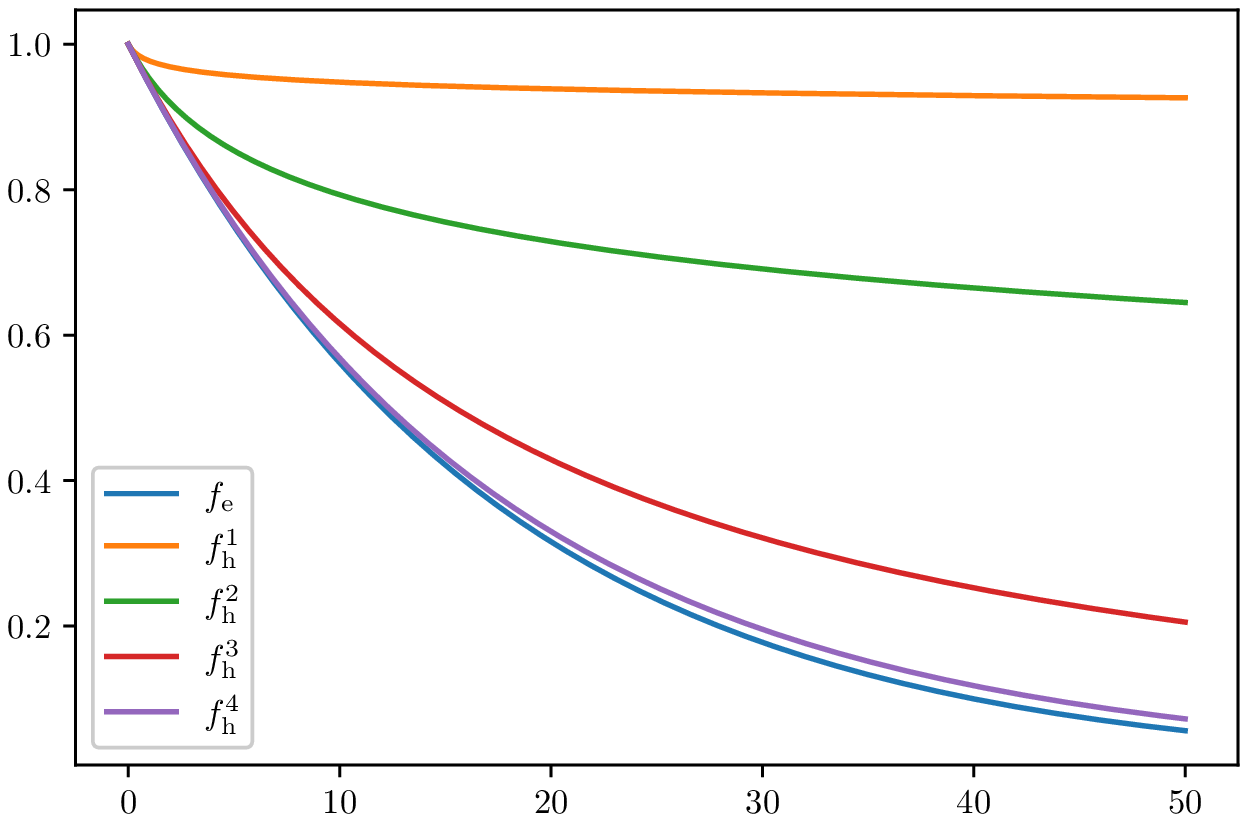}
    }
    \subfloat{
        %\label{ref_label2}
        \includegraphics[width=0.3\textwidth, trim = 10mm 5mm 9mm 10mm]{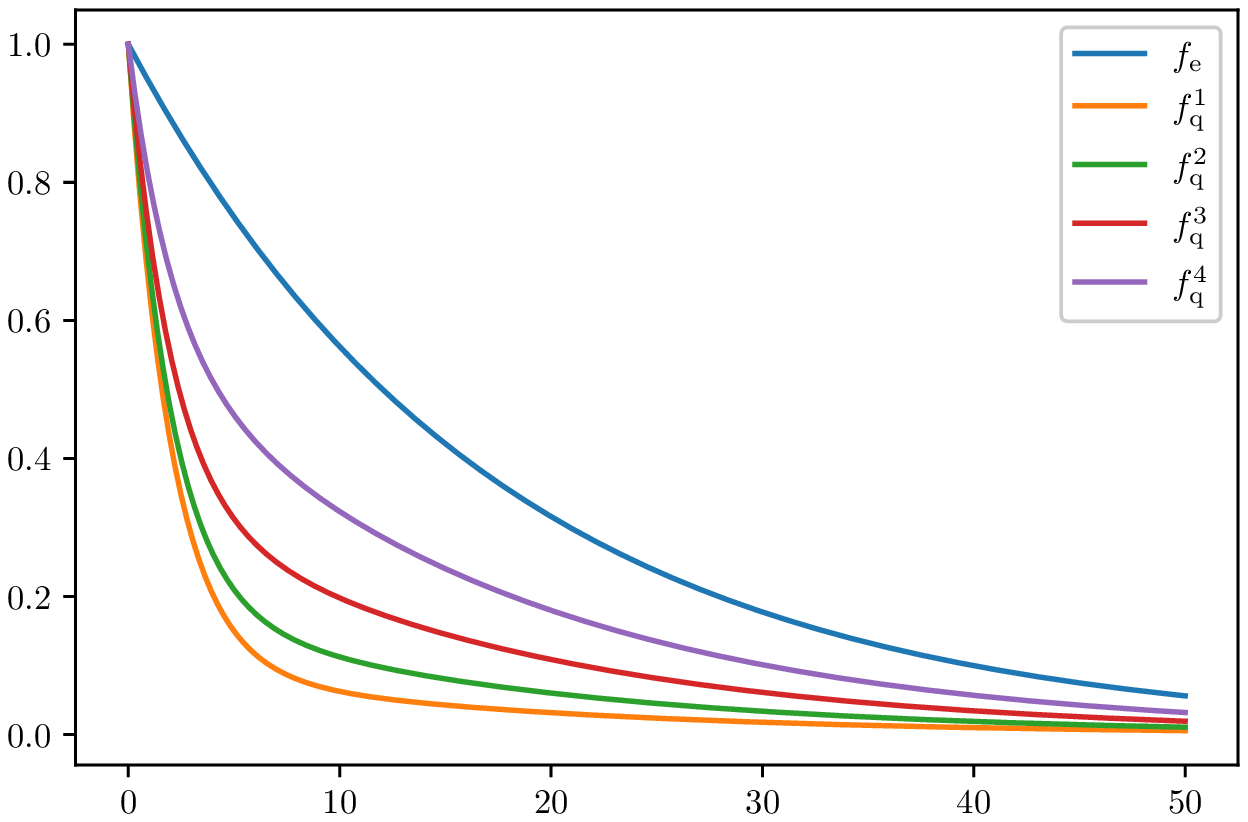}
    }
     \subfloat{
        %\label{ref_label2}
        \includegraphics[width=0.30\textwidth, trim = 10mm 5mm 10mm 10mm]{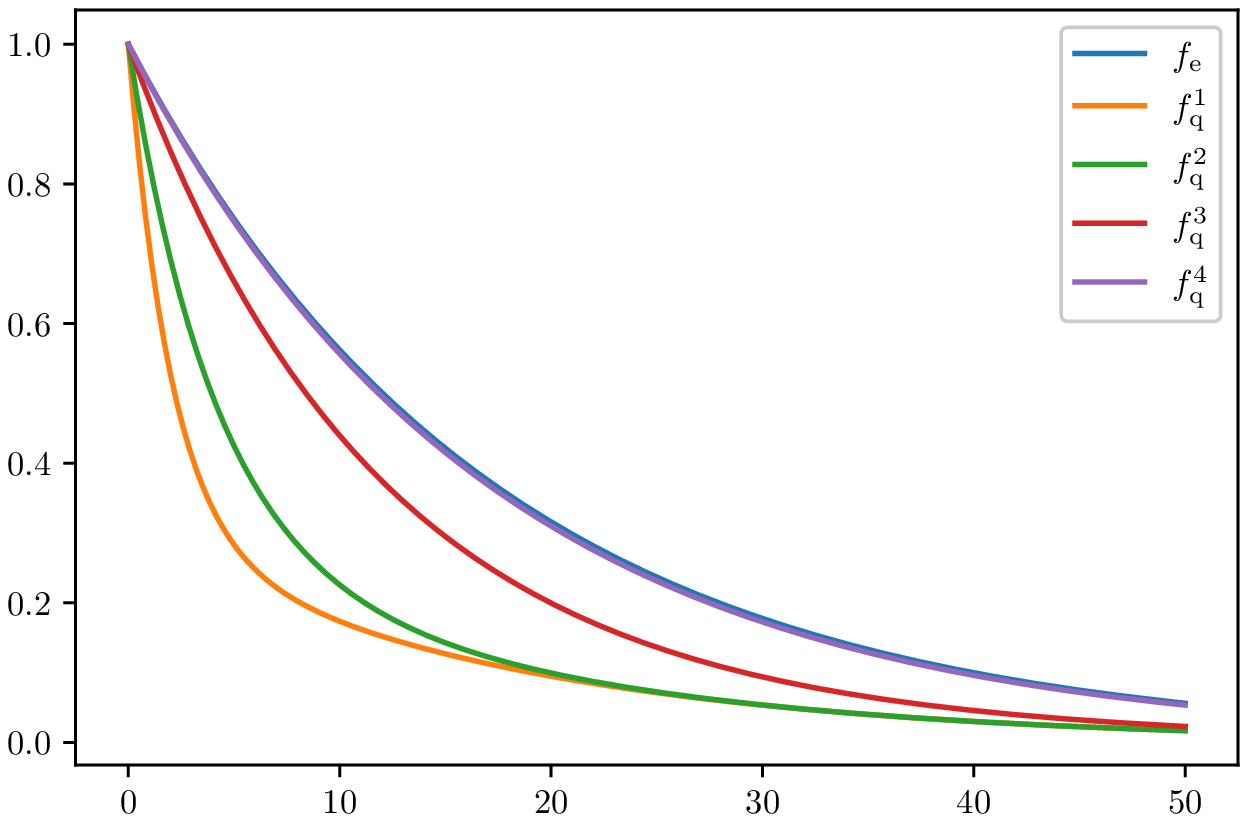}
    }
    \hfill
    \subfloat{
        %\label{ref_label1}
        \includegraphics[width=0.30\textwidth, trim = 10mm 5mm 9mm 10mm]{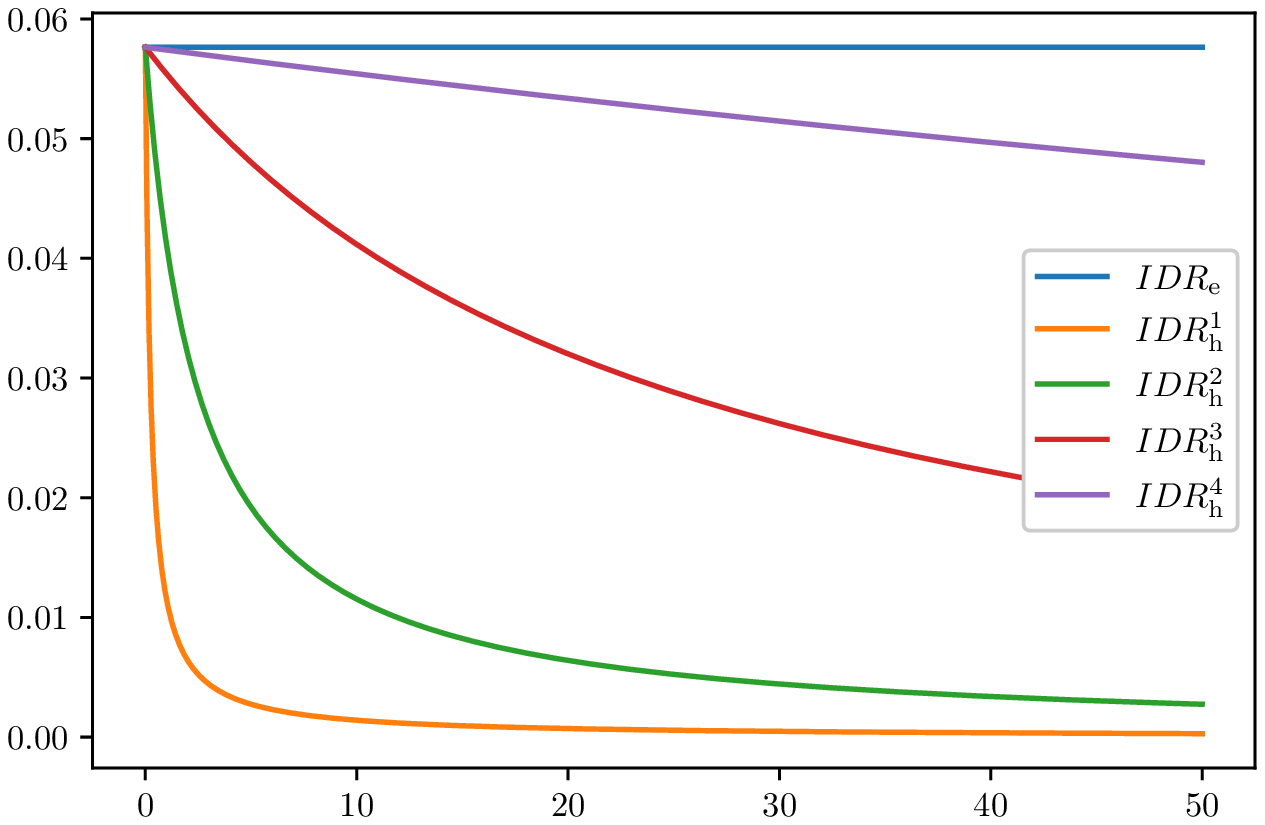}
    }
    \subfloat{
       % \label{ref_label2}
        \includegraphics[width=0.30\textwidth, trim = 10mm 5mm 9mm 10mm]{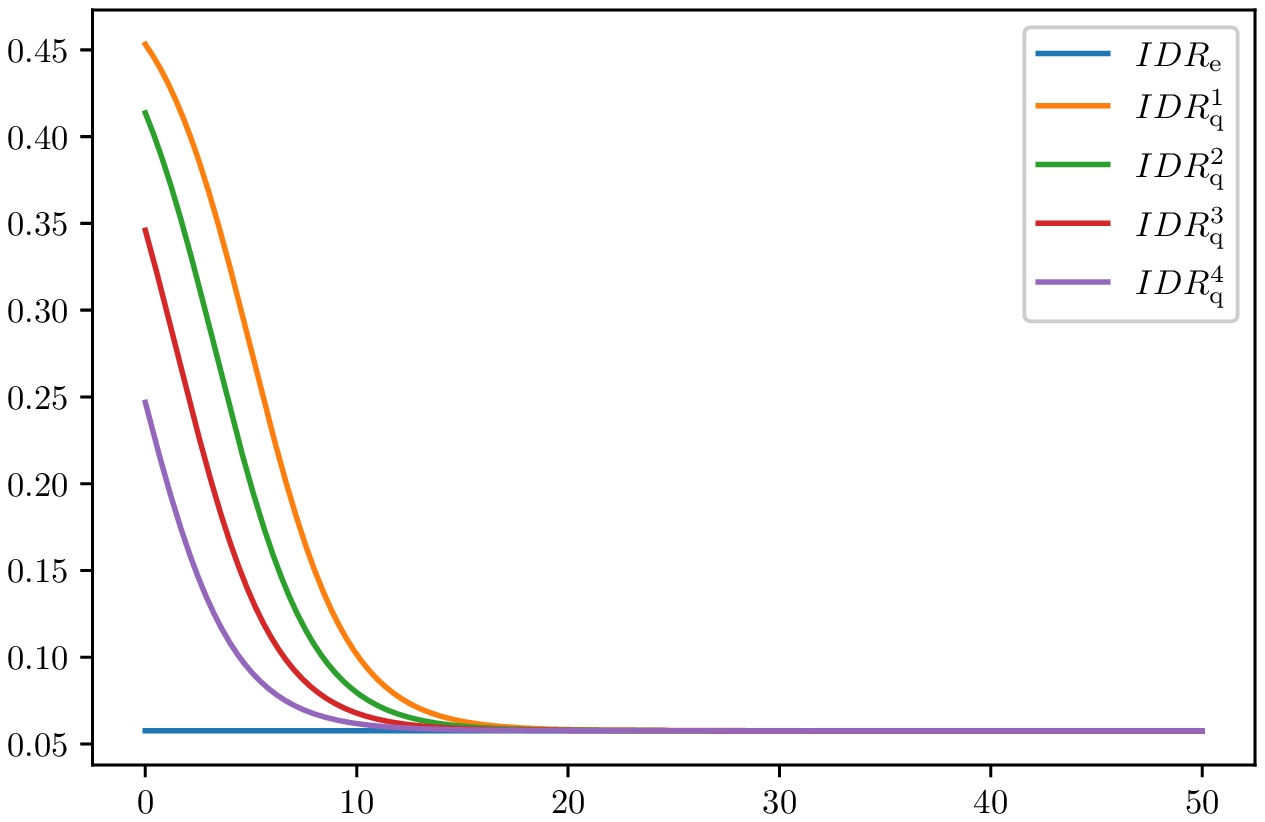}
    }
     \subfloat{
        %\label{ref_label2}
        \includegraphics[width=0.30\textwidth, trim = 10mm 5mm 10mm 10mm]{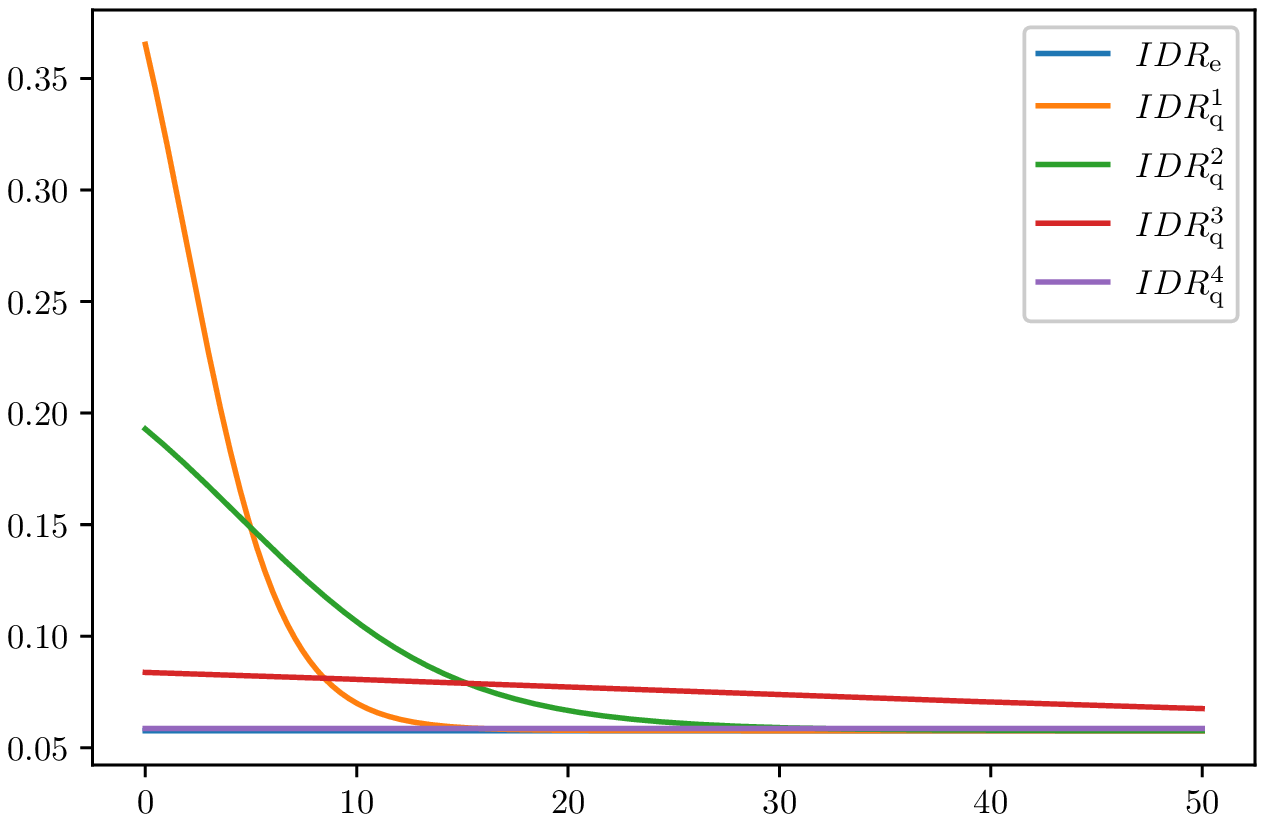}
    }
     \hfill
      \end{figure}%
\begin{figure}[H]\ContinuedFloat
    \subfloat{
        \label{ref_label1}
        \includegraphics[width=0.30\textwidth, trim = 14mm 11mm 6mm 10mm]{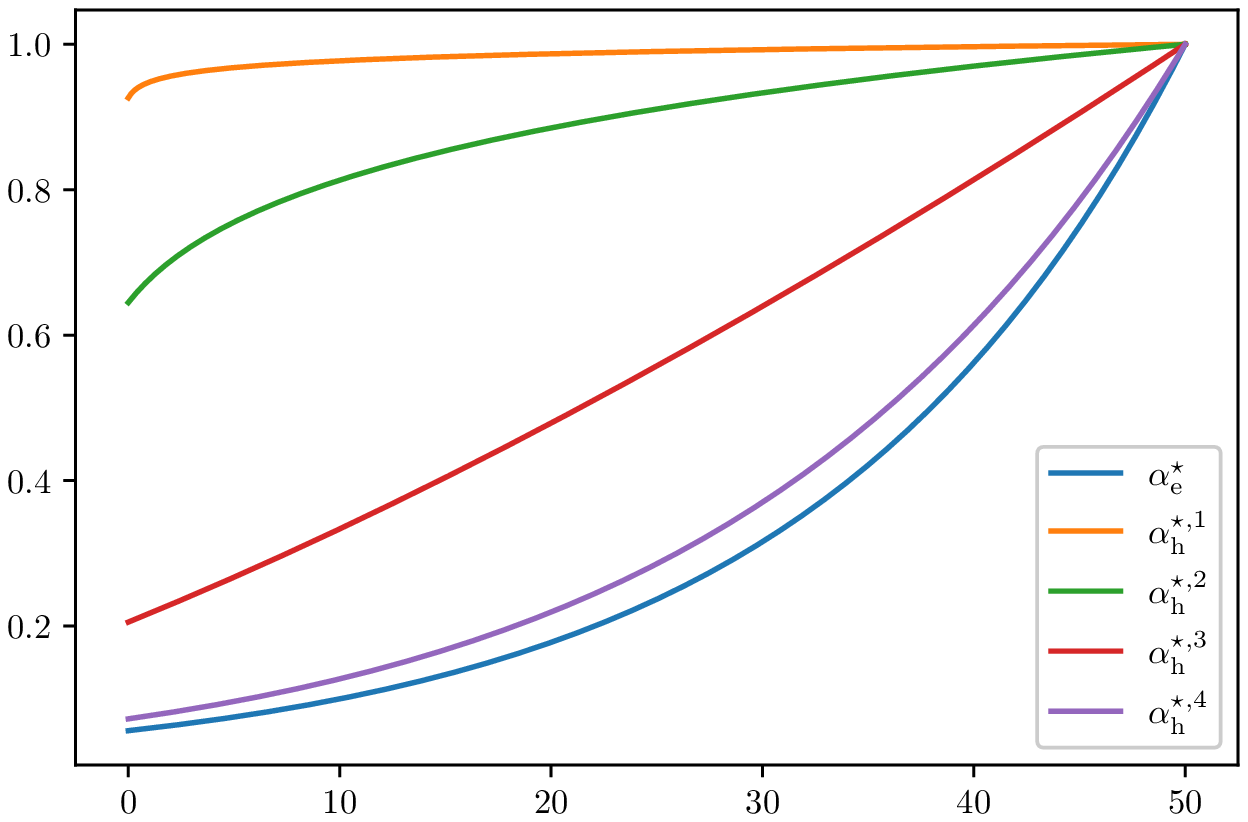}
    }
    \subfloat{
        %\label{ref_label2}
        \includegraphics[width=0.30\textwidth, trim = 45mm 100mm 37mm 98mm]{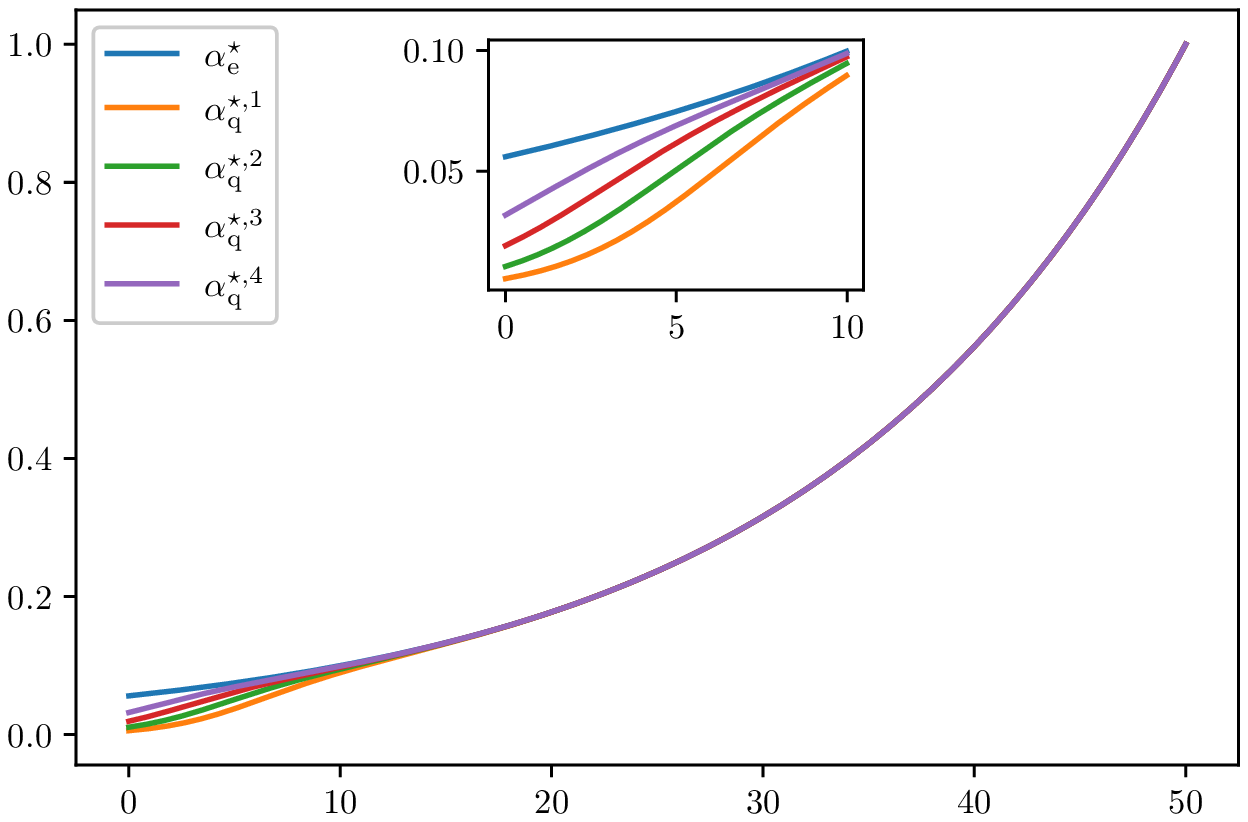}
    }
    \subfloat{
        %\label{ref_label2}
        \includegraphics[width=0.30\textwidth, trim = 45mm 100mm 37mm 98mm]{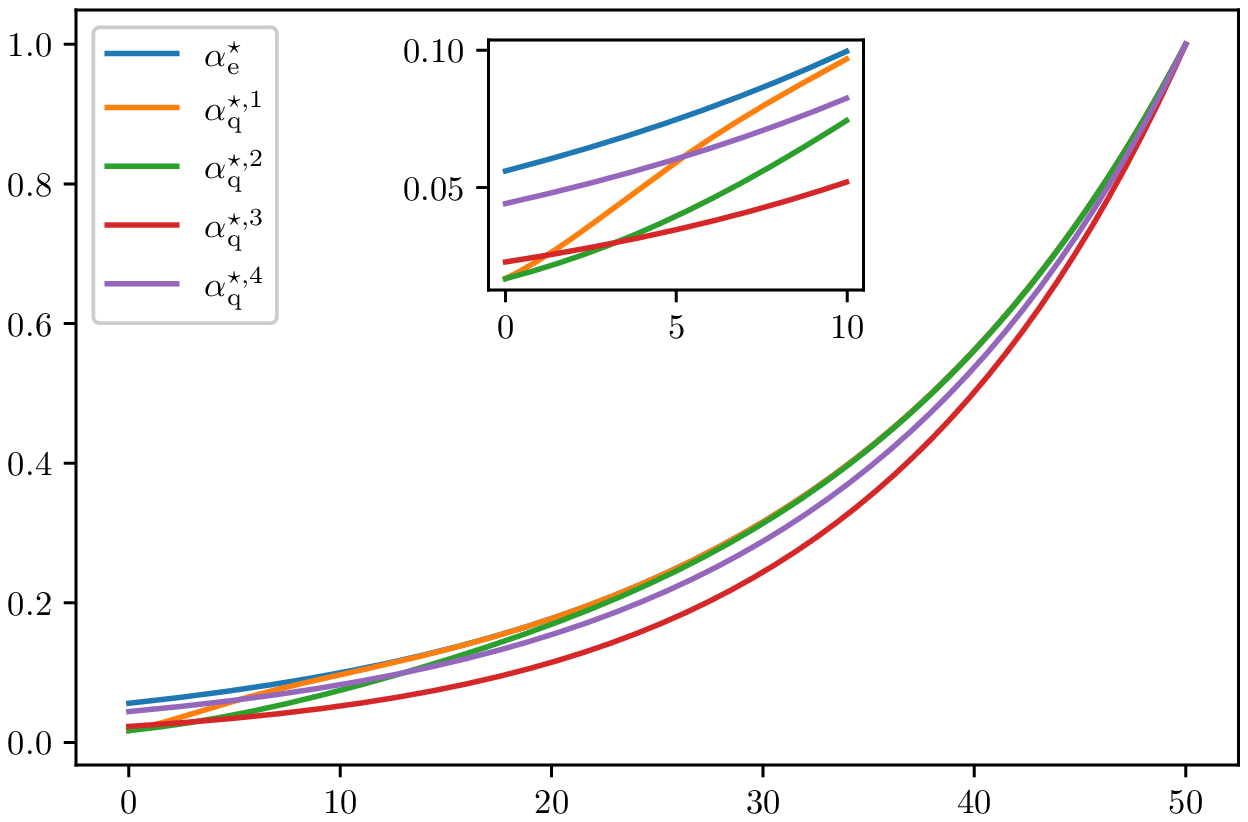}
    }
    \caption{\small First row: discounting models. Second row: IDR. Third row: agent's equilibrium action under optimal contract. $ \gamma_{\rm e}=\gamma_{\rm h}=\gamma_{\rm q}=0.0575$. Left: $(\alpha_1, \alpha_2,\alpha_3,\alpha_4) =(4,0.4,0.04,0.004)$. Center: $\lambda=0.439$, $(\beta_1, \beta_2, \beta_3, \beta_4)=(0.1, 0.19,0.343, 0.569)$. Right: $\beta=0.3$, $(\lambda_1,\lambda_2,\lambda_3, \lambda_4)=(0.439,0.1927,0.0371, 0.0013)$.}
    \label{ref_label_overall}
\end{figure}

The rest of this article is organised as follows. \Cref{sec:probstate} takes care of the formulation of the problem, and \Cref{sec:firstbestexamples} presents the solution to the first-best problem under three different specifications of preferences. The common feature in these examples is that the problem boils down to solving a standard stochastic control problem. \Cref{sec:secondbestproblem} introduces our general approach to the second-best problem and presents the proof of \Cref{thm:repcontractgeneral}. \Cref{sec:sbexamples} is devoted to the analyses of three examples under different specifications of time-inconsistent preferences for the agent. Lastly, we include an Appendix section collecting some new results for time-inconsistent control problems with BSVIE rewards and other technical results.

\medskip
{\small\textbf{Notations:} $\R_+$ denotes the set of non-negative real numbers. Let $(E,\|\cdot \|_E)$ be an arbitrary finite dimensional normed space. Given a positive integer $p$ and a non-negative integer $q$, $\Cc_{q}(E,\R^p)$ will denote the space of functions from $E$ to $\R^{p}$ which are at least $q$ times continuously differentiable. In the case $q=0$, of continuous functions, we drop the dependence on $q$ and write $\Cc([0,T],\R^{p})$. 
%For positive integers $m$ and $n$, we denote by $\Mc_{m,n}(\R)$ the space of $m\times n$ matrices with real entries. 
By %$0_{m,n}$ and 
$\text{I}_n$ we denote %the $m\times n$ matrix of zeros and 
the identity matrix of $\R^{n\times n}$. $\S_n^+(\R)$ denotes the set of $n\times n$ symmetric positive semi-definite matrices. ${\rm Tr} [M]$ denotes the trace of a matrix $M\in \R^{n\times n}$.\medskip

Let $(E,\|\cdot\|_E)$ be an arbitrary finite dimensional normed space.  For a $\sigma$-algebra $\Fc$, $\Pc_{\text{meas}}(E,\Fc)$ denotes the space of $\Fc$-measurable $E$-value functions. $\Pc^2_{\text{meas}}(E,\Fc)\big)$ denotes the space of $U=(U_t^s)_{(s,t) \in \R_+^2 }$ such that $(\R_+^2\times \Omega, \Bc(\R_+^2)\otimes \Fc)  \longrightarrow (\Bc(E),E): s\longmapsto U^s$ measurable. For a filtration $\F:=(\Fc_t)_{t\geq 0}$ on $(\Omega, \Fc)$, $\Pc_{\rm pred}(E,\F)$ (resp. $\Pc_{\rm prog}(E,\F)$, $\Pc_{\rm opt}(E,\F)$, $\Pc_{\rm meas}(E,\F)$) denotes the set of $E$-valued, $\F$-predictable processes (resp. $\F$--progressively measurable processes, $\F$-optional processes, $\F$-adapted and measurable).

 \section{Problem statement}\label{sec:probstate}
We fix two positive integers $n$ and $d$, which represent respectively the dimension of the process controlled by the agent, and the dimension of the Brownian motion driving this controlled process. We fix a time horizon $T>0$, and consider the canonical space $\Omega:=\Cc([0,T],\R^{n})$, with canonical process $X$, and whose generic elements we denote $x$. We reserve the notation ${\rm x}$ and ${\rm \bf x}$ to denote $\R$-valued variables.
%We consider the canonical space $\Omega:=\Cc(\R_+,\R^{n})$,with canonical process $X$, and whose generic elements we denote $x$ equipped with the norm $\|x \|_\infty:=\sum_{n\geq 0} 2^{-n} \big( \sup_{t\in [0,n]} | x_t| \wedge 1\big)$. We reserve the notation ${\rm x}$ and ${\rm \bf x}$ to denote $\R$-valued variables.

\medskip

We let $\Fc$ be the Borel $\sigma$-algebra on $\Omega$ (for the topology of uniform convergence), and we denote by $\F^X:=(\Fc^X_t)_{t\in[0,T]}$ the natural filtration of $X$. We let $A$ be a compact subspace of a finite-dimensional Euclidean space (typically $A$ is a subset of $\R^k$ for some positive integer $k$), where the controls will take values. }

\subsection{Controlled state equation}\label{sec:controlledstateeq}

We fix a bounded Borel measurable map $\sigma:[0,T]\times\Omega \longrightarrow \R^{n\times d}$, and an initial condition $x_0\in\R^n$, and assume that there is a unique solution, denoted by $\P$, to the martingale problem for which $X$ is an $(\F^X,\P)$--local martingale, such that $X_0=x_0$ with probability $1$, and $\mathrm{d}\langle X\rangle_t=\sigma_t(X_{\cdot\wedge t})\sigma^\top_t(X_{\cdot\wedge t})\mathrm{d}t$. Enlarging the original probability space if necessary (see \citet*[Theorem 4.5.2]{stroock1997multidimensional}), we can find an $\R^{d}$-valued Brownian motion $B$ such that
\[
X_t=x_0+\int_0^t\sigma_s(X_{\cdot\wedge r})\mathrm{d}B_r,\; t\in[0,T].
\]

We now let $\F:=(\Fc_t)_{t\in[0,T]}$ be the $\P$--augmentation of $\F^X$ which we assume is right-continuous. We recall that uniqueness of the solution to the martingale problem implies that the predictable martingale representation property holds for $(\F,\P)$-martingales, which can be represented as stochastic integrals with respect to $X$ (see \citet*[Theorem III.4.29]{jacod2003limit}). \textcolor{black}{We also mention that the right-continuity of $\F$ guarantees that $(\F,\P )$ satisfies the Blumenthal zero--one law and consequently all $\Fc_0$--measurable random variables are deterministic. Let us note that these assumptions are standard in the existing literature on the continuous-time principal--agent problem.}
\medskip

% \todo{Comments?}

We can then introduce our drift functional $b:[0,T]\times\Omega \times A\longrightarrow \R^d$, which is assumed to be Borel-measurable with respect to all its arguments. Let us recall that for any $A$-valued, $\F$-predictable process $\alpha$ such that
\begin{equation}\label{eq:integalpha}
\E^{\P}\bigg[\exp\bigg(\int_0^Tb_r\big(X_{\cdot\wedge r},\alpha_r\big)\cdot\mathrm{d}B_r-\frac12\int_0^T\big\|b_s\big(X_{\cdot\wedge s},\alpha_r\big)\big\|^2\mathrm{d}r\bigg)\bigg]<\infty,
\end{equation} 
we can define the probability measure $\P^\alpha$ on $(\Omega,\Fc_T)$, whose density with respect to $\P$ is given by
\[
\frac{\mathrm{d}\P^\alpha}{\mathrm{d}\P}:=\exp\bigg(\int_0^Tb_r\big(X_{\cdot\wedge r},\alpha_r\big)\cdot\mathrm{d}B_r-\frac12\int_0^T\big\|b_r\big(X_{\cdot\wedge r},\alpha_r\big)\big\|^2\d r\bigg).
\]
Moreover, by Girsanov's theorem, the process $B^\alpha:=B-\int_0^\cdot b_r(X_{\cdot\wedge r},\alpha_r)\mathrm{d}r$ is an $\R^d$-valued, $(\F,\P^\alpha)$--Brownian motion and we have
\[
X_t=x_0+\int_0^t\sigma_r(X_{\cdot\wedge r})b_r(X_{\cdot\wedge r},\alpha_r)\d r+\int_0^t\sigma_r(X_{\cdot\wedge r})\mathrm{d}B^\alpha_r,\; t\in[0,T],\; \P\text{\rm --a.s.}
\]

Let us emphasise that we are working under the so-called weak formulation of the problem. 
	This means that the state process $X$ is fixed and, in contrast to the typical strong formulation, the Brownian motion, and the probability measure are not fixed. 
	Indeed, the choice of $\alpha$ corresponds to the choice of probability measure $\P^\alpha$ and thus impacts the distribution of process $X$.

\subsection{The agent's problem}\label{sec:agentprobstatement}

We aim to cover various specifications of time-inconsistent utility functions for the agent. 
	To motivate our formulation, let us start with an informal discussion on the typical nature of the reward functionals assigned to the agent in contract theory. 
	A contract $\Cc$ consists of a tuple $(\pi,\xi)$, where $\pi$ belongs the set of $\F$-predictable processes, and $\xi$ is a $\Fc_T$-measurable random variable. 
	%\todo[inline]{You cannot assume $\xi\geq 0$ here, this introduces limited liability and will screw things up later. I would also not give a sign to $\pi$.}
	At the intuitive level, a contract consists of a flow of continuous payments $\pi:=(\pi_t)_{t\in [0,T)}$, and a terminal compensation $\xi$. 
	The class of admissible contracts $\Xi$ is introduced later in \Cref{sec:principalproblemformulation} after imposing some integrability requirements.
	\medskip

Given a contract $\Cc$ the value received by a time-inconsistent agent at the beginning of the problem from choosing an action $\alpha$ typically takes the form
\[
\Vr^\Ar_0(\Cc,\alpha):= \E^{\P^\alpha}\big[ \Ur_\Ar(0,\xi,C_{0,T}^\alpha)\big],\; \text{with}\; C_{t,T}^\alpha:=\int_t^T c_r(t,X_{\cdot \wedge r},\pi_r,\alpha_r)\d r,
\]
where $\Ur_\Ar:[0,T]\times \R\times \R\longrightarrow \R, (t,{\rm x},{\rm c})\longmapsto \Ur_\Ar(t,{\rm x}, {\rm c})$ denotes the agent's utility function and $C_{t,T}^\alpha$ denotes the cumulative net cost functional. 
	We highlight that the generic dependence of both $\Ur_\Ar$ and $c$ on $t$ accounts for the sources of time-inconsistency. 
	In the classic literature, utilities are usually classified under two categories, namely
\begin{enumerate}[label=$(\roman*)$, ref=.$(\roman*)$,wide, labelindent=0pt]
\item \emph{separable utility} functions, \emph{i.e.} ${\Ur_\Ar}(t,{\rm x}  ,{\rm c})={\Ur_\Ar}(t,{\rm x})-{\rm c}$,
\item \emph{non-separable utility} functions, \emph{i.e.} ${\Ur_\Ar}(t,{\rm x},{\rm  c})={\Ur_\Ar}(t,{\rm x}-{\rm c})$.
\end{enumerate}

For instance, in the separable case the agent's value takes the familiar form

\[ \Vr^\Ar_0(\xi,\alpha)= \E^{\P^\alpha}\bigg[ \Ur_\Ar(0,\xi) -\int_0^T c_r(0,X_{\cdot \wedge r},\pi_r,\alpha_r)\d r\bigg],\]

which, by the Blumenthal zero--one law, satisfies $ \Vr^\Ar_0(\xi,\alpha)= Y_0^{0,\alpha}$, for $Y_0^{0,\alpha}$ the initial value of the first component of $(Y^{0,\alpha},Z^{0,\alpha})$ solution to the BSDE
\[ 
Y_t^{0,\alpha} = \Ur_\Ar(0, \xi) +\int_t^T \big( \sigma_r(X_{\cdot\wedge r})b_r(X_{\cdot\wedge r},\alpha_r) Z_r^{0,\alpha} -c_r(0,X_{\cdot \wedge r},\pi_r,\alpha_r)\big) \d r - \int_t^T  Z_r^{0,\alpha} \cdot \d X_r,\; \P\as
\]
Moreover, in the (time-consistent) case in which the agent discounts exponentially with constant factor $\rho$, \emph{i.e.} $\Ur_\Ar(t,{\rm x})= \e^{-\rho (T-t)}\Ur_\Ar({\rm x})$ and $c_t(s,x,p,a)= \e^{-\rho (t-s)}c_t(x,p,a)$, it holds that
\[ 
Y_t^{0,\alpha} = \Ur_\Ar(\xi) +\int_t^T \big( \sigma_r(X_{\cdot\wedge r})b_r(X_{\cdot\wedge r},\alpha_r) Z_r^{0,\alpha} -c_r(X_{\cdot \wedge r},\pi_r,\alpha_r)-\rho Y_r^{0,\alpha} \big) \d r - \int_t^T  Z_r^{0,\alpha} \cdot \d X_r,\; \P\as
\]

The previous representation corresponds to a so-called recursive utility particularly known as \emph{standard additive utility}, see \citet*{epstein1989substitution}. 
	Let us remark that an analogous argument holds in the case of the non-separable exponential utility and refer to \citet*{el1997backward} for more examples of recursive utilities. 
	Intuitively, a recursive utility can be viewed as an extension of the classic separable or non-separable utilities in which the instantaneous utility depends on the instantaneous action $\alpha_t$ and the future utility via $Y_t^{0,\alpha}$. 
	Extrapolating these ideas, we may arrive at considering rewards functionals of the form $ \Vr^\Ar_0(\xi,\alpha)=Y_0^{0,\alpha}$ where the pair $(Y^{\alpha},Z^\alpha)$ satisfies the BSVIE
\begin{align}\label{eq:bsvieval}
Y_t^{t,\alpha} = \Ur_\Ar(t, \xi) +\int_t^T  h_r\big(t,X_{\cdot\wedge r},Y_r^{t,\alpha}  , Z_r^{t,\alpha},\pi_r,\alpha_r\big) \d r-\int_t^T  Z_r^{t,\alpha} \cdot \d X_r,\; \P\as, \; t\in [0,T].
\end{align}

By letting both $\Ur_\Ar$ and $h$ depend on $t$ we allow for general discounting structures and incorporate time-inconsistency into the agent's preferences. 
	Moreover, the previous discussion shows that this formulation encompasses time-inconsistent recursive utilities too. 

\begin{remark}
In a Markovian framework, time-inconsistent agents whose reward functional is given by {\rm \eqref{eq:bsvieval}} have been considered in {\rm \citet*{wei2017time}}, {\rm \citet*{wang2019time}} and {\rm \citet*{hamaguchi2020extended}}. In these works, the dynamics of the controlled state process are given in strong formulation and, following the game-theoretic approach, they considered a refinement of the notion of equilibrium in {\rm \cite{ekeland2008investment}} that was suitable to each of their settings. In this work, we use {\rm BSVIEs} to model the agent's reward and extend the non-Markovian framework proposed in {\rm\cite{hernandez2020me}}.
\end{remark}

Let us now present this formulation properly. We define the set of admissible actions, recall $A$ is compact, as
\[ \Ac:=\{ \alpha \in \Pc_{\rm pred} ( A,\F): \eqref{eq:integalpha}\text{ holds}\},\]

and assume we are given jointly measurable mappings $h:  [0,T] \times \Omega\times \R  \times \R^{ n } \times \R \times A  \longrightarrow\R ,\;   h_\cdot(\cdot ,y,z,p,a )\in \Pc_{{\rm prog}}(\R,\F)$ for any $(y,z,p,a)\in \R \times\R^{n}\times \R \times A $, and $\Ur_\Ar:  [0,T]\times \R\longrightarrow \R$ satisfying the following set of assumptions.

\begin{assumption}\label{assump:datareward}

\begin{enumerate}[label=$(\roman*)$, ref=.$(\roman*)$,wide,  labelindent=0pt]
\item  \label{assump:datareward:0} For every $s\in [0,T]$, ${\rm x}\longmapsto \Ur_\Ar(s,{\rm x})$ is invertible, \emph{i.e.} there exists a mapping $\Ur_\Ar^{(-1)}:[0,T]\times \R\longrightarrow\R$ such that $\Ur_\Ar^{(-1)}(s,\Ur_\Ar(s,{\rm x}))={\rm x};$
\item  \label{assump:datareward:i} $(s,y,z) \longmapsto  h_t(s,x,y,z,p,a)$ $($resp. $s\longmapsto \Ur^\Ar(s,{\rm x}))$ is continuously differentiable. $\nabla h _\cdot(s,\cdot, u, v ,y,z,p,a)\in \Pc_{\rm prog}(\R,\F)$ for all $s\in [0,T]$, where $\nabla h:[0,T]^2\times \Omega \times (\R \times \R^{n} )^2\!\times \R\times A    \longrightarrow \R$ is defined by
\[
 \nabla h_t (s,x,u, v ,y,z,p,a):=\partial_s h_t(s,x,y,z,p,a)+\partial_y h_t(s,x,y,z,p,a){ u}+\sum_{i=1}^n \partial_{z_{i}} h_t(s,x,y,z,p, a){ v}_{i};
 \]
\item\label{assump:datareward:ii} for $\varphi \in \{h, \partial_s h \}$, $(y,z,a)\longmapsto  \varphi_t(s,x,y,z,p,a)$ is uniformly Lipschitz-continuous 
%\textcolor{magenta}{with linear growth... Not sure if I will still need this (growth)}, 
\emph{i.e.} there exists some $C>0$ such that $\forall (s,t,x,p,y,\tilde y,z,\tilde z,a,\tilde a)$,
\[
 |\varphi_t(s,x,y,z,p,a)-\varphi_t(s,x,\tilde y,\tilde z,p,\tilde a)|\leq C\big(|y-\tilde y|+|\sigma_t(x)^\t(z-\tilde z)|+|a-\tilde a|\big).
 %, \; |\varphi_t(s,x,y,z,a)|\leq C\big(1+|y|+|\sigma_t(x)^\t z|\big).
 \]
%\item \label{assump:datareward:iii}$ \big(\tilde h,  \nabla \tilde h):=\big(\tilde h(s) , \nabla \tilde h(s)\big)_{s\in [0,T]}\in \mathbb{L}^{1,2,2}(\R)\times \mathbb{L}^{1,2,2}(\R)$, where $ \big(\tilde h_\cdot(s), \nabla \tilde h_\cdot(s)\big):=\big(h_\cdot(s,\cdot,0,0,0),\nabla h_\cdot(s,\cdot,0,0,0,0,0)\big)$.
\end{enumerate}
\end{assumption}

\begin{remark}
Let us comment on the previous assumptions. The first condition guarantees we can identify units of utility with terminal contract payments. Indeed, the utility $\Ur_\Ar(s,\xi)$ is sufficient to identify, via $\Ur_\Ar^{(-1)}$, the payment $\xi$. The second assumption guarantees sufficient regularity, with respect to the variable source of inconsistency, of the data prescribing the agent's reward.
\end{remark}

We assume the agent has a reservation utility $R_0\in\R$ below which he refuses to take the contract. 
	The agent is hired at time $t=0$, and the contracts $\Cc$ offered by the principal, for which she can only access the information about the state process $X$, are assumed to provide the agent with a flow of continuous payments and a compensation at the terminal time $T$. 
	\textcolor{black}{Thus, we denote by $\Cf_0$, see \Cref{sec:spacesandassupm} for the definition of the integrability spaces, as the collection of contracts $\Cc=(\pi,\xi)\in\Pi\times \Xi$ for the families}
\begin{itemize}
\item $\Xi$ of $\R$-valued, $\Fc_T$-measurable $\xi$ such that $\big(( \Ur_\Ar(s,\xi))_{s\in [0,T]} ,(\partial_s  \Ur_\Ar(s,\xi))_{s\in [0,T]}\big)\in \Lc^{2,2}\times  \Lc^{2,2}$,
\item $\Pi$ of $\R$-valued, $\F$-predictable $\pi$ such that $\big(( h_\cdot(s,\cdot,0,0,\pi_\cdot, 0))_{s\in [0,T]} ,(\partial_s h_\cdot(s,\cdot,0,0,\pi_\cdot, 0))_{s\in [0,T]} \big)\in \mathbb{L}^{2,2}\times  \mathbb{L}^{2,2}$.
\end{itemize}
%\todo[inline]{Need to say where $\pi$ takes values.}
If hired, the agent chooses an effort strategy $\alpha\in \Ac$, and at any time $t\in[0,T]$, his value, from time $t$ onwards, from performing $\alpha$ is given by 
\[
\Vr_t^\Ar(\Cc,\alpha):= Y_t^{t,\alpha},
\] 
where the pair $(Y^{\alpha},Z^\alpha)$ satisfies the BSVIE \eqref{eq:bsvieval}. We recall ${\rm V}^{\rm A}(\Cc,\alpha)$ is commonly referred to in the literature as the \emph{continuation utility}. We always interpret ${\rm V}^{\rm A}(\Cc,\alpha)$ as a map from $[0,T]\times\Cc([0,T],\R^n)$ to $\R$.
\medskip

Given the choice of reward, the problem of the agent is time-inconsistent. 
	We therefore assume that the agent is a so-called sophisticated time-inconsistent agent who, aware of his inconsistency, can anticipate it, thus making his strategy time-consistent. 
	Consequently, the problem of the agent can be interpreted as an intra-personal game in which he is trying to balance all of his preferences and searches for sub-game perfect Nash equilibria. 
	We recall the definition of an equilibrium strategy introduced in \cite{hernandez2020me}, see further comments in \Cref{rmk:defequi}. 
	Let $\{\alpha^\star,\alpha\} \subseteq \Ac$, $t \in [0,T]$, and $\ell\in (0, T-t]$, we define $\nu\otimes_{t+\ell}\nu^\star:=\nu\1_{[ t, t+\ell)}+\nu^\star\1_{[ t+\ell ,T]}$.

\begin{definition}\label{def:equilibrium}
Let $\alpha^\star \in \Ac$. We say $\alpha^\star$ is an equilibrium if for any $\eps>0$, $\ell_\eps>0$, where
\begin{align*}
\ell_\eps:=\inf \Big\{\ell >0: \exists \alpha\in\Ac,\;  \P\big[\big\{x\in \Omega: \exists  t \in [0,T],\;  {\rm V}^{\rm A}_t(\Cc,\alpha^\star)< {\rm V}^{\rm A}_t(\Cc,\alpha\otimes_{t+\ell}\alpha^\star \big) -\eps \ell\big\} \big]>0 \Big\}.
\end{align*}
Given a contract $\Cc$, we call $\Ec(\Cc)$ the set of all equilibria associated with $\Cc$.
\end{definition}

As such, the agent's goal is, given a contract $\Cc$ that is guaranteed by the principal, to choose an effort that aligns with his sophisticated preferences, \emph{i.e.} to find $\alpha^\star\in \Ec(\Cc)$. 
	In contrast to the case of a classic time-consistent utility maximiser, for a time-inconsistent sophisticated agent, there could be more than one equilibria with potentially different rewards, see for instance \cite{landriault2018equilibrium}. 
	In this work, we will restrict our attention to the set of contracts \textcolor{black}{inducing a unique equilibrium}. See additional comments about this point in the following remark.

\begin{definition}\label{def:contractswithoneeqvalue}
$\Cf_o$ denotes the family of contracts $\Cc\in \Cf_0$ that lead to a unique equilibrium, \emph{i.e.} $\Ec(\Cc)=\{\alpha^\star\}$.
\end{definition}
%\todo[inline]{$\Ec(\Cc)=\{\alpha^\star\}$ for some $\alpha^\star\in \Ac$. Define admissibility of a contract so that it leads to a unique equilibrium}

All in all, for $\Cc\in \Cf_o$ we can now define
\begin{align*}
\Vr^\Ar_t(\Cc):=\Vr^\Ar_t(\Cc,\alpha^\star),\; \alpha^\star\in \Ec(\Cc).
\end{align*}

\begin{remark}\label{rmk:defequi}
\begin{enumerate}[label=$(\roman*)$, ref=.$(\roman*)$,wide,  labelindent=0pt]

\item In the non-Markovian framework, the strategy devised in {\rm \cite{hernandez2020me}} builds upon the approach in {\rm \cite{bjork2017time}} to study rewards given by conditional expectations of non-Markovian functionals. 
	This approach is based on decoupling the sources of inconsistency in the agent's reward and requires introducing the terms $\partial_s\Ur_\Ar$ and $\nabla h$ into the analysis, see {\rm \Cref{sec:appendixdpp}} for details.
	The integrability condition in the definition of $\Pi\times \Xi$ guarantees that the {\rm BSVIE} \eqref{eq:bsvieval}  is well-defined.  
	We also mention that {\rm \Cref{thm:dpp}} generalises the extended dynamic programming principle obtained in {\rm \cite{hernandez2020me}} for the case of rewards given by \eqref{eq:bsvieval} and equilibrium actions as in {\rm \Cref{def:equilibrium}}.

\item The previous definition of equilibrium can be regarded as a reformulation of the classic definition, in {\rm \cite{ekeland2008investment}}, via the $\liminf$. Indeed, it follows from {\rm \Cref{def:equilibrium}} that given $( \eps, \ell)\in (0,\infty)\times (0,\ell_\eps)$,  $\exists  \widetilde \Omega\subseteq \Omega$, $\P[\widetilde \Omega]=1$, such that
\[
\Vr^\Ar_t(\Cc,\alpha^\star)(x)- \Vr^\Ar_t(\Cc,\alpha \otimes_{t+\ell}\alpha^\star)(x) \geq  -\eps \ell,\; \forall (t,x,\alpha)\in [0,T]\times \widetilde \Omega \times \Ac.
\]
\item Lastly, we also expand on the necessity to \textcolor{black}{focus our attention on contracts that lead to a unique equilibrium}. 
	The need for said restriction is inherent to contract theory models involving a game-theoretic formulation at the level of the agent. 
	Indeed, in either the case of a finite number of competitive interacting agents seeking a Nash equilibrium, see {\rm \citet*{elie2019contracting}}, or a continuum of players seeking a mean-field equilibrium, see {\rm \citet*{elie2019tale}}, it is generally possible for multiple equilibria to exist. 
	In such cases, the existence of a Pareto-dominating equilibrium, one for which all agents receive no worse reward if deviating from a current equilibrium, is by no means guaranteed. 
	In the context of contract theory, this means that there is no clear rule at the level of the problem of the agent to decide which equilibria should be taken for any two equilibria providing different values to different players. 
	As giving control of this decision to the principal makes little practical sense, one way to bypass this is to \textcolor{black}{focus on contracts that lead to a unique equilibrium, as we did here}. 
\medskip 

\textcolor{black}{Anticipating our analysis in {\rm \Cref{sec:spacesandassupm}}, we mention that this assumption is intimately related to the well-posedness of a fairly novel class of {\rm BSVIEs}. 
	In the Lipschitz setting of this paper, we present conditions on the data of the problem under which this is the case for any $\Cc\in \Cf_o$, see {\rm \Cref{assump:uniquemax}} and {\rm \Cref{rmk:uniqueequilibria}}. 
	As such, this is not such a stringent assumption in our context.} 
	%Moreover, in all the examples considered in {\rm \Cref{sec:sbexamples}} the agent's participation constraint is saturated, \emph{i.e.} it is held to the reservation level $R_0$. 
	%Hence, one suspects this might be a more general phenomenon beyond this document's scope.
	%todo[inline]{is the last line related to the rest?}
\end{enumerate}
\end{remark}

\subsection{The principal's problem}\label{sec:principalproblemformulation}

We now present the principal's problem. We therefore let $\Cf\subseteq\Cf_o$ be the set of admissible contracts, defined by
\[
\Cf:=\big\{\Cc \in\Cf_o :{\rm V}^{\rm A}_0(\Cc)\geq R_0\big\}.
\]
In such manner, any contract $\Cc\in \Cf$ is \emph{implementable}, that is, there exists \textcolor{black}{an equilibrium strategy, namely $\alpha^\star\in \Ec(\Cc)$, for the agent's problem.} 
%Even though in the previous section, we focused our attention on equilibria for which the agent gets the same value, we still need a rule for whenever there are more than one such equilibrium. Following the standard convention in the literature, we assume that if there is more than one such action for which the agent is indifferent by assumption, he implements the one that is best for the principal.} 
\medskip

The principal has utility functionals, ${\rm U_P}:\Omega \times \R\longrightarrow \R$, and $u^{\rm p}:[0,T]\times \Omega \times \R \times A\longrightarrow \R$ and solves the problem
\[
{\rm V^P}:=\sup_{\Cc\in\Cf}%\sup_{\alpha\in\Ec(\Cc)}
\E^{\P^{\alpha^\star}}\bigg[{\rm U_P}\big(X_{\cdot \wedge T},\xi\big)+\int_0^T u_r^{\rm p}(X_{\cdot\wedge r},\pi_r,\alpha_r^\star)\d r \bigg].
\]
\begin{remark}
We point out that we have assumed the principal is a standard utility maximiser. 
	This is because, in our opinion, the crux of the problem lies in identifying a proper description of the problem of the principal when contracting a time-inconsistent sophisticated agent. 
	In the case of a time-consistent agent, {\rm \cite{cvitanic2015dynamic}} identifies this description as a standard stochastic control problem with an additional state variable. 
	Therefore, in the case of a classic time-consistent agent and a time-inconsistent principal, following {\rm \cite{cvitanic2015dynamic}}, one expects the problem of the principal to boil down to a non-Markovian time-inconsistent control problem with an additional state variable. 
	As studied in {\rm \cite{hernandez2020me}}, these problems are characterised by an infinite family of {\rm BSDEs}, analogue to the {\rm PDE} system in {\rm \cite{bjork2017time}} in the Markovian case.
\end{remark}

\section{The first-best problem}\label{sec:firstbestexamples}

%\textcolor{magenta}{Only terminal payment, check notation for contract}\medskip

In the first-best, or risk-sharing, problem, the principal chooses both the effort and the contract for the agent, and she is simply required to satisfy the participation constraint. To provide appropriate characterisations of the solution to several examples, we will focus on a particular class of reward functionals for the agent.  We recall that our goal is to study the second-best problem introduced in the previous section. As such, despite its inherent interest, the results in the current section serve mainly as a reference point for the general analysis we conduct in \Cref{sec:secondbestproblem}. Moreover, the following specification is covered by the general formulation presented in \Cref{sec:probstate}, see \Cref{rmk:specificationFBproblem}, and it is yet rich enough to cover examples of both separable and non-separable utilities. We highlight that in the next two examples, we consider contracts consisting of only a terminal payment, \emph{i.e.} $\Cc=\xi$.

 \medskip

Let us assume the agent has a given increasing and concave utility function $\Ur_\Ar^o: \R\longrightarrow \R$ and Borel-measurable discount functions $g$, and $f$ defined on $[0,T]$, taking values in $(0,+\infty)$, with $g(0)=f(0)=1$, which are assumed to be continuously differentiable with derivatives $g^\prime$, and $f^\prime$. Lastly, we have Borel-measurable functionals $k$ and $c$, defined on $[0,T]\times\Omega \times A$ and taking values in $\R_+$.\medskip

We then specify the agent's continuation utility by
\begin{align}\label{eq:reward}
\Vr_t^\Ar(\xi,\alpha)=\E^{\P^\alpha}\bigg[\Kc^{t,\alpha}_{t,T}f(T-t)\Ur_\Ar^o(g(T-t) \xi)-\int_t^T \Kc^{t,\alpha}_{t,r}  f(r-t)c_r(X_{\cdot \wedge r},\alpha_r)\d r\bigg| \Fc_t \bigg],\; (t,\alpha)\in [0,T]\times \Ac,
\end{align}
where
\begin{align*}
\Kc^{s,\alpha}_{t,T} :=\exp\bigg(\displaystyle\int_t^T g(r-s)k_r(X_{\cdot \wedge r},\alpha_r)\d r\bigg), \; (s,t,\alpha)\in [0,T]^2\times \Ac.
\end{align*}

Regarding the principal, we assume she has her own utility function $\Ur_\Pr^o:\R \longrightarrow \R$, which we assume to be concave and strictly increasing so that
\[
{\rm V^P}=\sup_{(\alpha,\xi)\in\Ac\times \Cf} \E^{\P^\alpha}\big[\Ur_\Pr^o\big(\Gamma(X_T)-\xi\big)\big],
\]
where $\Gamma:\R^n\longrightarrow\R$ denotes a mechanism by which the principal collects the values of the $n$ different coordinates of the state process $X$.

\begin{remark}\label{rmk:specificationFBproblem}
\begin{enumerate}[label=$(\roman*)$, ref=.$(\roman*)$,wide,  labelindent=0pt]
\item As commented above, the previous type of rewards are covered by the formulation via {\rm BSVIEs} \eqref{eq:bsvieval} and satisfy {\rm \Cref{assump:datareward}}. It corresponds to the choice $\Ur_\Ar(s, {\rm x} )=f(T-s) \Ur_\Ar^o(g(T-s) {\rm x} ),\; \partial_s  \Ur_\Ar  (s, {\rm x} )= -f^\prime(T-s)\Ur_\Ar^o(g(T-s) {\rm x} )-f(T-s)g^\prime(T-s)\partial_{{\rm x}}\Ur_\Ar^o(g(T-s) {\rm x} ),$ and 
\begin{align*}
h_t(s,x,y,z,a)&=\sigma_t(x)b_t(x,a)\cdot z-f(t-s)c_t(x,a)+g(t-s)k_t(x,a)y,\\
\nabla h_t(s,x,u,v,y,z,a)&=\sigma_t(x)b_t(x,a)\cdot v+f^\prime(t-s)c_t(x,a)-g^\prime(t-s)k_t(x,a)y+g(t-s)k_t(x,a)u.
\end{align*}
Regarding the principal, our specification corresponds to $\Ur_\Pr(x,{\rm x})=\Ur_\Pr^o(\Gamma(x_T)-{\rm x})$. Let us mention that, to facilitate the resolution of the following examples, we assumed that ${\rm U_P}$ depends only on the terminal value of $x_{T}$. This allow us to use the dynamics of $X$ as given in {\rm \Cref{sec:controlledstateeq}}. We highlight this assumption is not necessary in general analysis for the second best problem we present in {\rm \Cref{sec:secondbestproblem}}.
\end{enumerate}
\end{remark}

We now move on to characterise the solution to the first-best problem in the case of a time-inconsistent agent with both separable and non-separable utility functions. Anticipating the result, we highlight that in the first-best problem, the problem of the principal reduces to solving a standard stochastic control problem.

\subsection{Non-separable utility}\label{sec:raFB}

We recall that the CARA utility function, commonly known as the exponential utility, constitutes the stereotypical example of non-separable utility.  We then consider \eqref{eq:reward} under the choice $c= 0$, 
\[ {\Ur_\Pr^o}({\rm x})=-\frac{\e^{-\gamma_{\smallfont\Pr} {\rm x}}}{\gamma_{\rm P}}, \; { \Ur_\Ar^o}({\rm x} ):=-\frac{\e^{-\gamma_{\smallfont\Ar} {\rm x} }}{\gamma_{\rm A} }, \; {\rm x} \in \R, \;\gamma_\Ar>0,\; \gamma_\Pr>0,
\]   
$k_t(x,a)= \gamma_\Ar k_t^o(x,a)$ and assume $a\longmapsto k_t^o(x,a)$ is convex for any $(t,x)\in [0,T]\times \Omega$. We then have that
\begin{align}\label{eq:rewardraFB}
{\Vr}^{\Ar}_t(\xi,\alpha):=\E^{\P^\alpha}\Big [f(T-t) \Ur_\Ar^o\big(g(T-t)\xi-K_{t,T}^{t,\alpha} \big) \Big |\Fc_t\Big], \text{ where } K_{t,T}^{s,\alpha}:=\int_t^T g(r-s) k_r^o(X,\alpha_r)\d r.
\end{align}

The value of principal is thus obtained through the following constrained optimisation problem
\[
{\rm V^P_{FB}}:=\sup_{(\alpha,\xi)\in\Ac\times\Cf}\E^{\P^\alpha}\big[\Ur_\Pr^o(\Gamma(X_T)-\xi)\big],\; \text{\rm s.t.}\; \E^{\P^\alpha}\bigg[f(T){\rm U_A}\bigg(g(T)\xi-\int_0^Tg(r)k_r^o \big(X_{\cdot\wedge r},\alpha_r\big)\d r\bigg)\bigg]\geq R_0.
\]

Note that, the concavity (resp. convexity) of both $\Ur_\Ar^o$ and $\Ur_\Pr^o$ (resp. $a\longmapsto k_t^o(x,a)$) and the fact $(\Ac,\Cc)$ is a convex set, imply that $\Vr_{\rm FB}^\Pr$ is a concave optimisation problem. The Lagrangian associated to this problem, where $\rho\in \R_+$ denotes the multiplier of the participation constraint, is
\[
\mathfrak{L}(\alpha,\xi,\rho):=\E^{\P^\alpha}\bigg[\Ur_\Pr^o(\Gamma(X_T)-\xi)+\rho f(T)\Ur_\Ar^o\bigg(g(T)\xi-\int_0^Tg(r)k_r^o \big(X_{\cdot\wedge r},\alpha_r\big)\d r\bigg)\bigg]-\rho R_0  , \; (\alpha,\xi,\rho)\in \Ac\times \Cf\times \R_+.
\]
For convenience of the reader, we recall that the dual problem $\Vr_{\rm FB}^{\Pr,{\rm d}}$, which is an unconstrained control problem, is in general an upper bound of $\Vr_{\rm Fb}^\Pr$ and is defined by
\begin{align}\label{eq:weakduality}
\Vr_{\rm FB}^\Pr= \sup_{(\alpha,\xi)\in\Ac\times\Cf}  \inf_{ \rho \in \R_+ }\;  \mathfrak{L}(\alpha,\xi,\rho)\leq  \inf_{ \rho \in \R_+ }\sup_{(\alpha,\xi)\in\Ac\times\Cf} \mathfrak{L}(\alpha,\xi,\rho)=: \Vr_{\rm FB}^{\Pr,{\rm d}},
\end{align}
where we used the convention $\sup_{\emptyset}=-\infty$. As it is commonplace for convex problems, the next result exploits the absence of duality gap, \emph{i.e.} $\Vr_{\rm FB}^\Pr=\Vr_{\rm FB}^{\Pr,{\rm d}}$, to compute the value of $\Vr_{\rm FB}^\Pr$. It uses the following notations
\[
\bar \gamma:= \frac{\gamma_\Ar\gamma_\Pr g(T)}{\gamma_\Ar g(T)+\gamma_\Pr}, \;C_{y}:=-\frac{1}{\gamma_{\rm P}}\exp\bigg( \frac{\gamma_{\rm P}}{g(T)} {\Ur_\Ar^o}^{(-1)} (y)\bigg).
\]
\begin{proposition}\label{prop:solFBra}
Let 
\begin{align*}
\Vr_{{\rm cont}}:=\sup_{\alpha\in \Ac} \E^{\P^\alpha}\bigg[- \frac{1}{\bar\gamma}\exp\bigg(-\bar\gamma \Gamma(X_T)+\int_0^T\frac{\bar\gamma g(r)k_r^o}{g(T)}  \big(X_{\cdot\wedge r},\alpha_r\big)\d r   \bigg) \bigg],\;
\rho^\star := \frac{1}{g(T)f(T)}\bigg( \frac{ \bar \gamma f(T)  }{\gamma_\Ar R_0} \Vr_{{\rm cont}} \bigg)^{1+\frac{\gamma_\smallfont{\Pr}}{\gamma_\smallfont{\Ar} g(T)}} .
\end{align*}
Suppose $\Vr_{{\rm cont}}<\infty$ and for any $(\alpha,\rho)\in \Ac\times \R_+$, $\xi^\star(\rho,\alpha)\in \Cf$ where
\begin{align*}
\xi^\star(\rho,\alpha):= \frac{1}{g(T)\gamma_\Ar+\gamma_\Pr}\Big(\gamma_\Pr  \Gamma(X_T) +\gamma_\Ar K_{0,T}^{0,\alpha} +\log\big(\rho^\star g(T) f(T)\big) \Big).
\end{align*}
Then
\[
\Vr^{\Pr}_{\rm FB}=C_{\frac{R_\smallfont{0}}{f(T)}}  \Vr_{\rm cont}^{\frac{\gamma_\smallfont{\Pr}}{\bar \gamma}}.\; 
\]
Moreover, if $\alpha^\star$ is an optimal control for $\Vr_{\rm cont}$, then an optimal contract is given by $\xi^\star(\rho^\star,\alpha^\star)$.
\end{proposition}
\begin{comment}
Moreover, if in addition:
\begin{enumerate}[label=$(\roman*)$, ref=.$(\roman*)$,wide,  labelindent=0pt]
\item The maps $\sigma$, $b$ and $k$ do not depend on the $x$ variable$;$

\item For any $t\in[0,T]$, the map $A\ni a\overset{\tilde g}{\longmapsto} \lambda_t(a)-g(t)k^\star_t(a)-\frac{\bar \gamma}2|\sigma_t |^2$ has a unique maximiser $a^\star(t)$, such that $a^\star:=(a^\star(t))_{t\in [0,T]}$ is Lebesgue integrable.
\end{enumerate}
Then, 
\[
\frac{\ln(\rho^\star)}{\gamma_\Ar+\gamma_\Pr}=\Ur_\Ar^{(-1)}\bigg(\frac{ R_0}{f(T)}\bigg)- \int_0^T  \tilde g_r(a^\star(r))\d r
\]
\end{comment}

\subsection{Separable utility}\label{sec:separableFB}
We consider the case $k=0$ and $g= 1$ in \eqref{eq:reward}, and assume $a\longmapsto c(t,x,a)$ is convex for any $(t,x)\in [0,T]\times \Omega$. The agent's reward from time $t\in [0,T]$ onwards is given by
\begin{align}\label{eq:rewardseparable}
{\rm V}^{\rm A}_t(\xi,\alpha)=\E^{\P^\alpha}\bigg[f(T-t)\Ur_\Ar^o(\xi)-\int_t^Tf(s-t)c_s\big(X_{\cdot\wedge s},\alpha_s\big)\mathrm{d}s\bigg|\mathcal F_t\bigg], \; (t,\alpha,\xi)\in [0,T]\times \Ac\times \Cf.
\end{align}

The value of principal is thus obtained through the following constrained optimisation problem
\[
{\rm V^P_{FB}}:=\sup_{(\alpha,\xi)\in\Ac\times\Cf}\E^{\P^\alpha}\big[\Ur_\Pr^o(\Gamma(X_T)-\xi)\big],\; \text{\rm s.t.}\; \E^{\P^\alpha}\bigg[f(T)\Ur_\Ar^o (\xi)-\int_0^Tf(r)c_r\big(X_{\cdot\wedge r},\alpha_r\big)\mathrm{d}r\bigg]\geq R_0.
\]
The Lagrangian associated to this problem is
\[
\mathfrak{L}(\alpha,\xi,\rho):=\E^{\P^\alpha}\bigg[\Ur_\Pr^o(\Gamma(X_T)-\xi)+\rho f(T){\Ur_\Ar^o}(\xi)-\rho\int_0^Tf(r)c_r\big(X_{\cdot\wedge r},\alpha_r\big)\mathrm{d}r \bigg]-\rho R_0  , \; (\alpha,\xi,\rho)\in \Ac\times \Cf\times \R_+.
\]

\begin{proposition}\label{prop:solFBseparable}
\begin{enumerate}[label=$(\roman*)$, ref=.$(\roman*)$,wide,  labelindent=0pt]
\item \label{prop:solFBseparable:i} Suppose $\Ur_\Ar^o$ and $\Ur_\Pr^o$ are such that mapping $\xi^\star({\rm x},\rho)$ given as the solution to 
\begin{align*}
-\partial_{\rm x} \Ur_\Pr^o (\Gamma({\rm x})-\xi^\star({\rm x},\rho))+\rho f(T) \partial_{\rm x} \Ur_\Ar^o (\xi^\star({\rm x}, \rho) )=0, \; ({\rm x},\rho)\in \R^n\times \R_+,
\end{align*}
is well-defined and $\xi^\star(X_T, \rho  )\in \Cf$ for any $\rho\in \R_+$. Let
\[
\Vr_{\rm cont}(\rho):=\sup_{\alpha\in \Ac} \E^{\P^\alpha}\bigg[\Ur_\Pr^o\big(\Gamma(X_T)-\xi^\star(X_T,  \rho )\big)+\rho f(T)\Ur_\Ar^o\big(\xi^\star(X_T, \rho  )\big)-\rho\int_0^Tf(r)c_r\big(X_{\cdot\wedge r},\alpha_r\big)\mathrm{d}r \bigg], \; \rho\in \R_+.
\]
Then
\[
{\rm V^{P,d}_{FB}}=\inf_{\rho\in \R_+}\Big\{ -\rho R_0+\Vr_{\rm cont}(\rho) \Big\}.
\]
Moreover, suppose the pair $\big(\alpha^\star(\rho^\star),\xi^\star(X_T,\rho^\star)\big)$ is feasible for the primal problem, where $\alpha^\star(\rho)$ $($resp. $\rho^\star)$ denote the maximiser in ${\rm V_{cont}}(\rho)$ $($resp. the above problem$)$, which we assume to exist. Then, there is no duality gap, \emph{i.e.}
\[{\rm V^{P}_{FB}}={\rm V^{P,d}_{FB}},\]
the optimal contract is given by $\xi^\star(X_T,\rho^\star)$.

\item \label{prop:solFBseparable:ii} If $\Ur_\Pr^o({\rm x})=\Ur_\Ar^o({\rm x})={\rm x}$, for $\alpha\in \Ac$ let
\begin{align*}\label{eq:FBcontract}
\hat \Cf(\alpha)=\bigg\{\xi\in \Cf:  \E^{\P^\alpha}[\xi^\star]= \E^{\P^\alpha}\bigg[\int_0^T\frac{f(r)}{f(T)}c_r\big(X_{\cdot\wedge r},\alpha_r\big)\mathrm{d}r\bigg]+ \frac{R_0}{f(T)}\bigg\}.
\end{align*}
Then, the problem of the principal is given by the solution to the standard control problem
\[
\Vr^\Pr_{\rm FB}=-\frac{R_0}{f(T)}+\sup_{\alpha\in \Ac}  \E^{\P^\alpha}\bigg[\Gamma(X_T)-\int_0^Tf(r)c_r\big(X_{\cdot\wedge r},\alpha_r\big)\mathrm{d}r\bigg].
\]
Moreover, for $\alpha^\star\in \Ac$ an optimal control of this problem, $\hat\Cf(\alpha^\star)$ contains all the optimal contracts for the principal, \emph{e.g.} the deterministic contract
\[
\xi^\star:=\frac{R_0}{f(T)}+f(T)^{-1}\E^{\P^{\alpha^\star}}\bigg[\displaystyle \int_0^Tf(r)c_r\big(X_{\cdot\wedge r},\alpha_r^\star\big)\mathrm{d}r\bigg].
\]
\end{enumerate}
\end{proposition}

\begin{remark}
We remark that the assumption on the utility functions in {\rm \Cref{prop:solFBseparable}} is relatively reasonable. Indeed, it is immediately satisfied, for instance, in either of the following scenarii
\begin{enumerate}[label=$(\roman*)$, ref=.$(\roman*)$,wide,  labelindent=0pt]
\item $\Ur_\Pr^o({\rm x})={\rm x}$ and $\Ur_\Ar^o({\rm x})$ is strictly increasing$;$

\item for $\varphi\in \{\Ur_\Ar^o,\Ur_\Pr^o\}$, $x\longmapsto \varphi({\rm x})$ is concave, strictly increasing and satisfies the following conditions
\begin{align*}
\lim_{{\rm x}\to -\infty} \partial_{\rm x} \varphi({\rm x})=\infty,\; \lim_{{\rm x}\to \infty}\partial_{\rm x} \varphi({\rm x})=0.
\end{align*}
\end{enumerate}
\end{remark}

\section{The second-best problem: general scenario}\label{sec:secondbestproblem}

In this section, we bring back our attention to the second-best problem faced by the principal
\[
{\rm V^P}=\sup_{\Cc\in\Cf} \E^{\P^{\alpha^\smallfont{\star}}}\bigg[{\rm U_P}\big(X_{\cdot \wedge T},\xi\big)+\int_0^T u_r^{\rm p}(X_{\cdot\wedge r},\pi_r,\alpha_r^\star)\d r \bigg].
\]

We will exploit the theory of type-I BSVIEs. Consequently, we first introduce suitable integrability spaces.

\subsection{Integrability spaces and Hamiltonian}\label{sec:spacesandassupm}

Following \cite[Section 2.2]{hernandez2020unified}, to carry out the analysis we introduce the spaces

\begin{enumerate}[label=$\bullet$, ref=.$(\roman*)$,wide,  labelindent=5pt]

\item $\Lc^2$ of $\xi \in \Pc_{\rm meas}(\R,\Fc)$, such that $ \|\xi\|_{\Lc^\smallfont{2}}^2:= \E\big[ |\xi |^2\big]  <\infty;$

\item $\S^2$ of c\`adl\`ag $Y\in \Pc_{\text{prog}}(\R,\F)$ such that $ \|Y\|_{\S^\smallfont{2} }^2:=\E \bigg[ \displaystyle \sup_{t\in [0,T]} | Y_t|^2\bigg] <\infty;$

\item  $\mathbb{L}^{2}$ of $Y\in \Pc_{\text{opt}}(\R,\F)$, with $ \|Y\|_{\mathbb{L}^{\smallfont 2}}^2:= \E\bigg[ \bigg( \displaystyle\int_0^T |Y_r|^2\d r \bigg)\bigg]<\infty$;

\item $\H^2$ of $Z\in \Pc_{\rm pred}(\R^n,\F)$ such that $ \|Z\|_{\H^\smallfont{2} }^2:=\E \bigg[ \displaystyle  \int_0^T |\sigma_r \sigma^\t_r Z_r|^2\d r  \bigg] <\infty;$

To make sense of the class of systems considered in this paper we introduce some extra spaces.

\item Given a Banach space $(\I,\| \cdot \|_{\I})$ of $E$-valued processes, we define $(\I^2,\|\cdot \|_{\I^{\smallfont 2}})$ the space of $U\in \Pc^2_{\text{meas}}(E,\Fc)$ such that $([0,T],\Bc([0,T]))$ $ \longrightarrow (\I^{2},\|\cdot \|_{ \I^{\smallfont 2}}): s \longmapsto U^s $ is continuous and $ \|U\|_{\I^{\smallfont 2}}:=\displaystyle \sup_{s\in [0,T]}  \|U^s\|_{\I} <\infty.$

For instance, $\S^{2,2} $  denotes the space of  $Y\in \Pc^2_{\rm meas }(\R,\Fc)$ such that $([0,T],\Bc([0,T])) \longrightarrow (\S^{2} (\R),\|\cdot \|_{ \S^{\smallfont 2} }): s \longmapsto Y^s $ is continuous and $\displaystyle \| Y\|_{\S^{\smallfont{2}\smallfont{,}\smallfont{2}} }:= \sup_{s\in[0,T]} \| Y^s\|_{\S^{\smallfont 2}} <\infty$.

\item $\overline{\H}^{_{\raisebox{-1pt}{$ \scriptstyle 2,2$}}}$ of $(Z_\uptau)_{\uptau \in [0,T]^2 }\in \Pc^2_{\text{meas}}(\R^n,\Fc)$ such that $([0,T],\Bc([0,T])) \longrightarrow (\H^{2},\|\cdot \|_{ \H^{\smallfont 2}}): s \longmapsto Z^s $ is continuously differentiable with derivative $\partial Z$, and $ \|Z\|_{\overline{\H}^{\smallfont{2}\smallfont{,}\smallfont{2}}}^2:=\|Z\|_{\H^{\smallfont{2}\smallfont{,}\smallfont{2}}}^2+\|\Zc\|_{\H^\smallfont{2}}^2<\infty$, where $\Zc:=(Z_t^t)_{t\in [0,T]}\in \H^2$ is given by
\[ Z_t^t:=Z_t^0+\int_0^t \partial Z_t^r \d r. \] 

Lastly, we introduce the space ${\mathfrak H}:=\S^2\times \H^2\times \S^{2,2}\times \Ho \times \S^{2,2}\times \H^{2,2}$. 

\begin{remark}
The second set of these spaces are suitable extensions of the classical ones, whose norms are tailor-made to the analysis of the systems we will study. 
	Some of these spaces have been previously considered in the literature on {\rm BSVIEs}, \emph{e.g.} {\rm \cite{hernandez2020unified}} and {\rm\cite{wang2019time}}. 
	Of particular interest is the space $\Ho$ which allows us to define a good candidate for $(Z_t^t)_{t\in [0,T]}$ as an element of $\H^{2}$, see {\rm \cite{hamaguchi2020extended}}. 
\begin{comment}
Indeed, let $\widetilde \Omega:=[0,T]\times \Xc$, $\tilde \omega:=(t,x)\in \widetilde \Omega$ and
\begin{align*}
\mathfrak{Z}_s(\tilde \omega):= Z^T_t(x)-\int_s^T \partial Z^r_t(x) \d r, \; \d t \otimes \d \P\ae\ \tilde \omega\in \widetilde \Omega,\; s\in [0,T],\vspace{-0.5em}
\end{align*}
so that the Radon--Nikod\'ym property and Fubini's theorem imply $\mathfrak{Z}_s=Z^s, \d t\otimes \d \P\ae$, $s\in [0,T]$. Lastly, as for $\tilde \omega\in \widetilde \Omega$, $s \longmapsto \mathfrak{Z}_s(\tilde \omega)$ is continuous, we may define
\begin{align*} 
Z^t_t := Z^T_t-\int_t^T \partial Z^r_t \d r,\; \text{\rm for} \; \d t \otimes \d \P\ae\; (t,x) \text{ \rm in } [0,T]\times \Xc.\vspace{-0.5em}
\end{align*}
\end{comment}
\end{remark}

\begin{comment}
\item $\overline{\H}^{_{\raisebox{-1pt}{$ \scriptstyle 2,2$}}}_{\rm BMO}(E)$ of $(Z_\uptau)_{\uptau \in [0,T]^2 }\in \Pc^2_{\text{meas}}(E,\Fc_T)$ such that $([0,T],\Bc([0,T])) \longrightarrow (\H^{2}_{\rm BMO}(E),\|\cdot \|_{ \H^{2}_{\rm BMO}}): s \longmapsto Z^s $ is absolutely continuous with respect to the Lebesgue measure with density $\eth Z$, $\Zc\in \H^2(E)$, where $\Zc:=(Z_t^t)_{t\in [0,T]}$ is given by
\[ Z_t^t:=Z_t^0+\int_0^t \eth Z_t^r \d r,\text{ and, } \|Z\|_{\overline{\H}^{_{\raisebox{-2pt}{$ \scriptstyle p,2$}}}}^2:=\|Z\|_{\H^{2,2}}^2+\|\Zc\|_{\H^2}^2<\infty \]
\end{comment}
\end{enumerate}

\subsection{Characterising equilibria and the BSDE system}\label{sec:system}

Building upon the results in \cite{hernandez2020me}, where only the case of an agent with separable utility was considered, we wish to obtain a characterisation of the equilibria that are associated to any $\Cc\in \Cf$. For this we must introduce the Hamiltonian functional $H:[0,T]\times\Omega \times \R \times\R^n\times \R\longrightarrow \R$ given by
\[
H_t(x,{\rm y},{\rm z},p):=\sup_{a\in A} h_t(t,x,{\rm y},{\rm z},p,a),\; (t,x,{\rm y},{\rm z},p)\in [0,T]\times \Omega\times \R\times \R^n\times \R.
\]

Our standing assumptions on $H$ are the following.
\begin{assumption}\label{assump:uniquemax}
\begin{enumerate}[label=$(\roman*)$, ref=.$(\roman*)$,wide,  labelindent=0pt]
\item \label{assump:uniquemax:i}The map $\R\times \R^n \ni (y,z)\longmapsto H_t(x,y,z)$ is uniformly Lipschitz-continuous, %with linear growth, 
\emph{i.e.} there is $C>0$ such that for any $ (t,x,{\rm y},{\rm \tilde y},p,{\rm z},{\rm \tilde z})\in [ 0,T]\times\Omega \times \R^3\times(\R^n)^2$
\[
\big|H_t(x,{\rm y},{\rm z},p)-H_t(x,{\rm \tilde y},{\rm \tilde z},p)\big|\leq C\big( |{\rm y}-{\rm \tilde y}|+ |\sigma_t(x)^\t({\rm z}- {\rm \tilde z})|\big);%\; \big|H_t(x,{\rm y},{\rm z})\big|\leq C\big(1+|{\rm y}|+|\sigma_t(x)^\t {\rm z}|\big);
\]
\item \label{assump:uniquemax:ii}there exists a unique Borel-measurable map $a^\star:[0,T]\times\Omega \times\R\times \R^n\longrightarrow A$ such that
\[
H_t(x,{\rm y},{\rm z},p)=h_t\big(t,x,{\rm y},{\rm z},p,a^\star(t,x,{\rm y},{\rm z},p)\big),\; \forall(t,x,p,{\rm y},{\rm z})\in[0,T]\times\Omega \times \R^2 \times\R^n.
\]
\item \label{assump:uniquemax:iii}The map $\R\times \R^n \ni ({\rm y},{\rm z})\longmapsto a^\star(t,x,{\rm y},{\rm z})$ is uniformly Lipschitz-continuous, %with linear growth, 
\emph{i.e.} there is $C>0$ such that for any $ (t,x,p,{\rm y},{\rm \tilde y},{\rm z},{\rm \tilde z})\in [ 0,T]\times\Omega \times \R^3\times(\R^n)^2$.
\[
\big|a^\star(t,x,{\rm y},{\rm z},p)-a^\star(t,x,{\rm \tilde y},{\rm \tilde z},p)\big|\leq C\big( |{\rm y}-{\rm \tilde y}|+|\sigma_t(x)^\t({\rm z}-{\rm \tilde z})|\big);%\; \big|a^\star(t,x,{\rm y},{\rm z})\big|\leq C\big(1+|{\rm y}|+|\sigma_t(x)^\t {\rm z}|\big);
\]
%\item \label{assump:uniquemax:iv} $(\tilde H,\tilde a^\star)\in \mathbb{L}^{1,2}(\R)\times \mathbb{L}^{1,2}(\R)$, where $\big(\tilde H_\cdot,\tilde a^\star_\cdot \big):=\big(H_\cdot(\cdot,0,0),a^\star_\cdot(\cdot,0,0)\big)$.
\end{enumerate}
\end{assumption}

To ease the notation we introduce $h_r^\star (s,x,y , z,{\rm y}, {\rm z},p):=h_r\big(s,x,y , z,p,a^\star(r,x,{\rm y} , {\rm z},p)\big)$, $\nabla h_r^\star(s,x,u ,  v,y , z,{\rm y},{\rm z},p):=\nabla h_r\big(s,x,u ,v,y , z,p,a^\star(r,x,{\rm y} , {\rm z},p)\big)$, and $b_t^\star\big(x,{\rm y},{\rm z},p\big):=b_t\big(x,a^\star(t,x,{\rm y},{\rm z},p)\big)$.

\begin{remark}\label{rmk:uniqueequilibria}
Let us comment on the previous set of assumptions. 
	Even in the non-Markovian setting of this document, the problem faced by a sophisticated agent is related to a system of equations instead of just one, see {\rm \cite{hernandez2020me}}. 
	This raises many issues, among which is the possibility of multiplicity of equilibria with different values. 
	{\rm \Cref{assump:uniquemax}\ref{assump:uniquemax:i}}, {\rm \ref{assump:uniquemax}\ref{assump:uniquemax:iii}} guarantee that for a given $\Cc\in \Cf$ any equilibria $\alpha^\star\in \Ec(\Cc)$ corresponds to a maximisers of the Hamiltonian. 
	Ultimately, {\rm \Cref{assump:uniquemax}\ref{assump:uniquemax:ii}} guarantees that there is only one maximiser of the Hamiltonian. 
	Let us mention that the existence of $a^\star$ is guaranteed under {\rm \Cref{assump:datareward}\ref{assump:datareward:ii}} by {\rm \citet*[Theorem 3]{schal1974selection}}. 
	This conciliates our focus on contracts leading to unique equilibria as we stated in {\rm \Cref{sec:agentprobstatement}}.
\end{remark} 

Under this set of assumptions, we are able to show, see \Cref{sec:appendixdpp}, that for any $\Cc\in\Cf$
\[
\Ec(\Cc)=\big\{(a^\star(t,X_{\cdot\wedge t},Y_t(\Cc),Z_t(\Cc),\pi_t))_{t\in[0,T]}\big\},
\]

where the processes $\big(Y(\Cc),Z(\Cc)\big)$ come from the solution to the following infinite family of BSDEs which for any $s\in[0,T]$ satisfies, $\P$--a.s.
\begin{align}\label{HJB0:general}
Y_t(\Cc)&= \Ur_\Ar(T,\xi)+\int_t^T\! \Big( H_r\big(Y_r(\Cc),Z_r(\Cc),\pi_r\big)-\partial Y_r^r(\Cc) \Big) \d r-\int_t^T\!  Z_r(\Cc) \cdot  \mathrm{d} X_r, \; t\in [0,T],\notag \\
 Y_t^s (\Cc)&= \Ur_\Ar(s, \xi) +\int_t^T  h_r^\star\big(s,Y_r^s (\Cc) , Z_r^s(\Cc),Y_r(\Cc) , Z_r(\Cc),\pi_r \big) \d r-\int_t^T \! Z_r^s(\Cc) \cdot \d X_r, \; t\in [0,T],\\
\partial Y_t^s(\Cc)&=\partial_s \Ur_\Ar(s, \xi ) +\int_t^T\! \nabla h_r^\star\big(s,\partial Y_r^s(\Cc),\partial Z_r^s(\Cc), Y_r^s(\Cc), Z_r^s(\Cc),Y_r(\Cc),Z_r(\Cc),\pi_r\big) \mathrm{d} r-\int_t^T\! \partial Z_r^s(\Cc) \cdot \mathrm{d} X_r,\; t\in [0,T].\notag
\end{align}

Moreover, we have that 
\begin{align}\label{eq:identityHJB0}
\Vr_t^\Ar\big(\Cc,a^\star(\cdot,X_{\cdot},Y_\cdot(\Cc),Z_\cdot(\Cc),\pi_\cdot)\big)=Y_t(\Cc) =Y_t^t(\Cc),\; t\in[0,T],\;\P\as,\;  Z_t(\Cc)=Z^t_t(\Cc),\; \d t\otimes \d \P\ae
\end{align}
Given that \Cref{assump:datareward} guarantees that ${\rm x}\longmapsto \Ur_\Ar(s,{\rm x})$ invertible for every $s\in [0,T]$, we also have that
\begin{align}\label{eq:constraintHJB0}
\Ur_\Ar^{(-1)}\big(s,Y^s_T(\Cc)\big)=\xi=\Ur_\Ar^{(-1)}\big(u,Y^u_T(\Cc)\big),\;\P \as,\; (s,u)\in [0,T]^2.
\end{align}

\begin{remark}\begin{enumerate}[label=$(\roman*)$, ref=.$(\roman*)$,wide,  labelindent=0pt]\label{rmk:correspondencesystemandbsvie}
\item We recall that the diagonal process $(Z_t^t)_{t\in [0,T]}$ is well-defined for elements in $\overline{\H}^{_{\raisebox{-1pt}{$ \scriptstyle 2,2$}}}$, see {\rm \Cref{sec:spacesandassupm}}.

\item \label{rmk:correspondencesystemandbsvie:ii} Links between time-inconsistent control problems and a broader class of {\rm BSVIEs} have been identified in the past. 
	The first mention of this link appears, as far as we know, in the concluding remarks of {\rm \citet*{wang2019backward}}. 
	The link was then made rigorous independently by {\rm \cite{hernandez2020me}} and {\rm \citet*{wang2019time}}. In our setting, in light of \eqref{eq:identityHJB0}, such an equation appears as the one satisfied by the reward of the agent along the equilibrium. 
	As such, the pair $\big(Y_t^s(\Cc), Z^s_t(\Cc)\big)_{(s,t)\in[0,T]^\smallfont{2}}$ solves a so-called extended {\rm type-I BSVIE}, which for any $s\in [0,T]$ satisfies 
\begin{align}\label{eq:bsdeVoltgen}
Y_t^s (\Cc)= \Ur_\Ar(s, \xi) +\int_t^T  h_r^\star \big(s,X,Y_r^s (\Cc) , Z_r^s(\Cc),Y_r^r(\Cc) , Z_r^r(\Cc),\pi_r\big) \d r-\int_t^T  Z_r^s(\Cc) \cdot \d X_r, \; t\in [0,T],\; \P\as
\end{align}
We highlight that this {\rm BSVIE} involves the diagonal processes $\big(Y_t^t(\Cc), Z^t_t(\Cc)\big)_{(s,t)\in[0,T]^2}$ and that in light of {\rm \cite[Theorem 4.4]{hernandez2020unified}} the solutions of \eqref{HJB0:general} are in correspondence to those of \eqref{eq:bsdeVoltgen}.\end{enumerate}
\end{remark}

\subsection{The family of restricted contracts}\label{sec:restrictedcontracts}
In light of our previous observation, namely \eqref{eq:constraintHJB0}, we will introduce next a family of restricted terminal payments, which we will denote $\overline \Xi$, and
	$\overline \Cf$ will denote the associated class of contracts. 
	For any contract in this family, we can solve the associated time-inconsistent control problem faced by the agent. 
	Moreover, we will show that any admissible contract available to the principal admits a representation as a contract in $\overline \Cf$. 
	Consequently, the principal's optimal expected utility is not reduced if she restricts herself to offer contracts in this family and optimises.\medskip

In order to define the family of restricted contracts, we introduce next the process $Y^{y_\smallfont{0},Z,\pi}$, which for a suitable process $Z$ will represent the value of the agent. 
	This is a preliminary step based on the observation, see \eqref{eq:constraintHJB0}, that the value of the agent at the terminal time $T$ coincides with the payment offered by the contract. 
	To alleviate the notation let us set $\Ic:=\{ y_0\in \Cc_1([0,T],\R): y_0^0\geq R_0\}$.

\begin{definition}\label{def:indexedcontracts}
Let $\pi\in \Pi$. We denote by $\Hc^{2,2}$ the collection of processes $Z\in \overline{\H}^{_{\raisebox{-1pt}{$ \scriptstyle 2,2$}}}$ satisfying $\|Y^{y_\smallfont{0},Z,\pi}\|_{\S^{2,2}}<\infty$, where for $y_0\in \Ic$,  $Y^{y_\smallfont{0},Z,\pi}:=(Y^{s,y_\smallfont{0},Z,\pi})_{s\in [0,T]}$ satisfies for every $s\in [0,T]$,
\begin{align}\label{eq:fsvie}
Y_t^{s,y_\smallfont{0},Z,\pi}=y_0^s-\int_0^t h_r^\star\big(s,X,Y_r^{s,y_\smallfont{0},Z,\pi} , Z_r^s,Y_r^{r,y_\smallfont{0},Z,\pi} , Z_r^r,\pi_r\big) \d r+\int_0^t  Z_r^s  \cdot \d X_r,\; t\in [0,T],\; \P\as
\end{align}

\vspace{-1em}
\begin{align}\label{eq:constraintgen}
\Ur_\Ar^{(-1)}\big(s,Y^{s,y_\smallfont{0},Z,\pi}_T \big)=\Ur_\Ar^{(-1)}\big(u,Y^{u,y_\smallfont{0},Z,\pi}_T \big),\;\P\as,\; (s,u)\in [0,T]^2.
\end{align}
\end{definition}
With this, it is natural to consider the class of contracts $\overline \Cf:=\Pi\times\overline\Xi$ where $\overline \Xi$ denotes the set of terminal payments of the form 
\[
\Ur_\Ar^{(-1)}\big(T,Y^{T,y_\smallfont{0},Z,\pi}_T \big),\; (y_0,Z,\pi)\in  \Ic\times \Hc^{2,2}\times \Pi .
\]

The main novelty of our argument, compared to that in the time-consistent case, is the fact that {\rm \eqref{eq:constraintgen}} imposes a constraint on the elements $Z\in \Hc^{2,2}$.

\begin{remark}
\begin{enumerate}[label=$(\roman*)$, ref=.$(\roman*)$,wide,  labelindent=0pt]
\item We highlight $\Hc^{2,2}$ is independent of the choice of $\pi\in \Pi$ and that establishing $\Hc^{2,2}\neq \emptyset$ is inherently associated with the existence of solutions to \eqref{eq:bsdeVoltgen}. 
	For results on {\rm type-I BSVIEs}  we refer to {\rm \cite{hernandez2020unified}} and {\rm \cite{wang2022backward}}. 
	
\item The process $Y^{y_\smallfont{0},Z,\pi}$ denotes a solution to a so-called forward Volterra integral equation $(${\rm FSVIE}, for short$)$. 
	%These objects have been studied since {\rm \citet*{berger1980volterra,berger1980volterra2}, \citet*{protter1985volterra}, \citet*{pardoux1990adapted}}. 
	However, this is not a classic {\rm FSVIE} in the sense that, in addition to $Y^{s,y_\smallfont{0},Z,\pi}$, the diagonal processes $\big(Y_t^{t,y_\smallfont{0},Z,\pi}\big)_{t\in[0,T]}$ appears in the generator. 
	For completeness, {\rm \Cref{sec:appforwardVolterra}} includes a suitable well-posedness result.

\item As mentioned at the beginning of this section, we chose to work with a representation for the agent's value as opposed to the value of the contract itself. 
	This determines the form of the terminal payments in the definition of $\overline \Xi$ and provides a quite general and comprehensive approach. 
	For instance, one could have chosen to represent the value of $\xi$ directly for an agent with a time-inconsistent exponential utility. 
	This would have produced a version of \eqref{HJB0:general} whose generators have quadratic growth in $Z$ and whose analysis is more delicate than in the Lipschitz case. 
	See for instance, %{\rm \citet*{briand2008quadratic}, \citet*{kobylanski2000backward}} and {\rm\citet*{tevzadze2008solvability}} for the study of quadratic {\rm BSDEs}, and 
	{\rm \citet*{wang2018recursive}}, {\rm \citet*{fan2022multidimensional}}, {\rm \citet*{hernandez2020quadratic}} for the study of quadratic {\rm BSVIEs}. 
	We recall that taking that approach in the time-consistent scenario requires, at the very least, assuming the contracts have exponential moments of sufficiently large order. 
	Our approach prevents this given our growth assumptions in {\rm \Cref{assump:datareward}}. 
	However, one cannot expect to avoid such restrictions for problems that are inherently quadratic.
\end{enumerate}
\end{remark}

\begin{remark}\label{rmk:targetconstraint}
\textcolor{black}{We would like to highlight the nature of the constraint \eqref{eq:constraintgen}. Indeed, for any $\bar\xi\in \overline \Xi$ satisfying \eqref{eq:constraintgen}, it holds that $\bar\xi=\Ur_\Ar^{(-1)}\big(s,Y^{s,y_\smallfont{0},Z,\pi}_T \big)$, $\P\as$, $s\in [0,T]$. That is, if we let $\mathfrak S$ denote the family of continuously differentiable functions ${\it \gamma}:[0,T]\longmapsto \R$ such that the map $s\longmapsto \Ur_\Ar^{(-1)}\big(s,{\it \gamma}(s) \big)$ is constant, \eqref{eq:constraintgen} is equivalent to the stochastic target constraint}
\begin{align}\label{eq:targetconstraint}
Y^{s,y_\smallfont{0},Z,\pi}_T\in \mathfrak{S}, \; \P\text{\rm--a.s.}
\end{align}
\textcolor{black}{Moreover, we emphasise that this constraint is there due to time-inconsistency. Indeed, going back to the time-consistent, \emph{i.e.} exponential discounting, scenario presented in {\rm\Cref{sec:agentprobstatement}}, it is not hard to see that $\e^{-\rho s} Y_t^{s,\alpha}=\e^{-\rho u} Y_t^{u,\alpha}$, $\P\as$, for any $(u,s,t)\in [0,T]^3$. Thus, \eqref{eq:constraintgen} as well as the stochastic target constraint \eqref{eq:targetconstraint} are automatically fulfilled in the time-consistent, exponential discounting, scenario.}
\end{remark}
%\todo[inline]{Shouldn't we make the comments about stochastic target constraints here? I think this would be the right place. We can introduce the set $\mathfrak S$ of continuous (or $C^1$ if this is the regularity we have for $Y^s$) functions $f$ on $[0,T]$ such that the map $s\longmapsto \Ur_\Ar^{(-1)}\big(s,f(s) \big)$ is constant, and say that the target constraint is $Y^{s,y_\smallfont{0},Z,\pi}_T\in \mathfrak{S}$, $\P${\rm--a.s.}}

In light of our previous remarks, as a preliminary step, we must verify that \eqref{eq:fsvie} uniquely defines $Y^{y_\smallfont{0},Z,\pi}$. 
	At the formal level, the following auxiliary lemma says that the integrability conditions on the pair $(\pi,Z)$ guarantees this.

\begin{lemma}\label{lemma:solfsvie}
Let {\rm \Cref{assump:datareward}} and {\rm  \Cref{assump:uniquemax}} hold. Given $(\pi,y_0,Z)\in \Pi\times \Ic \times\Hc^{2,2}$ there exist unique processes $(Y^{y_\smallfont{0},Z,\pi},\partial Y^{y_\smallfont{0},Z,\pi})\in \S^{2,2}\times \S^{2,2}$ such that $Y^{y_\smallfont{0},Z,\pi}$ satisfies \eqref{eq:fsvie} and $\partial Y^{y_\smallfont{0},Z,\pi}$ satisfies
\begin{align}\label{eq:fsviepartial}
\partial Y _t^{s,y_\smallfont{0},Z,\pi}=\partial y_0^s-\int_0^t \nabla h_r^\star\big(s,X,\partial Y_r^{s,y_\smallfont{0},Z,\pi} , \partial Z_r^s, Y_r^{s,y_\smallfont{0},Z,\pi} , Z_r^s,Y_r^{r,y_\smallfont{0},Z,\pi} , Z_r^r,\pi_r\big) \d r+\int_0^t  \partial Z_r^s  \cdot \d X_r.
\end{align}
 
\begin{proof}
Let us first argue the result for $Y^{y_\smallfont{0},Z,\pi}$. Note that the integrability of $(\pi,Z)\in \Pi\times\Hc^{2,2}$, \Cref{assump:datareward}\ref{assump:datareward:i} and \Cref{assump:uniquemax}\ref{assump:uniquemax:iii} yields
\begin{align*}
\sup_{s\in [0,T]}\E\bigg[\bigg(\int_0^T | h_r^\star\big(s,X,0,Z_r^s,0,Z_r^r,\pi_r)\big) | \d r\bigg)^2\bigg]<\infty.
\end{align*}
The result follows from {\rm \Cref{prop:wpfsvie}}. The second part of the statement is a consequence of {\rm \Cref{prop:wpfsviepartial}} and the integrability of $\pi\in \Pi$.
\end{proof}
\end{lemma}

We are now ready to state our main result, in words it guarantees that there is no loss of generality for the principal in offering contracts of the form given by $\overline \Xi$.

\begin{theorem}\label{thm:repcontractgeneral}
\begin{enumerate}[label=$(\roman*)$, ref=.$(\roman*)$,wide,  labelindent=0pt]
\item We have $\overline \Cf=\Cf$. Moreover, for any contract $\Cc=(\pi,\xi)\in \overline \Cf$, with $\xi$ associated to $(y_0,Z)\in \Ic \times  \Hc^{2,2}$, we have
\[ \Ec(\Cc)=\big\{ a^\star\big(t,X_{\cdot\wedge t},Y_t^{t,y_\smallfont{0},Z,\pi},Z_t^t,\pi_t \big)_{t\in [0,T]}\big\},\; \Vr_0^\Ar(\Cc)=y_0^{0}.\]

\item Let $\P^{\star}(Z,\pi):=\P^{a^\smallfont{\star}(\cdot,X_\smallfont{\cdot},Y_\smallfont{\cdot}^{\smallfont{\cdot}\smallfont{,}\smallfont{y}_\tinyfont{0}\smallfont{,}\smallfont{Z}\smallfont{,}\smallfont{\pi}},Z_\smallfont{\cdot}^\smallfont{\cdot},\pi_\smallfont{\cdot})}$. The problem of the principal admits the following representation
\begin{align}\label{eq:reductionprincipal}
{\rm V}^{\rm P}=\sup_{\substack{ y_\smallfont{0}\in \Ic }}\; \underline V(y_0),
\end{align}
where
\[
\underline V(y_0):= \sup_{(Z,\pi)\in \Hc^{\smallfont{2}\smallfont{,}\smallfont{2}}\times \Pi}\E^{\P^\smallfont{\star}(Z,\pi)}\bigg[{\rm U_P}\bigg( X_{\cdot \wedge T}, \Ur_\Ar^{(-1)}\Big(T,Y^{T,y_\smallfont{0},Z,\pi}_T \Big) \bigg)-\int_0^T c_r^{\rm p}(X_{\cdot\wedge r},\pi_r)\d r \bigg].
\]
\end{enumerate}
\begin{proof}
We first argue $\Cf\subseteq \overline \Cf$. Let $\Cc\in \Cf$. 
	In light of \Cref{assump:datareward}, the fact that $\pi\in \Pi$, and \Cref{rmk:correspondencesystemandbsvie}\ref{rmk:correspondencesystemandbsvie:ii}, \Cref{thm:necessity} guarantees that for $\Cc\in \Cf$ there exists $(Y(\Cc),Z(\Cc))\in \S^{2,2}\times \overline{\H}^{_{\raisebox{-1pt}{$ \scriptstyle 2,2$}}}$ solution to \eqref{eq:bsdeVoltgen} and a process $\partial Y(\Cc)\in \S^{2,2}$ satisfying that the mapping $([0,T],\Bc([0,T]))\longrightarrow (\S^2,\|\cdot \|_{\S^2}):s\longmapsto \partial Y^s(\Cc)$ is the derivative of $([0,T],\Bc([0,T]))\longrightarrow (\S^2,\|\cdot \|_{\S^2}):s\longmapsto Y^s(\Cc)$. 
	We also note that \Cref{assump:datareward}\ref{assump:datareward:0} guarantees \eqref{eq:constraintgen} holds.
	Moreover, \eqref{eq:identityHJB0} implies $Y_0(\Cc)=\Vr_0^\Ar\big(\Cc,a^\star(\cdot,X_{\cdot},Y_\cdot(\Cc),Z_\cdot(\Cc),\pi_\cdot)\big) \geq R_0$, recall $\Cc\in \Cf$. From this, taking $y_0(\Cc)=Y_0(\Cc)$ we have that $(y_0(\Cc),Z(\Cc))\in \Ic \times \Hc^{2,2}$. 
	Thus $\Cc\in \overline \Cf$.
\medskip

%\textcolor{magenta}{HERE: the new family should have the appropriate integrability in $\pi$}\medskip

To show the reverse inclusion, let $\bar\Cc=(\pi,\bar\xi)\in \Pi\times \overline \Xi$.  This is, $\bar \xi=\Ur_\Ar^{(-1)}(T,Y^{T,y_\smallfont{0},Z,\pi}_T )$, where, in light of \Cref{lemma:solfsvie}, $Y^{y_\smallfont{0},Z,\pi}$ denotes the process, induced by $(y_0,Z,\pi)\in\Ic\times \Hc^{2,2}\times \Pi$, such that $\|Y^{y_\smallfont{0},Z,\pi}\|_{\S^{\smallfont{2}\smallfont{,}\smallfont{2}}}<\infty$ and \eqref{eq:constraintgen} holds. In particular
\[ Y_T^{s,y_\smallfont{0},Z,\pi}=\Ur_\Ar\big(s,\bar \xi),\; \P\as,\; s\in [0,T].\]
Therefore, for any $s\in [0,T]$
\begin{align}\label{eq:bsvieauxprop}
Y_t^{s,y_\smallfont{0},Z,\pi} = \Ur_\Ar(s, \bar \xi) +\int_t^T  h_r^\star\big(s,X ,Y_r^{s,y_\smallfont{0},Z,\pi}  , Z_r^s,Y_r^{r,y_\smallfont{0},Z,\pi} , Z_r^r,\pi_r\big) \d r-\int_t^T  Z_r^s \cdot \d X_r, \;t\in [0,T],\;\P\as
\end{align}

We now show $\bar \xi\in \Xi$, see \Cref{sec:agentprobstatement}. It is immediate to see that
\begin{align*}
\|\Ur_\Ar(\cdot,\bar \xi)\|_{\Lc^{\smallfont{2}\smallfont{,}\smallfont{2}}}^2=&\sup_{s\in [0,T]}\E\Big[\big|\Ur_\Ar(s, \bar \xi)\big|^2\Big]=\sup_{s\in [0,T]} \E\Big[\big|Y_T^{s,y_\smallfont{0},Z,\pi}\big|^2\Big]\leq\|Y^{y_\smallfont{0},Z,\pi}\|_{\S^{\smallfont{2}\smallfont{,}\smallfont{2}}}^2<\infty.
\end{align*}
Now, given $Y^{y_0,Z,\pi}$ solution to \eqref{eq:fsvie} and $\partial Z$ by definition of $Z\in\overline{\H}^{_{\raisebox{-1pt}{$ \scriptstyle 2,2$}}}$, \Cref{lemma:solfsvie} guarantees there exists $\partial Y^{y_0,Z,\pi}\in \S^{2,2}$ such that the pair $(\partial Y^{Z}, \partial Z)$ satisfies \eqref{eq:fsviepartial}.  Moreover, by \Cref{prop:wpfsviepartial} for any $s\in [0,T]$
\[
\partial Y^{s,Z,\pi}_T=\partial_s \Ur^\Ar(s,\bar \xi), \; \P\as
\]
Thus, $ \|\partial_s \Ur_\Ar(\cdot, \bar \xi)\|^2_{\Lc^{\smallfont{2}\smallfont{,}\smallfont{2}}}\leq\|\partial Y^{y_0,Z,\pi}\|_{\S^{\smallfont{2}\smallfont{,}\smallfont{2}}}^2<\infty$. This shows $\bar \xi\in \Xi$. \medskip

Let us argue $\bar \Cc\in \Cf_o$ as in \Cref{def:contractswithoneeqvalue}, \emph{i.e.} that $\bar \Cc$ leads to a unique equilibrium.
	% any $\alpha^\star\in \Ec(\bar \xi)$ provide the agent the same value.
	In light of \Cref{assump:uniquemax}, \Cref{thm:necessity} and \Cref{thm:verification}, it suffices to establish $\bar \Cc$ leads to a solution of \eqref{HJB0:general}.
	Let us recall that by \cite[Theorem 4.4]{hernandez2020unified}, the solutions of \eqref{HJB0:general} are in correspondence to those of \eqref{eq:bsdeVoltgen}.
	We now simply note that \eqref{eq:bsvieauxprop} defines a solution.
	%In light of \Cref{assump:datareward}, $\bar \xi\in \Xi$, and $\pi\in \Pi$ 
%\begin{align*}
%\|\tilde h\|_{\mathbb{L}^{1,2,2}}^2=& \sup_{s\in [0,T]}\E\bigg[\bigg(\int_0^T \big| h^\star_r\big(s,X_{\cdot\wedge r},0,0,0,0\big) \big| \d r\bigg)^2\bigg]<\infty,\\
%\|\nabla \tilde h\|_{\mathbb{L}^{1,2,2}}^2=& \sup_{s\in [0,T]}\E\bigg[\bigg(\int_0^T \big| \nabla h_r^\star\big(s,X_{\cdot\wedge r},0,0,0,0,0,0\big) \big| \d r\bigg)^2\bigg]<\infty,
%\end{align*}
%and \cite[Theorem 3.5]{hernandez2020unified}, \eqref{HJB0:general} is well-posed with solution $\big((Y^{t,y_0,Z}_t)_{t\in [0,T]},(Z_t^t)_{t\in [0,T]},Y^{y_0,Z},Z,\partial Y^{y_0,Z}, \partial Z\big)$. 
Thus
\[
\Ec(\bar \Cc)= \big\{a^\star\big(t,X_{\cdot\wedge t},Y_t^{t,y_\smallfont{0},Z,\pi},Z_t^t,\pi_t\big)_{t\in [0,T]}\big\}.
\]

To conclude $\bar \Cc\in \Cf$, note that by \Cref{thm:verification}, $\Vr^\Ar_0(\bar \Cc)=y_0^{0},$ so that $y_0^0\geq R_0$ guarantees the participation constraint is satisfied.
\end{proof}
\end{theorem}

In view of \Cref{thm:repcontractgeneral}, the problem of the principal involves controlling, via $(\pi,Z)\in \Pi\times\Hc^{2,2}$, the processes $(X,Y^{y_\smallfont{0},Z,\pi})$. The dynamics of $X$ are given, in weak formulation, by
\begin{align}\label{eq:dynprincipal}
X_t =  x_0 +\int_0^t \sigma_r(X_{\cdot \wedge r})\Big( b^\star_r\big(X_{\cdot\wedge r},Y_r^{r,y_0,Z,\pi},Z_r^r,\pi_r\big)  \d r +   \d B^{\star}_r\Big),\; t\in [0,T],\; \P\as
\end{align}
where $B^\star:=B-\int_0^\cdot b_r^\star\big(X_{\cdot\wedge r},Y_r^{r,y_\smallfont{0},Z,\pi},Z_r^r,\pi_r\big)\mathrm{d}r$ is a $\P^{\star}(Z,\pi)$--Brownian motion, and those of $Y^{y_\smallfont{0},Z,\pi}$ are given by
\begin{align*}
Y_t^{s,y_0,Z,\pi}&=y_0^s-\int_0^t h_r^\star\big(s,X ,Y_r^{s,y_0,Z,\pi} , Z_r^s,Y_r^{r,y_0,Z,\pi} , Z_r^r,\pi_r\big) \d r+\int_0^t  Z_r^s  \cdot \d X_r,\; t\in [0,T],\; \P\as,\; s\in [0,T].
\end{align*}

{\color{black} We highlight that on top of the Volterra nature of both the state process $Y^{y_\smallfont{0},Z,\pi}$ and the control $Z$, the constraint \eqref{eq:constraintgen} must be satisfied. 
	However, building upon the discussion in \Cref{rmk:targetconstraint}, we see that the problem of the principal corresponds to a stochastic target control problem of FSVIEs with Volterra controls. 
	Indeed, the principal
	\begin{enumerate}[label=$(\roman*)$, ref=.$(\roman*)$,wide]
	\item controls the forward Volterra process $(Y_t^{s,y_\smallfont{0},Z,\pi})_{(s,t)\in [0,T]^2}$; 
	\item with Volterra-type controls $(Z_t^s)_{(s,t)\in [0,T]^2}$, recall both $(Z_t^t)_{t\in [0,T]}$ and $(Z_t^s)_{t\in [0,T]}$ impacts the dynamics;
	\item  the state process $Y^{y_\smallfont{0},Z,\pi}$ is subject to the stochastic target constraint \eqref{eq:targetconstraint}.
	\end{enumerate}

The literature on controlled FSVIEs began, to the best of our knowledge, with \citet*{chen2007linear} where the authors studied the control of FSVIE by means of a stochastic maximum principle.\footnote{Ever since, several works have extended this approach, a probably incomplete list includes \citet*{shi2015optimal}, \citet*{wang2018linear} and \citet*{hamaguchi2022linear2}.} 
	A recent milestone in the study of this problem is \citet*{viens2019martingale} where, via a dynamic programming approach, the authors arrive at a path-dependent HJB equation. 
	Nevertheless, in all of these works the control consists of an unconstrained stochastic process. 
	Thus the approach \cite{viens2019martingale} is inoperable as it does not cover $(ii)$ nor $(iii)$ above.\medskip

	Regarding the study of stochastic target control problems, the seminal works are due to \citet*{soner2002stochastic,soner2009dynamic} where the state process is a controlled SDEs. 
	We also remark on the recent extension to targets in the Wasserstein space by \citet*{bouchard2017quenched}, which shows the possibility of extending the original approach to infinite dimensional target problems like the one faced by the principal, namely $(iii)$ above. 
	Particularly important to our analysis are the results in \citet*{bouchard2010optimal} on optimal control problems with stochastic target constraints. 
	Indeed, this work elucidates the blueprint that needs to be extended to the Volterra case to be able to obtain (infinite-dimensional) HJB-type PDEs that characterise the problem of the principal. 
	As the reader might be able to notice, in general, this seems to be quite a challenging task. 
	Therefore, we will, for now, concentrate our attention on simpler cases where we can actually transfer the stochastic target constraint on $Y^{y_\smallfont{0},Z,\pi}$ into a more manageable constraint on the controls $Z$ directly. 
	The general case will be the subject of future research and will be studied in a separate paper.\medskip

As a motivation for our approach in the following examples, we recall that: 
	first, the flow of continuous payments $(\pi_t)_{t\in [0,T]}$ enters the reduced problem of the principal as a standard control on the drift which raises no major challenges in the analysis, and thus we will omit it from the following examples and consider contracts consisting of only a terminal payment, \emph{i.e.} $\Cc=\xi$. 
	Second, for classic separable utilities with exponential discounting it is known, see \Cref{rmk:targetconstraint}, that the Volterra nature of the state process $Y^{y_\smallfont{0},Z,\pi}$ becomes redundant. 
	Indeed, in this scenario is sufficient to describe $(Y_t^t)_{t\in[0,T]}$ to characterise the entire family. 
	This motivates the study of $\Hc^{2,2}$ under particular specifications of utility functions for both the agent and the principal, hoping to be able to

\medskip
 \quad$(a)$ reduce the complexity of the set $\Hc^{2,2}$; 
 
 \medskip
\quad $(b)$ exploit its particular structure to formulate an ansatz to the problem of the principal. 
 
 \medskip
 This is exactly what we do in the following sections.}
\begin{comment}
{\color{magenta} Ignore this\medskip

And letting $g(x,y)=\Ur_\Pr\big( x-\Ur_\Ar^{(-1)} (T,y  )\big)$
\[
\Yc_t^{Z}:=\E^{\P^{\star,Z}}\bigg[ \Ur_\Pr\bigg( X_T-\Ur_\Ar^{(-1)}\Big(T,Y_T^{T,Z}\Big)\bigg)\bigg|\Fc_t\bigg]
\]
we might guess
\begin{align*}
\Yc_t^Z=g\Big(X_T,Y_T^{T,Z}\Big)+\int_t^T b_r^\star(X_{\cdot\wedge r},Y_r^{r,Z},Z_r^r)  \Zc_r^Z \d r-\int_t^T \Zc_r^Z\cdot \d X_r, \; \P\as
\end{align*}
  .. the second is a forward SVE... so the problem of the principal is, in general, one in control of a constrained FSVE\medskip
\todo[inline]{Y is non-Markovian, because of the two-variable. Therefore the value function of $\Vr^\Pr$ can be represented by a BSDE.  This BSDE is not Markovian as soon as $Y$ is not a Markov process}
}

\end{comment}

\section{The second-best problem: examples}\label{sec:sbexamples}

\subsection{Agent with discounted utility reward}\label{sec:ra1}
As an initial example, let us consider the scenario in \Cref{sec:raFB} under the additional choice $g= 1$, which implies $K^{s,\alpha}_{t,T} $ does not depend on $s\in [0,T]$. Thus, we have
\begin{align}\label{eq:ra1reward}
{\Vr}^{\Ar}_t(\xi,\alpha)=\E^{\P^\alpha}\Big [f(T-t) \Ur_\Ar^o\big(\xi- K_{t,T}^\alpha \big) \Big | \Fc_t \Big], \;  K_{t,T}^\alpha:=\int_t^T k_r^o(X_{\cdot\wedge r},\alpha_r)\d r,\; (t,\alpha,\xi)\in [0,T]\times \Ac\times \Cf.
\end{align}

Under this specification, \eqref{HJB0:general} reduces significantly. Indeed
\begin{align*}
h_t(s,x,y,z,a)&=\sigma_t(x)b_t(x,a)\cdot z+\gamma_\Ar k_t^o(x,a)y,\; \nabla h_t(s,x,u,v,y,z,a)=\sigma_t(x)b_t(x,a)\cdot v+\gamma_\Ar k_t^o(x,a)u,
\end{align*}
\[\Ur_\Ar^o(s,{\rm x})=f(T-s)\Ur_\Ar^o({\rm x}),\; \partial_s \Ur_\Ar^o(s,{\rm x})=-f^\prime (T-s)\Ur_\Ar^o({\rm x}).
\]

\begin{remark}
\begin{enumerate}[label=$(\roman*)$, ref=.$(\roman*)$,wide,  labelindent=0pt]
\item We highlight that the absence of accumulative cost in the agent's reward functional, \emph{i.e.} $c=0$, together with the choice $g=1$ makes the driver in the second family of {\rm BSDEs} independent of the variable $s$, \emph{i.e.} $\nabla h=h$. Moreover, it coincides with the functional maximised in the Hamiltonian $H$. 

\item We remark that in this scenario, the non-exponential discount factor, \emph{i.e.} the time-inconsistent preferences, does not add much to the problem. Even though the agent's continuation utility changes by a factor, the optimal/equilibrium control state pair coincides for both problems. Our aim in presenting it is to illustrate how the technique presented in {\rm \Cref{sec:restrictedcontracts}} is compatible with the results known in the case of a time-consistent agent.
\end{enumerate}
\end{remark}

The next result provides a drastic simplification of the infinite dimensional system introduced in {\rm \Cref{sec:system}}. This is due to the particular form of the reward of the agent \eqref{eq:ra1reward}. \medskip

\begin{comment}
 to the following system, which for any $s\in[0,T]$ holds $\P$--a.s. for all $t\in[0,T]$
\begin{align}\label{eq:HJB:ra1}
\begin{split}
Y_t(\xi) &={\rm U_A}(\xi)+\int_t^T\Big( H_r\big(X_{\cdot\wedge r}, Y_r (\xi) ,  Z_r (\xi) \big)-\partial   Y_r^r (\xi) \Big) \mathrm{d} r-\int_t^T    Z_r (\xi) \cdot  \mathrm{d} X_r,\\
 Y_t^s (\xi)&=f(T-s)\Ur(\xi)+\int_t^T h_r(X_{\cdot\wedge r}, Y_r^s (\xi),  Z_r^s (\xi), a^\star(r,X_{\cdot\wedge r},  Y_r (\xi),  Z_r (\xi)))\d r-\int_t^T   Z_r^s (\xi) \cdot \d X_r,\\
\partial   Y_t^s (\xi) &=-f^\prime(T-s)\mathrm{U_A}(\xi)+\int_t^T h_r \big(X_{\cdot\wedge r},\partial Y_r^s (\xi) , \partial   Z_r^s (\xi) ,a^\star(r,X_{\cdot\wedge r},   Y_r (\xi),  Z_r (\xi))\big) \mathrm{d} r-\int_t^T\partial  Z_r^s (\xi) \cdot \mathrm{d} X_r,
\end{split}
\end{align}
\end{comment}

\begin{lemma}\label{lemma:ra1agentreduction}
\begin{enumerate}[label=$(\roman*)$, ref=.$(\roman*)$,wide,  labelindent=0pt]
\item Let $\xi \in \Cf$ and the agent's reward be given by \eqref{eq:ra1reward}. Then, \eqref{HJB0:general} is equivalent to the {\rm BSDE}
\begin{align*}
 Y_t (\xi)&=\Ur_\Ar^o(\xi)+\int_t^T\bigg( H_r\big(X_{\cdot\wedge r}, Y_r (\xi), Z_r (\xi)\big)+\frac{f'(T-r)}{f(T-r)}  Y_r (\xi)\bigg) \mathrm{d} r-\int_t^T   Z_r (\xi)\cdot  \mathrm{d} X_r,\; t\in[0,T], \;\P\as
\end{align*}

\item Let $Z\in \Hc^{2,2}$. Then $(Y_t^{t,y_\smallfont{0},Z})_{t\in [0,T]}$ solves the {\rm BSDE}
\begin{align*}
 Y_t^{t,y_\smallfont{0},Z}=Y_T^{T,y_\smallfont{0},Z}+\int_t^T\bigg( H_r\big(X_{\cdot\wedge r}, Y_r^{r,y_\smallfont{0},Z}, Z_r^r\big)+\frac{f'(T-r)}{f(T-r)}  Y_r^{r,y_\smallfont{0},Z}\bigg) \mathrm{d} r-\int_t^T   Z_r^r\cdot  \mathrm{d} X_r,\; t\in[0,T], \;\P\as
\end{align*}

\item $\Hc^{2,2}=\Hc^2$, where $\Hc^2$ denotes the family of $Z\in \H^2$ satisfying $\|Y^{y_\smallfont{0},Z}\|_{\S^2}<\infty$ where, for any $y_0\in(R_0,\infty)$,
\[
Y_t^{y_\smallfont{0},Z}=y_0-\int_0^t \bigg( H_r\big(X_{\cdot\wedge r}, Y_r^{y_\smallfont{0},Z}, Z_r\big)+\frac{f'(T-r)}{f(T-r)}  Y_r^{y_\smallfont{0},Z}\bigg) \mathrm{d} r+\int_0^t   Z_r\cdot  \mathrm{d} X_r,\; t\in[0,T], \;\P\as
\]

\item $\overline \Cf=\big\{{\Ur_\Ar^o}^{(-1)}(Y_T^{y_\smallfont{0},Z}): (y_0,Z)\in (R_0,\infty)\times \Hc^2\big \}$. Moreover, for any $\xi\in \overline \Cf$
\[\Ec(\xi)=\big \{(a^\star(t,X_{\cdot\wedge t}, Y_t^{y_\smallfont{0},Z},Z_t)_{t\in [0,T]}\big\},\; \Vr_0^\Ar(\xi)=y_0.\]

\end{enumerate}

\begin{proof}
It is immediate from \eqref{HJB0:general} that, $\P\as$
\begin{align*}
Y_t^s (\xi) =\E^{\P^\smallfont{\star}(\xi) }\Big[ f(T-s) \Ur_\Ar^o\big(\xi - K_{t,T}^{a^\star }\big)\Big| \Fc_t\Big], \; \partial  Y_t^s (\xi) =-\E^{\P^\smallfont{\star}(\xi) }\Big[ f^\prime(T-s) \Ur_\Ar^o\big(\xi -K_{t,T}^{a^\star}\big)\Big| \Fc_t\Big],\text{ and }  Y_t^t (\xi) = Y_t (\xi).
\end{align*}
Thus
\begin{align*}
Y_t^s(\xi)&=\frac{f(T-s)}{f(T-t)}Y_t(\xi), \; \P\as , \;s\in [0,T],\; \partial  Y_t^t (\xi) =- \frac{f^\prime(T-t)}{f(T-t)}  Y_t^t (\xi) =-\frac{f^\prime(T-t)}{f(T-t)}  Y_t (\xi), \; \P\as ,
\end{align*}
and, for any $(s,u)\in [0,T]^2$
\begin{align*}
\Ur_\Ar^{(-1)}\big(s,Y^{s}_T(\xi) \big)= {\Ur_\Ar^o}^{(-1)}\bigg(\frac{Y_T^{s}(\xi)}{f(T-s)}\bigg)=\xi={\Ur_\Ar^o}^{(-1)}\bigg(\frac{Y_T^{u}(\xi)}{f(T-u)}\bigg)=\Ur_\Ar^{(-1)}\big(u,Y^{u}_T(\xi) \big),\;\P\as
\end{align*}

All together, this shows that \eqref{HJB0:general} reduces to the equation in the statement. The result then follows as we can trace back the argument and construct a solution to \eqref{HJB0:general} starting from a solution to the BSDE in the statement.\medskip

We now argue $(ii)$. Let $Z\in \Hc^{2,2}$. Then, there is $(y_0^s)_{s\in [0,T]}$ such that \eqref{eq:constraintgen} holds and
\[
Y_t^{s,y_\smallfont{0},Z}=y_0^s-\int_0^t h_r^\star\big(s,X_{\cdot\wedge r},Y_r^{s,y_\smallfont{0},Z} , Z_r^s,Y_r^{r,y_\smallfont{0},Z} , Z_r^r\big) \d r+\int_0^t  Z_r^s  \cdot \d X_r.
\]
Let us note that \eqref{eq:constraintgen} implies $Y_T^{s,y_\smallfont{0},Z}=f(T-s) Y_T^{T,y_\smallfont{0},Z}$. Since $h_t(s,x,y,z,a)=h_t(u,x,y,z,a), (s,u)\in [0,T]^2$, we obtain
\begin{align*}
Y_t^{s,y_\smallfont{0},Z}=f(T-s) Y_T^{T,y_\smallfont{0},Z}+\int_t^T h_r^\star\big(r,X_{\cdot\wedge r},Y_r^{s,y_\smallfont{0},Z} , Z_r^s,Y_r^{r,y_\smallfont{0},Z} , Z_r^r\big) \d r-\int_t^T  Z_r^s  \cdot \d X_r,
\end{align*}
so that 
\begin{align*}
Y_t^{s,y_\smallfont{0},Z}=\E^{\P^\smallfont{\star}(Z)}\bigg[f(T-s) Y_T^{T,y_\smallfont{0},Z}\exp\bigg(\gamma_\Ar \int_t^T {k^o_r}^\star (X_{\cdot\wedge r},Y_r^{r,y_\smallfont{0},Z}, Z^r_r)\d r \bigg)\bigg|\Fc_t\bigg],\;  \partial Y_t^{t,y_\smallfont{0},Z}=-\frac{f^\prime(T-t)}{f(T-t)}Y_t^{t,y_\smallfont{0},Z}.
\end{align*}
Note that $(Y_t^{t,y_\smallfont{0},Z})_{t\in [0,T]}\in \S^{2}$. Thanks to \Cref{thm:repcontractgeneral}, the result follows replacing $\partial Y_t^{t,y_\smallfont{0},Z}$ in the first equation of \eqref{HJB0:general}.\medskip

We are left to argue $(iii)$ as $(iv)$ is argued as in \Cref{thm:repcontractgeneral}. $\Hc^{2,2}\subseteq \Hc^2$ follows by $(ii)$. Indeed, there is $y_0^0$ such that
\begin{align*}
 Y_t^{t,y_\smallfont{0},Z}=y_0^0-\int_0^t \bigg( H_r\big(X_{\cdot\wedge r}, Y_r^{r,y_\smallfont{0},Z}, Z_r^r\big)+\frac{f'(T-r)}{f(T-r)}  Y_r^{r,y_\smallfont{0},Z}\bigg) \mathrm{d} r+\int_0^t   Z_r^r\cdot  \mathrm{d} X_r.
\end{align*}
Conversely, let $(y_0,Z)\in (R_0,\infty)\times \Hc^2$ and $Y^{y_\smallfont{0},Z}\in \S^2$ as in the statement. Then, letting
\begin{align*}
Y_t^{s,y_\smallfont{0},Z}& :=\frac{f(T-s)}{f(T-t)}Y_t^{y_\smallfont{0},Z}=\E^{\P^\star(Z)}\bigg[f(T-s) Y_T^{y_\smallfont{0},Z}\exp\bigg(\gamma_\Ar \int_t^T {k^o_r}^\star (X_{\cdot\wedge r},Y_r^{y_\smallfont{0},Z}, Z_r)\d r \bigg)\bigg|\Fc_t\bigg],\\
 \partial Y_t^{s,y_\smallfont{0},Z}&:=-\frac{f^\prime(T-s)}{f(T-t)}Y_t^{y_\smallfont{0},Z},
\end{align*}
the martingale representation theorem, which holds in light of \eqref{assump:datadpp} and the integrability of $(Y^{y_\smallfont{0},Z},Z)$, guarantees the existence of $(\tilde Z,\partial \tilde Z)\in \overline \H^{2,2}\times \H^{2,2}$ such that, as elements of $\H^2$,
\[ \tilde Z^s=\tilde Z^0+\int_0^s\partial \tilde Z^r \d r ,\text{ and } \tilde Z_t^t=Z_t, \d t\otimes \d \P\ae\]
It then follows that $(\tilde Z^s)_{s\in [0,T]}\in \Hc^{2,2}$.
\end{proof}
\end{lemma}

\begin{comment}
In this case, we denote

\begin{proposition}\label{prop:repcont:ra1}
Let $\widetilde R_0:= \Ur_\Ar^{(-1)}(R_0)$. $\widetilde \Xi=\Xi$ and for any contract $\xi\in\widetilde\Xi$, associated to a pair $(Y_0,Z)\in [\widetilde R_0,\infty)\times \H^2$ we have
\[
\Ec(\xi)=\big\{\big(a^\star(t,X_{\cdot\wedge t},Z_t)\big)_{t\in[0,T]}\big\},\; {\rm V}^{\rm A}_0(\xi)=\Ur_\Ar(Y_0).
\]
\begin{proof}
Let us first note that in light of \Cref{assump:datareward}, for any bounded $Z\in \H^2(\R^n)$ we have that $\|Y^Z\|_{\S^\infty}<\infty$, \emph{i.e.} $Y^Z$ is bounded. Consequently, it is clear that $\xi=Y_T^Z$ has exponential moments of any order and thus $\xi\in \Cc$. As in light of \Cref{prop:redsystemdiscutility} we are left to study a scalar BSDE, the result is a straightforward consequence of the comparison theorem for BSDEs with quadratic growth, see \cite{briand2008quadratic}. The analysis is similar to that in \cite[Theorem 4.2]{cvitanic2014moral}.
\end{proof}
\end{proposition}
\end{comment}

\subsubsection{Principal's second-best solution}\label{sec:ra1principalsb}
In the following, we will exploit the so-called certainty equivalent, \emph{i.e.} the relation $ \xi={\Ur_\Ar^o}^{(-1)}\big(V_T^{{\rm A}}(\xi, \alpha)\big )$ between the contract and the terminal value of the value function. The benefits of this are twofold: it lays down an expression that can be replaced directly into the principal's criterion, and it removes $Y^{y_\smallfont{0},Z}$ from the generator of the expression representing the contract in exchange for a term which is quadratic in $Z$. For this we need to introduce some extra notation.\medskip

Let $\widehat H:[0,T]\times\Cc([0,T],\R^n)\times\R^n\longrightarrow \R$ be given by
\[
\widehat H_t(x,{\rm z}):=\sup_{a\in A} \widehat h_t(x,{\rm z},a),\; (t,x,{\rm z})\in[0,T]\times\Cc([0,T],\R^n)\times\R^n,
\]
with $\widehat h_t(x,{\rm z},a):=\sigma_t(x)b_t(x,a)\cdot {\rm z}-k_t^o(x,a)$. The mapping $[0,T]\times\Cc([0,T],\R^n) \times \R^n \longmapsto \hat a^\star(t,x,z)\in A$ is defined, as before, by the relation $ \widehat H_t(x,{\rm z})=\widehat h_t(x,{\rm z}, \hat a^\star(t,x,{\rm z}))$, and $ \lambda^\star_t(x,{\rm z})$, $ k^{o\star}_t(x,{\rm z})$ are also defined. 

\begin{proposition}\label{proposition:ra1principalreduction}
The problem of the principal can be represented as the following standard control problem
\begin{align*}
{\rm V}^{\rm P}=\sup_{y_0 \geq R_0}\; \underline V(y_0),\; \text{\rm with}\;
\underline V(y_0)=  \sup_{Z\in \Hc^{\smallfont2}}\E^{\P^\smallfont{\star}(Z)}\Big[\Ur_\Pr\Big( X_{\cdot \wedge T}, \widehat Y_T^{y_\smallfont{0},Z} \Big)\Big],
\end{align*}
where $\P^\star(Z):=\P^{a^\smallfont{\star}(\cdot,X_\smallfont{\cdot},  \hat Z_\smallfont{\cdot})}$ and $\widehat Y_T^{y_\smallfont{0},Z}$ is given by the terminal value of
\begin{align*}
\widehat Y_t^{y_\smallfont{0},Z}&:=-\frac{1}{\gamma_\Ar} \ln\big(-\gamma_\Ar  y_0\big) -\int_0^t \bigg(\widehat H_r(X_{\cdot\wedge r},  \widehat Z_r) -\frac{\gamma_{\rm A}}2 |\sigma_r^\t(X_{\cdot\wedge r}) \widehat  Z_r|^2-\frac{1}{\gamma_\Ar}\frac{f'(T-r)}{ f(T-r)}  \bigg)\d r +\int_0^t \widehat Z_r \cdot \d X_r,\\
\widehat Z_t&:=-\frac{1}{\gamma_\Ar}\frac{  Z_t}{  Y_t^{y_\smallfont{0},Z}}.
\end{align*}
\begin{proof}
We first note that in light \Cref{lemma:ra1agentreduction}, we may replace the optimisation over $\Hc^{2,2}$ with $\Hc^2$. Let $Z\in \Hc^{2}$. The result then follows from \Cref{lemma:ra1agentreduction} by applying It\^o's formula to $\Ur_\Ar^{(-1)}\big( Y_t^{y_{\smallfont0},Z}\big)$.
\end{proof}
\end{proposition}

\begin{remark}\label{rmk:ra1hjbeq}
\begin{enumerate}[label=$(\roman*)$, ref=.$(\roman*)$,wide,  labelindent=0pt]
\item Let us highlight the main message behind {\rm \Cref{proposition:ra1principalreduction}}. When the agent's reward is given by \eqref{eq:ra1reward}, the principal's second-best problem reduces to a standard control problem. This is a drastic simplification of the result in {\rm \Cref{thm:repcontractgeneral}} and a consequence of the particular form of the agent's reward. \medskip

\item In a Markovian setting in which the dependence of the data on the path $X$ is via the current value, we see from the controlled dynamics for $X$ and $\widehat{Y}^{y_\smallfont{0},Z}$ that the problem boils down to computing $\underline V$. Employing the standard dynamic programming approach we obtain that the relevant term for this problem is given for $(t,{\rm x}, {\rm y})\in [0,T]\times \R^n\times \R$ by
\[
\partial_t \underline V(t,{\rm x},{\rm y})+{\rm H}\big(t,{\rm x},{\rm y}, \partial \underline V(t,{\rm x},{\rm y}),\partial^2 \underline V(t,{\rm x},{\rm y})\big)=0,
\]
where
\begin{align*}
{\rm H}(t,{\rm x},{\rm y},p,M)	:=\sup_{z\in \R^n}\bigg\{& \lambda_t^\star(x,z)\cdot p_{\rm x}+\bigg(\frac{\gamma_{\rm A}}2 |\sigma_r^\t({\rm x}) z |^2 - \widehat H_r({\rm x},  z) + \frac{1}{\gamma_\Ar} \frac{f^\prime(T-t)}{f(T-t)} \bigg) p_{\rm y} \\
&+\frac{1}{2}   {\rm Tr} \big [\sigma \sigma_t^\t ({\rm x})(M_{\rm xx} +z z^\t M_{\rm yy} +2z \cdot  M_{\rm xy} )\big]\bigg\},
\end{align*}
for $p:=\begin{pmatrix}
p_{\rm x} \\
p_{\rm y}
\end{pmatrix}\in \R^n\times \R$, $M:=\begin{pmatrix}
M_{{\rm xx}}& M_{{\rm xy}}\\
 M_{{\rm xy}} & M_{{\rm yy}}
 \end{pmatrix}\in\S_{n+1}^+(\R) $,  $M_{{\rm xx}}\in\S_{n}^+(\R)$, $M_{{\rm yy}}\in\S_{1}^+(\R)$, and $M_{{\rm xy}}\in\R^{n\times 1}$.
\end{enumerate}
\end{remark}

In the following proposition, whose proof is available in \Cref{sec:apenexamples}, we study the case $n=1$, so that 
\[
\underline V(y_0)=\sup_{Z\in \Hc^{2}} \E^{\P^\star(Z)}\Big[{\rm U_P}\Big(X_T- \widehat Y_T^{y_\smallfont{0},Z} \big) \Big)\Big].
\]
This result is equivalent to solving the HJB equation in \Cref{rmk:ra1hjbeq}.

\begin{proposition}\label{prop:sol2Bra1}
Let principal and agent have exponential utility with parameters $\gamma_{\rm P}$ and $\gamma_{\rm A}$, respectively.  Let $C_{y}:=-\frac{1}{\gamma_{\smallfont\Pr}}\e^{-\gamma_{\smallfont \Pr} (x_0-y)}$, $\widehat R_0:={\Ur_\Ar^o}^{(-1)}( R_0)$, and assume that
\begin{enumerate}[label=$(\roman*)$, ref=.$(\roman*)$,wide,  labelindent=0pt]
\item the maps $\sigma$, $\lambda^\star$ and ${k^o}^\star$ do not depend on the $x$ variable$;$

\item for any $t\in[0,T]$, the map $\R\ni z\overset{g}{\longmapsto} \lambda^\star_t(z)-{k^o_t}^\star(z)-\frac{\gamma_{\rm A}}2 |\sigma_t^\t z|^2-\frac{\gamma_{\rm P}}2 |\sigma_t^\t (1-z)|^2$ has a unique maximiser $z^\star(t)$, such that $[0,T]\ni t\longmapsto z^\star(t)$ is square integrable.
\end{enumerate}

Then
\[
\xi^\star:= \Ur_\Ar^{(-1)}\bigg( \frac{R_0}{f(T)}\bigg) -\int_0^T \bigg(\widehat H_r(z_r^\star) -\frac{\gamma_{\rm A}}2 |\sigma_r^\t  z^\star_r|^2 \bigg)\d r +\int_0^T z_r^\star \d X_r,
\]
is an optimal solution to principal's second-best problem and 
\begin{align*}
{\rm V}^{\rm P}= C_{\widehat R_0} {f(T)}^\frac{\gamma_\smallfont{\Pr}}{\gamma_\smallfont{\Ar}}\exp\bigg(-\gamma_\Pr \int_0^T g(z^\star(t))\d t\bigg).
\end{align*}
\end{proposition}

\begin{remark}\label{rmk:solra1HM}
To close this section we present a few remarks:
\begin{enumerate}[label=$(\roman*)$, ref=.$(\roman*)$,wide,  labelindent=0pt]
\item comparing the results in {\rm \Cref{prop:sol2Bra1}} and {\rm \Cref{prop:solFBra}} we see that, as expected, in general the solution to the second-best and first-best problem are not equal$;$

\item if we bring ourselves back to the setting of {\rm \cite{holmstrom1987aggregation}}, \emph{i.e.} $b_t(x,a)= a/\sigma$, $\sigma_t(x)=\sigma$, $k_t^o(x,a)=k a^2/2$, we have

\[ z^\star(t)=\frac{ 1+\sigma^2 \gamma_\Pr k}{1+\sigma^2k (\gamma_\Ar+\gamma_\Pr)},\;a^\star(t)=\frac{ 1+\sigma^2 \gamma_\Pr k}{c\big(1+\sigma^2k (\gamma_\Ar+\gamma_\Pr)\big)}. \]
This recovers the result for the case of a risk-neutral principal, \emph{i.e.} $\gamma_\Pr=0$, presented in {\rm \cite{holmstrom1987aggregation}}. The optimal contract and the respective rewards differ by a factor which depends on the discount factor and agent's risk aversion parameter$;$

\item following upon the previous comment, we add that the optimal contract takes the form of a Markovian rule. Moreover, it is linear. This is consistent with the seminal work of {\rm \cite{holmstrom1987aggregation}} and the conclusion of {\rm \cite{carroll2015robustness}} in which the robustness of these policies was studied. Nevertheless, as we will see in {\rm \Cref{sec:ra2}}, this appears to be a consequence of the simplicity of the source of time-inconsistency considered in this section.
\end{enumerate}
\end{remark}

\subsection{Agent with separable utility}\label{sec:separable}

We consider the scenario in \Cref{sec:separableFB}, \emph{i.e.}
\begin{align*}
{\rm V}^{\rm A}_t(\xi,\alpha)=\E^{\P^\alpha}\bigg[f(T-t)\Ur_\Ar^o(\xi)-\int_t^Tf(s-t)c_s\big(X_{\cdot\wedge s},\alpha_s\big)\mathrm{d}s\bigg|\mathcal F_t\bigg], \; (t,\alpha,\xi)\in [0,T]\times \Ac\times \Cf,
\end{align*}

and we have $\Ur_\Ar(s,{\rm x})=f(T-s)\Ur_\Ar^o({\rm x})$, $\partial_s\Ur_\Ar(s,{\rm x})=-f^\prime(T-s)\Ur_\Ar^o({\rm x})$,
\begin{align*}
h_t(s,x,z,a)= \sigma_t(x)b_t(x,a)\cdot z-f(t-s)c_t(x,a),\; \nabla h_t(s,x,v,a)= \sigma_t(x)b_t(x,a)\cdot v+f^\prime(t-s)c_t(x,a).
\end{align*}
The mappings $H_t(x,{\rm z})$, $a^\star(t,x,{\rm z})$, $\lambda^\star_t(x,{\rm z})$, $c^\star_t(x,{\rm z})$, and the probability $\P^\star(Z)=\P^{a^\smallfont{\star}(\cdot,X_\smallfont{\cdot},Z_\smallfont{\cdot}^\smallfont{\cdot})}$ are obtained accordingly. \medskip
\begin{comment}
\begin{align}\label{HJB0}
\begin{split}
Y_t(\xi)&={\rm U_A}(\xi)+\int_t^T\Big( H_r\big(X_{\cdot\wedge r},Z_r(\xi)\big)-\partial Y_r^r(\xi)\Big) \mathrm{d} r-\int_t^T  Z_r(\xi) \cdot  \mathrm{d} X_r,\; t\in[0,T], \\
Y_t^s(\xi)&=f(T-s)\mathrm{U_A}(\xi)+\int_t^T h_r\big(s,X_{\cdot\wedge r}, Z_r^s(\xi),a^\star(r,X_{\cdot\wedge r},Z_r(\xi))\big)\mathrm{d} r-\int_t^T Z_r^s(\xi) \cdot \mathrm{d} X_r,\; t\in[0,T],\\
\partial Y_t^s(\xi)&=-f^\prime(T-s)\mathrm{U_A}(\xi)+\int_t^T \nabla h_r\big(s,X_{\cdot\wedge r},\partial Z_r^s(\xi),a^\star(r,X_{\cdot\wedge r},Z_r(\xi))\big) \mathrm{d} r-\int_t^T\partial Z_r^s(\xi) \cdot \mathrm{d} X_r,\; t\in[0,T].
\end{split}
\end{align}
\end{comment}

In this section, we are trying to get a deeper understanding of the family $\Hc^{2,2}$ under the previous specification of preferences for the agent. In particular, we want to understand how the elements of the family $(Y^{s,Z},Z^s)_{s\in[0,T]}$ are related to each other.  In light of \Cref{assump:datareward} and \Cref{assump:uniquemax}, for any $Z\in \Hc^{2,2}$ we denote $M^{s,Z}$ the $(\F,\P^\star(Z))$-square integrable martingale
\[
M_t^{s,Z}:=\E^{\P^\star(Z)}\bigg[ \int_0^T\delta^\star_r\big(s,X_{\cdot\wedge r},Z_r^r\big)\mathrm{d}r\bigg|\Fc_t\bigg],\; t\in[0,T],
\]
where
\begin{align*}
 \delta_r^\star(s,x,{\rm z}):=  c_r^\star(x,{\rm z})\bigg(f(r-s)-\frac{f(T-s)}{f(T)}f(r)\bigg),\; (s,t,x,{\rm z})\in [0,T]^2\times \Omega \times \R^n.
\end{align*}
We also recall that $\P^\star(Z)$ is the unique solution to the martingale problem for which $X$ has characteristic triplet $(\lambda^\star, \sigma\sigma^\t,0)$. Thus, the representation property holds for $(\F, \P^\star(Z))$-martingales (see \cite[Theorem III.4.29]{jacod2003limit}) and we can introduce the unique $\F$-predictable process $\widetilde Z^{s,Z}$ such that $\sup_{s\in [0,T]} \E^{\P^\star(Z)} \bigg[ \displaystyle  \int_0^T |\sigma_r \sigma^\t_r \widetilde Z^{s,Z}_r|^2\d r  \bigg] <\infty$,\footnote{The integrability for $s\in [0,T]$ fixed is clear. The $\sup$ follows as in \cite[Theorem 3.5]{hernandez2020me} as $\delta^\star$ is uniformly continuous in $s$.} and, in light of \eqref{eq:dynprincipal},
\[
M^{s,Z}_t=M_0^{s,Z}+\int_0^t\widetilde Z_r^{s,Z} \cdot\big(\mathrm{d}X_r-\lambda^\star_r\big(X_{\cdot\wedge r},Z_r^r\big)\d r\big),\; t\in[0,T],\; \P\as
\]

The next lemma, proved in \Cref{sec:apenexamples}, presents relationships satisfied by the family $(Y^{s,Z},Z^s)_{s\in[0,T]}$ and how we can use them to obtain another characterisation of $\Hc^{2,2}$ and $\overline \Xi$.

\begin{lemma}\label{lemma:sepagentreduction}

\begin{enumerate}[label=$(\roman*)$, ref=.$(\roman*)$,wide,  labelindent=0pt]
\item Let $Z\in\Hc^{2,2}$, for any $s\in [0,T]$
\[
Y_t^{s,Z}=\frac{f(T-s)}{f(T)}Y_t^{0,Z}-\E^{\P^\smallfont{\star}(Z)}\bigg[ \int_t^T\delta^\star_r(s,X_{\cdot\wedge r},Z_r^r ) \d r\bigg|\Fc_t\bigg],\; t\in[0,T],\; \P\as
\]
\item Let $Z\in\Hc^{2,2}$, for any $s\in [0,T]$
\[
Z_t^s=\frac{f(T-s)}{f(T)}Z_t^0-\widetilde Z_t^{s,Z},\; \mathrm{d}t\otimes\mathrm{d}\P\text{\rm--a.e.}
\]
\item $\Hc^{2,2}= \Hc^{\bullet}$, where $\Hc^{\bullet}$ denotes the class of $Z\in \Ho(\R^d)$ such that $\|Y^{y_{\smallfont0},Z}\|_{\S^\smallfont{2}}<\infty$, where for $y_0\in (R_0,\infty)$
\begin{align*}
Y_t^{y_\smallfont{0},Z}:=\frac{y_0}{f(T)}-\int_0^t f(T)^{-1}h_r^\star\big(0,X_{\cdot\wedge r}, Z_r^0,Z_r^r\big)\mathrm{d} r+\int_0^tf(T)^{-1}Z_r^0\cdot \mathrm{d} X_r,\; t\in[0,T], \;\P\as
\end{align*}
and
\begin{align}\label{eq:constraint}
Z_t^s=\frac{f(T-s)}{f(T)}Z_t^0-\widetilde Z_t^{s,Z},\; \mathrm{d}t\otimes\mathrm{d}\P\text{\rm--a.e.}
\end{align}

\item $\overline \Xi=\big\{{\Ur_\Ar^o}^{(-1)}(Y_T^{y_\smallfont{0},Z}): (y_0,Z)\in (R_0,\infty)\times \Hc^\bullet\big \}$. For any $\xi\in \overline \Cf$,
\[\Ec(\xi)=\big \{(a^\star(t,X_{\cdot\wedge t}, Z_t^t)_{t\in [0,T]}\big\},\; \Vr_0^\Ar(\xi)=y_0.\]

\end{enumerate}
\end{lemma}

\begin{remark}\label{rmk:separablecharac}
\begin{enumerate}[label=$(\roman*)$, ref=.$(\roman*)$,wide,  labelindent=0pt]
\item In the exponential discounting case, \emph{i.e.} $f(t):=\mathrm{e}^{-\rho t}$ for some $\rho>0$, we have
\[
f(r-s)-\frac{f(T-s)f(r-u)}{f(T-u)}=\mathrm{e}^{-\rho (r-s)}-\mathrm{e}^{-\rho(r-s)}=0,\; (r,s,u)\in[0,T]^3.
\]
Thus, $\delta^\star=0$ and the result of {\rm\Cref{lemma:sepagentreduction}} simplifies to
\[
Y_t^{s,y_\smallfont{0},Z}=\frac{f(T-s)}{f(T)}Y_t^{0,y_\smallfont{0},Z},\; t\in[0,T] ,\; \P\text{\rm--a.s.}\; Z_t^{s}=\frac{f(T-s)}{f(T)} Z_t^{0},\; \d t\otimes \d \P\ae,\;  s\in [0,T].
\]
Therefore, this implies that in the non-exponential discounting case, the term
\[
\E^{\P^\star(Z)}\bigg[ \int_t^T\delta^\star_r(s,X_{\cdot\wedge r},Z_r^r )  \bigg|\Fc_t\bigg],
\]
is exactly the correction due to time-inconsistency.
\item We also remark that the choice $Z^0$ in the constraint for the family $Z$ is arbitrary. Indeed, it could be replaced by any other element $Z^u$ of the family $Z\in \Hc^{2,2}$.
\end{enumerate}
\end{remark}

\subsubsection{Principal's second best solution}

Thanks to \Cref{lemma:sepagentreduction}, we have now proved that

\begin{proposition}\label{proposition:sepprincipalreduction}
The problem of the principal can be represented as the following control problem
\begin{align*}
{\rm V}^{\rm P}=\sup_{y_\smallfont{0} \geq R_\smallfont{0}} \underline V(y_0),\; \text{\rm where}\;
\underline V(y_0)=  \sup_{Z\in \Hc^{\smallfont\bullet}}\E^{\P^\smallfont{\star}(Z)}\Big[{\rm U_P}\Big(X_{\cdot \wedge T},  \Ur_\Ar^{(-1)}(Y_T^{y_\smallfont{0},Z}) \Big)\Big],
\end{align*}
where $\P^\star(Z)=\P^{a^\smallfont{\star}(\cdot,X_\smallfont{\cdot},Z_\smallfont{\cdot}^\smallfont{\cdot})}$ and 
\begin{align*}
Y_t^{y_\smallfont{0},Z}=\frac{y_0}{f(T)}-\int_0^t f(T)^{-1}h_r^\star\big(0,X_{\cdot\wedge r}, Z_r^0,Z_r^r\big)\mathrm{d} r+\int_0^tf(T)^{-1}Z_r^0\cdot \mathrm{d} X_r,\; t\in[0,T], \;\P\as
\end{align*}
\end{proposition}

We remark that contrary to the example in \Cref{sec:ra1}, \Cref{proposition:sepprincipalreduction} reduces the problem of the principal to a non-standard control problem. Indeed, we have to optimise over $\Hc^\bullet$, a family of infinite-dimensional controls which has to satisfy a novel type of constraint, namely \eqref{eq:constraint}. Nonetheless, under additional assumptions on the model, we can proceed with the resolution.\medskip

As in \Cref{sec:ra1principalsb}, we focus on the case $n=1$ so that
\begin{align}\label{eq:reductionSBsep}
\underline V(y_0)=\sup_{Z\in \Hc^{\smallfont\bullet}} \E^{\P^\smallfont{\star}(Z)}\Big[\Ur_\Pr^o \Big(X_T- {\Ur_\Ar^o}^{(-1)}\big(Y_T^{y_\smallfont{0},Z}\big) \Big)\Big].
\end{align}

\begin{proposition}\label{prop:sol2ndbest0}
Let $n=1$, the principal and the agent be risk-neutral, \emph{i.e.} $\Ur_\Ar^o({\rm x})=\Ur_\Pr^o({\rm x})={\rm x}$, ${\rm x}\in\R$.
\begin{enumerate}[label=$(\roman*)$, ref=.$(\roman*)$,wide,  labelindent=0pt]
\item \label{prop:sol2ndbest0:i} Suppose there is a unique measurable map $z^\star:[0,T]\times \Omega\times \R^n\longrightarrow \R^n$ satisfying \[
{\rm H}_t(x,v)=v\lambda^\star_t(x,z^\star(t,x,v))+\lambda^\star_t(x,z^\star(t,x,v))-\frac{f(r)}{f(T)} c^\star(x,z^\star(t,x,v)),\; \text{\rm for any}\; (t,x)\in [0,T]\times \Omega,
\] where for any $(t,x,v)\in [0,T]\times \Omega\times \R^n$
\[
{\rm H}_t(x,v):=\sup_{z\in \R}\bigg\{ v\lambda^\star(x,z)+\lambda^\star_t(x,z)-\frac{f(r)}{f(T)} c^\star(x,z)\bigg\},
\]
Moreover, assume the mapping $\R^n\ni v\longmapsto {\rm H}_t(x,v)\in \R$ is Lipschitz-continuous uniformly in $(t,x)$ with linear growth. Then, $\Vr^\Pr=x_0-\frac{R_0}{f(T)} + U_0$ where the pair $(U,V)$ denotes a solution to the {\rm BSDE}
\[
U_t=\int_t^T {\rm H}_r(X_{\cdot\wedge r},V_r) \d r -\int_t^T V_r\cdot  \d X_r.
\]
In addition, let 
 \[y_0^\star:=R_0,\; \Zc_t:=z^\star(t,X_{\cdot\wedge t},V_t),\; Z^{0,\star}_t:=\frac{f(T)}{f(T-t)}\Zc_t,\; \P^\star(\Zc):=\P^{a^\smallfont{\star}(\cdot,X_\smallfont{\cdot},\Zc_\smallfont{\cdot})},\]
 
and suppose $\E^{\P^\smallfont{\star}(\Zc)} \big[   \int_0^T |\sigma_r \sigma^\t_r \Zc_r|^2\d r  \big] <\infty$. Then, there exists $Z^\star\in \Ho$, such that $(y_0^\star, Z^\star)\in [R_0,\infty)\times \Hc^\bullet$ define a solution to the second-best problem and the optimal contract is given by
\[
\xi^\star :=\frac{R_0}{f(T)}-f(T)^{-1}\int_0^T h_r^\star\big(0,X_{\cdot\wedge r}, Z^{0,\star}_r,\Zc_r\big)\mathrm{d} r+f(T)^{-1} \int_0^T Z^{0,\star}_r \cdot \mathrm{d} X_r,\; t\in[0,T], \;\P\as
\]
\item \label{prop:sol2ndbest0:ii} Suppose the maps $\lambda^\star$ and $c^\star$ do not depend on the $x$ variable and for any $t\in[0,T]$, the map $\R\ni z\overset{g}{\longmapsto} \lambda^\star_t(z)-f(t)/f(T)c^\star_t(z)$ has a unique maximiser $z^\star(t)$, such that $[0,T]\ni t\longmapsto z^\star(t)$ is Lebesgue integrable.
\end{enumerate}
Then, a solution $(y_0^\star, Z^\star)\in [R_0,\infty)\times \Hc^\bullet$ for the second-best problem is given by 
\begin{align*}
y_0^\star=R_0,\;  Z^s_t:=\frac{f(T-s)}{f(T-t)}z^\star(t),\; \widetilde Z^s_t:=0,\; (s,t)\in[0,T]^2,\; \text{\rm and}\; \Vr^\Pr=x_0-\frac{R_0}{f(T)}+\int_0^Tg_t\big(z^\star(t)\big)\mathrm{d}t.
\end{align*} 
Moreover, the associated optimal contract is given by
\[
\xi^\star:=\frac{R_0}{f(T)}-f(T)^{-1}\int_0^T\big(\lambda^\star_t(z^\star(t))-f(t)c^\star_t(z^\star(t)) \big)\mathrm{d} t+\int_0^T\frac{z^\star(t)}{f(T-t)} \mathrm{d} X_t.
\]
\begin{proof}
Let us show $(i)$. As both agent and principal are risk neutral, we have
\begin{align*}
\underline {\rm V}(y_0)&=\sup_{Z\in\Hc^\bullet}\E^{\P^\smallfont{\star}(Z)}\bigg[\int_0^T\bigg(\lambda^\star_t(X_{\cdot \wedge r},Z_r^r)-\frac{f(r)}{f(T)}c^\star_r(X_{\cdot \wedge r},Z_r^r)\bigg)\mathrm{d}r\bigg]
\end{align*}
An upper bound $\underline\Vr(y_0)$ is obtained by ignoring \eqref{eq:constraint}. In such scenario, the mapping ${\rm H}$ in the statement denotes the Hamiltonian and by classical arguments in control, see \citet*{el1997backward}, its value is given by $U_0$ where $(U,V)$ are as in the statement. We are left to show this bound is attained. For this we must verify $Z^\star\in \Hc^\bullet$. 

\medskip
On the one hand, note that the integrability of $\Zc$ together with \Cref{assump:datareward} guarantee 
\[
\E^{\P^\smallfont{\star}(\Zc)}\big[ | \xi^\star|^2\big]<\infty.
\]
Therefore, by \cite[Theorem 3.5]{wang2022backward}, there exists a unique solution $(Y^\star,Z^\star)\in \S^{2,2}\times \Ho$ to the BSVIE with data $(\xi^\star,h^\star)$ given by
\[
Y_t^{s,\star}=f(T-s) \xi^\star+\int_t^T h_r^\star\big(s,X_{\cdot\wedge r}, Z^{s,\star}_r,Z_r^{r,\star}\big)\mathrm{d} r-\int_t^T Z^{s,\star}_r \cdot \mathrm{d} X_r, \; t\in [0,T],\; \P\as,\; s\in [0,T].
\]

On the other side, under the integrability assumption on $\Zc$ we have that for every $s\in [0,T]$
\[
\widehat Y_t^{s}:=\E^{\P^\smallfont{\star}(\Zc)}\bigg[ f(T-s) \xi^\star -\int_t^T f(r-s)c^\star(X_{\cdot\wedge r},\Zc_r) \d r \bigg| \Fc_t\bigg],
\]
defines a $\P^\star(\Zc)$-square integrable martingale. Thus, there exists a family of process $(\widehat Z^{s})_{s\in [0,T]}$ such that
\[\widehat Y_t^{s}=f(T-s) \xi^\star+\int_t^T h_r^\star\big(s,X_{\cdot\wedge r}, \widehat Z^{s}_r,\Zc_r\big)\mathrm{d} r-\int_t^T \widehat Z^{s}_r \cdot \mathrm{d} X_r\; t\in [0,T],\; \P\as,\; s\in [0,T].
\]
Moreover, in light of \eqref{assump:datareward} we have that $\widehat Z\in \Ho$. Therefore, by uniqueness of the solution
\[
\big(Y^\star, Z^\star,(Z^{t,\star}_t)_{t\in [0,T]}\big)=\big(\widehat Y ,\widehat Z, \Zc\big), \; \text{in } \S^{2,2}\times \Ho\times\H^2
\]

From this, arguing as in \Cref{lemma:sepagentreduction} we obtain that $Z^\star$ satisfies \eqref{eq:constraint}.\medskip

We now argue $(ii)$. Note that we can find an upper bound for ${\rm V}^{\rm P}$. Indeed, we have
\begin{align*}
{\rm V}^{\rm P}&=x_0+\sup_{y_0\geq R_0}\sup_{Z\in\Hc^\bullet}\E^{\P^\star(Z)}\bigg[-\frac{y_0}{f(T)}+\int_0^T\bigg(\lambda^\star_t(Z_t^t)-\frac{f(t)}{f(T)}c^\star_t(Z_t^t)\bigg)\mathrm{d}t\bigg]\leq x_0-\frac{R_0}{f(T)}+\int_0^Tg_t\big(z^\star(t)\big)\mathrm{d}t=:{\rm V}^{{\rm P},\star}.
\end{align*}
We now show that the pair $(y_0^\star,Z^\star)$ given in the statement is a feasible solution that attains ${\rm V}^{{\rm P},\star}$. To verify feasibility note that, by assumption, $z^\star(\cdot)$ is deterministic, and so is $Z^\star$. Thus, it is straightforward from the definition that $Z^\star\in \Hc^\bullet$. Lastly, it follows by definition that under $(y_0^\star,Z^\star)$ the upper bound ${\rm V}^{{\rm P},\star}$ is attained.
\end{proof}
\end{proposition}

\begin{remark}\label{rmk:separablesol}

\begin{enumerate}[label=$(\roman*)$, ref=.$(\roman*)$,wide,  labelindent=0pt]
\item Let us now present a formal argument regarding our choice $\Ur_\Ar^o({\rm x})=\Ur_\Pr^o({\rm x})={\rm x }$ in the previous result for solving \eqref{eq:reductionSBsep}. Suppose for simplicity the maps $\sigma$, $\lambda^\star$ and $c^\star$ do not depend on the $x$ variable so that the dynamics of the state variables are given by
\[ X_t=x_0+\int_0^t \lambda^\star(Z_r)\d r +\int_0^t \sigma_r \cdot  \d B_r^\star,\; Y_t^{y_\smallfont{0},Z}=\frac{y_0}{f(T)}-\int_0^t \frac{f(r)}{f(T)} c_r^\star(Z_r)\d r +\int_0^t \frac{\sigma_r}{f(T)}Z_r^0 \cdot \mathrm{d} B_r^\star.\]
Moreover, suppose the value function $v(t,{\rm x},{\rm y})$ is regular enough so that It\^o's formula yields, $\P^\star(Z)\as$
\begin{align*}
&v(T,X_T,Y_T^{y_\smallfont{0},Z})-v(t,X_t,Y_t^{y_\smallfont{0},Z})+\int_t^T\Big(\sigma_r \partial_x v+\frac{\sigma_r}{f(T)} Z_r^0 \partial_y v \Big)(r,X_r,Y_r^{y_\smallfont{0},Z}) \d B_r^\star \\
=&\  \int_t^T \bigg( \partial_t v- \lambda_r^\star(Z_r^r)\partial_x v - \frac{f(t)}{f(T)} c_r^\star(Z_r^r) \partial_y v+ \frac{\sigma_r^2}2 \partial_{xx} v+\frac{\sigma_r^2}{f(T)} Z_r^0\partial_{xy} v+\frac{\sigma_r^2}{2 f(T)^2}  |Z_r^0|^2 \partial_{yy} v\bigg)(r,X_r,Y_r^{y_\smallfont{0},Z})\d r.
\end{align*}
Let us highlight the presence of both $Z_t^t$ and $Z_t^0$ in the last term. From this we can see, formally, that for general $\Ur_\Ar^o$ and $\Ur_\Pr^o$ the process $(Z_t^t)_{t\in[0,T]}$ alone is not sufficient to obtain the solution of \eqref{eq:reductionSBsep}. Moreover, recall we can not take $Z_t^t$ and $Z_t^0$ independently due to the constraint \eqref{eq:constraint}. Lastly, under the assumptions of {\rm \Cref{prop:sol2ndbest0}} one expects, intuitively, that $\partial_{xx}v=\partial_{xy}v=\partial_{yy}v=0$ so that the choice $Z^0$ can be made after optimising over $(Z_t^t)_{t\in[0,T]}$.
\end{enumerate}
\end{remark}

\begin{remark}
We close this section with a few remarks.
\begin{enumerate}[label=$(\roman*)$, ref=.$(\roman*)$,wide,  labelindent=0pt]
\item It is worth mentioning that even in the setting of {\rm \Cref{prop:solFBseparable}\ref{prop:solFBseparable:ii}} the optimal contract is neither linear nor Markovian. Moreover, from the expression describing the optimal contract we see that this is entirely related to the presence of the discounting structure which is the source of time-inconsistency.

\item It follows from {\rm \Cref{prop:solFBseparable}} that for risk-neutral preferences, the utility of the principal is the same for both the first-best and second-best problem and that the optimal second-best contract is also optimal there. This is a typical result for time-consistent risk-neutral agents, and it would certainly be worth studying whether this remains true for more general specifications of $\Ur_\Pr^o$ and $\Ur_\Ar^o$. In light of {\rm \Cref{rmk:separablesol}}, this question further motivates the study of the general class of non-standard control problems introduced by {\rm \Cref{thm:repcontractgeneral}}.
\end{enumerate}
\end{remark}

\subsection{Agent with utility of discounted income}\label{sec:ra2}

We now consider the scenario in \Cref{sec:raFB} under the additional choice $f= 1$. We then have
\begin{align}\label{eq:rewardra2}
{\Vr}^{\Ar}_t(\xi,\alpha):=\E^{\P^\alpha}\Big [\Ur_\Ar^o \Big(g(T-t)\xi-K_{t,T}^{t,\alpha} \Big) \Big |\mathcal F_t\Big], \text{ where } \; K_{t,T}^{s,\alpha}:=\int_t^T g(r-s) k_r^o(X,\alpha_r)\d r.
\end{align}
In the context of \eqref{HJB0:general}, this corresponds to 
\begin{align*}
h_t(s,x,y,z,a)&=\sigma_t(x)b_t(x,a)\cdot z +\gamma_\Ar g(t-s)k_t^o(x,a)y,\\
\nabla h_t(s,x,u,v,y,a)&=\sigma_t(x)b_t(x,a)\cdot v-\gamma_\Ar g^\prime(t-s)k_t^o(x,a)y+\gamma_\Ar g(t-s)k_t^o(x,a)u,
\end{align*}
$\Ur_\Ar(s,\xi)=\Ur_\Ar^o(g(T-s)\xi)$, and $\partial_s\Ur_\Ar(s,x)=-g^\prime(T-s)\partial_{{\rm x}}\Ur_\Ar^o(g(T-s)\xi) \Ur_\Ar^o(g(T-s)\xi)$. 
\begin{remark}\label{rmk:ra2intro}
\begin{enumerate}[label=$(\roman*)$, ref=.$(\roman*)$,wide,  labelindent=0pt]
\item \label{rmk:ra2intro:i} The problem introduce by \eqref{eq:rewardra2} is time-inconsistent even in the case of exponential discounting, \emph{i.e.} $g(t)=\e^{-\rho t}$, $t\in [0,T]$, for some $\rho>0$. This is due to the exponential utility $\Ur_\Ar$. Indeed, the {\rm BSDE} representation allows us to interpret the reward of the agent as a recursive utility in which the terminal value is discounted at a rate $e^{g(T-s)}$ whereas the generator discounts at a rate $g(t-s)$. It is known, see {\rm \citet[Section 4.5]{marin2010consumption}}, that even in the case of exponential discounting the problem becomes time-inconsistent as soon as the rates at which the terminal value and the running reward are discounted differ. We also recall that the case of no discounting, \emph{i.e.} $g(t)=1$, corresponds to the seminal work {\rm \citet{holmstrom1987aggregation}}.

\item \label{rmk:ra2intro:ii} Let us note that $h$ exhibits both of the features of the examples in {\rm Sections \ref{sec:ra1}} and {\rm \ref{sec:separable}}, this is, the second term includes the discount factor and the $y$ variable.\footnote{In fact, \eqref{eq:rewardra2} covers the situation in {\rm \Cref{sec:separable}} in the particular case of a risk-neutral agent, recall $1/\gamma_\Ar - \Ur_\Ar ({\rm x}) \longrightarrow {\rm x}$, whenever $\gamma_\Ar\rightarrow 0$.} We highlight that a key element in {\rm \Cref{prop:sol2ndbest0}} was the fact that the dynamics of $Y^{s,y_\smallfont{0},Z}$ were given by $(y_0,Z)$ without $Y^{y_\smallfont{0},Z}$ on the right hand side. Consequently, the presence of $y$ in $h$ forces us to begin by changing variables to the certainty equivalent for the problem of the agent, \emph{i.e.} from $Z$ to $\widehat Z$ as we denote below.  In this way, we remove $y_0$ in the dynamics of $Y^{y_\smallfont{0},Z}$ at the expense of the mapping $\delta^\star$, which we use to identify an auxiliary martingale, becoming quadratic in the new variable $\widehat Z$. On the one hand, this creates a subtle issue when trying to establish a correspondence between the natural integrability of the variables $Z$ and $\widehat Z$, and will ultimately prevent us from obtaining a complete characterisation of the family $\Hc^{2,2}$. On the other hand, the quadratic term does not correspond to the diagonal values of the control variable $\hat Z$. This makes the approach in {\rm \Cref{sec:separable}}, namely {\rm \Cref{prop:sol2ndbest0}}, inoperable and forces us to restrict ourselves to a suitable subclass that is amenable to the analysis.
\end{enumerate}
\end{remark}

As we may probably expect after our analysis in \Cref{sec:ra1}, the process $Y^{y_\smallfont{0},Z}$ in the definition of $\Hc^{2,2}$ becomes more amenable to the analysis by working in terms of the certainty equivalent. For this, we introduce, for $(t,s,x,{\rm z}, z,v)\in[0,T]^2\times\Omega\times (\R^n)^3$,
\begin{align*}
\widehat H_t(x,{\rm z})&:= \sup_{a\in A}  \widehat h_t (t,x,z,a),\;  \widehat h_t (s,x,z,a) :=\sigma_t(x)b_t(x,a )\cdot z-g(t-s) k_t^o(x, a),\\
\nabla \widehat h_t (s, x, v,z, a)&:=\sigma_t(x)b_t(x,a ))\cdot v+ g^\prime(t-s)k_r^o(x,a ) -\gamma_\Ar\sigma_t^\t(x)z\cdot \sigma_t^\t(x)v.
\end{align*}
The maps $\hat a^\star(t,x,{\rm z})$, $\lambda^\star_t(x,{\rm z})$, ${k^o_t}^\star(x,{\rm z})$, $\widehat h_t^\star(s,x,z,{\rm z})$, $\nabla \widehat h_t^\star(s,x,v,z,{\rm z})$, and the probability $\P^\star({\rm z})$ are defined accordingly.\medskip

Moreover, inspired by \Cref{sec:separable},  we introduce the mapping $\delta^\star$ given, for $(s,t,x,{\rm z},z,\tilde z)\in [0,T]^2\times \Omega \times(\R^n)^3$, by
\begin{align*}
 \delta_t^\star(s,x,{\rm z},z,\tilde z):=  {k_r^o}^\star(x,{\rm z})\bigg(g(r-s)-\frac{g(T-s)}{g(T)}g(r)\bigg)+\frac{\gamma_\Ar}2 \bigg( |\sigma_r^\t(x) z|^2-\frac{g(T-s)}{g(T-u)} |\sigma_r^\t(x) \tilde z |^2\bigg).
\end{align*}

The following result is analogue to \Cref{lemma:sepagentreduction}, we defer its proof to \Cref{sec:apenexamples}.

\begin{lemma}\label{lemma:ra2agentreduction}
Let $Z\in \Hc^{2,2}$.
\begin{enumerate}[label=$(\roman*)$, ref=.$(\roman*)$,wide,  labelindent=0pt]
\item There exists family of processes $(\widehat Y^{s,y_\smallfont{0},Z},\widehat Z^s)_{s\in [0,T]}$ such that for every $s\in [0,T]$
\begin{align*}
\widehat Y_t^{s,y_\smallfont{0},Z} &=-\frac{1}{\gamma_\Ar} \ln\big(-\gamma_\Ar  y_0^s\big)-\int_0^t \Big( \widehat h_r^\star\big(s,X_{\cdot\wedge r} , \widehat Z_r^s, \widehat Z_r^r\big)- \frac{\gamma_\Ar}2 |\sigma_r^\t \widehat Z_r^s |^2\Big)\d r+\int_0^t  \widehat Z_r^s \cdot \d X_r, \; t\in [0,T],\; \P\as
\end{align*}
\item If $\widehat Z\in \H^{2,2}$ then for every $s\in [0,T]$
\begin{align*}
\widehat Y_t^{s,y_\smallfont{0},Z}=\frac{g(T-s)}{g(T)}\widehat Y_t^{0,y_\smallfont{0},Z} - \E^{\P^\smallfont{\star}(Z)}\bigg[\int_t^T \delta_r^\star(s,X_{\cdot \wedge r} ,\widehat Z_r^r,\widehat Z^s_r,\widehat Z^0_r)\d r \bigg| \Fc_t \bigg]
\end{align*}

\item Moreover, if the process $M^{s,  Z}$ given by
\begin{align*}
M_t^{s,Z}:=\E^{\P^\smallfont{\star}(Z)}\bigg[\int_0^T \delta_r^\star(s,X_{\cdot \wedge r} ,\widehat Z_r^r,\widehat Z^s_r,\widehat Z^0_r)\d r  \bigg| \Fc_t \bigg], \;\P\as, \;(s,t)\in [0,T]^2.
\end{align*}
is a square integrable $(\F,\P^\star(Z))$-martingale, then
\begin{align*}
\widehat Z_t^s=\frac{g(T-s)}{g(T)}\widehat Z_t^0-\widehat Z_t^{s,Z}, \d t\otimes \d \P\ae
\end{align*}
where $\widehat Z^{s,Z}$ denotes the term in the representation of $M^{s,Z}$.
\end{enumerate}
\end{lemma}

\begin{remark}
\begin{enumerate}[label=$(\roman*)$, ref=.$(\roman*)$,wide,  labelindent=0pt]
\item We highlight that in contrast to the analysis presented in {\rm Sections} {\rm \ref{sec:ra1}} and {\rm \ref{sec:separable}}, the previous result does not provide an equivalent representation of the set $\Hc^{2,2}$. This is intimately related to the square integrability condition on the process $M^{s,Z}$ required in $(iii)$ above, and the fact that $(z,\tilde z)\longmapsto \delta^\star_t(x,{\rm z},z,\tilde z)$ is quadratic for $(t,s,x,{\rm z})\in [0,T]^2\times \Omega\times \R^m$ fixed.

\item As a sanity check at this point, let us verify the coherence of the previous system in terms of the analysis of the previous section. In the following we omit the dependence on $X$ and assume $\widehat Z\in \H^{2,2}(\R^n)$. Let
\begin{align*}
(\Delta^{s} g)(t):=g(t-s)-g(t),\;  K_{t,\tau}^{Z,s}:=\exp\bigg(\gamma_\Ar \int_t^\tau\Big( k_r^\star( \widehat Z_r)(\Delta^{s}g)(r)-\gamma_\Ar|\sigma_r  \widehat Z_r^0|^2+\gamma_\Ar \widehat Z_r^{s,\t} \sigma_r  \cdot  \widehat Z_r^{0\t} \sigma_r\Big)\d r\bigg),
\end{align*}
By applying It\^o's formula to $\widetilde Y_r^{s}:=K_{t,r}^{Z,s}\Ur_\Ar(\widehat Y_r^{s,y_{\smallfont0},Z}-\widehat Y_r^{0,y_\smallfont{0},Z})$, we have that for any $s \in [0,T]$, $\P\as$
\[
 \Ur_\Ar^o\big(\widehat Y_t^{s,y_\smallfont{0},Z}-\widehat Y_t^{0,y_\smallfont{0},Z}\big)  =\E^{\P^\star(Z)}\bigg[\Ur_\Ar^o \bigg(y^s_0-y_0^0 - \int_0^t \Big(  (\Delta^{s} g)(r) k^\star_r(\widehat Z_r)-\gamma_\Ar |\sigma_r\widehat Z_r^0|^2+\gamma_\Ar \widehat Z_r^{s,\t} \sigma_r \cdot  \widehat Z_r^{0\t}\sigma_r\Big)\d r \bigg)\bigg|\Fc_t \bigg],
\]
As $1/\gamma_\Ar - \Ur_\Ar^o (y)\xrightarrow{\gamma_\Ar\rightarrow 0} y$, we see the previous equation induces the corresponding one in {\rm \Cref{lemma:sepagentreduction}}.
\end{enumerate}
\end{remark}

\subsubsection{Principal's second-best solution}\label{sec:ra2principalsb}

Let us highlight that in contrast to \Cref{sec:separable}, the analysis in the previous section does not provide a full characterisation of $\Hc^{2,2}$ for rewards given by \eqref{eq:rewardseparable}. This is principally due to the integrability necessary on the variable $\widehat Z$, induced by the certainty equivalent, in order to apply the methodology devised in \ref{sec:separable}, see \Cref{lemma:ra2agentreduction}. Nevertheless, given that the current example generalises the previous two, we build upon the structure of those optimal solutions to propose a family over which the optimisation in the problem of the principal can be carried out.\medskip

We will focus on the case $n=1$ and we will pay special attention to the class $\widetilde\Hc \subseteq \Hc^{2,2}$ of processes $Z\in \Hc^{2,2}$ for which given the pair $(y_0,Z)\in \Ic_0 \times \Hc^{2,2}$,  and $Y^{y_\smallfont{0},Z}$ given by \eqref{eq:fsvie},  there exists a pair of predictable processes $(\eta,\zeta)$ such that
\begin{align*}
\widehat Z_t^s=  \frac1{-\gamma_\Ar}\frac{ Z_t^s}{ Y_t^{y_\smallfont{0},s,Z}}=  \eta_{s,t} \zeta_t,\; \eta_{t,t}=1,\; t\in [0,T].
\end{align*}

Therefore, we have from \Cref{thm:repcontractgeneral} and \eqref{eq:constraint} that
\[
\underline \Vr^\Pr:= \sup_{\substack{y_\smallfont{0}^\smallfont{0}\geq R_\smallfont{0}} }\sup_{Z \in\tilde \Hc}\E^{\P^\smallfont{\star}(Z)}\Big[\Ur_\Pr^o\Big(X_T-  \widehat{Y}_T^{0,y_\smallfont{0},Z} \Big)\Big]\leq {\rm V}^{\rm P}
\]

\begin{remark}
\begin{enumerate}[label=$(\roman*)$, ref=.$(\roman*)$,wide,  labelindent=0pt]
\item We remark that the previous definition implicitly requires that for any $t\in [0,T]$ the mapping $s\longmapsto \eta_{s,t}$ is differentiable.

\item In addition, provided ${k^o}^\star(x,z)$ does not depend on $x$ it is easy to verify that that $\widetilde \Hc \neq \emptyset$. In light of the previous lemma, we have that for $Z\in \widetilde \Hc$
\begin{align*}
M_t^{s,Z} =\E^{\P^\smallfont{\star}(Z)}\bigg[ \int_0^T  {k_r^o}^\star(\zeta_r)\bigg( g(r-s)-\frac{g(T-s)}{g(T)}g(r) \bigg)+\frac{\gamma_\Ar}2\big|\sigma_r^\t \zeta_r\big|^2\bigg( | \eta_{s,r} |^2-\frac{g(T-s)}{g(T)}| \eta_{0,r} |^2\Big|^2\bigg)\d r\bigg|\Fc_t\bigg],
\end{align*}
and
\begin{align*}
\widetilde Z_t^{s,Z}=\bigg(\frac{g(T-s)}{g(T)}\eta_{0,t} -\eta_{s,t}\bigg) \zeta_t, \d t \otimes \d \P\ae
\end{align*}
This implies that $\widetilde \Hc$ includes, in particular, all the processes $Z$ that are induced by deterministic pairs $(\zeta, \eta)$. Indeed, for such class of processes we have that $M^{s,Z}$ is deterministic, $\widetilde Z^{s,Z}=0$, and consequently, $\eta_{s,t}=g(T-s)/g(T-t)$ provides a non trivial element of $\widetilde \Hc$. The previous argument also holds in the case of exponential discounting, in which we recall that the agent's problem remains time-inconsistent.
\end{enumerate}
\end{remark}

The following result characterises the solution to $\underline \Vr^\Pr$. Its proof is available in \Cref{sec:apenexamples}.

\begin{proposition}\label{prop:sol2Bra2}
Let principal and agent have exponential utility with parameters $\gamma_{\rm P}$ and $\gamma_{\rm A}$, respectively.  Let $C_{y}:=-\frac{1}{\gamma_{\smallfont \Pr}}\e^{-\gamma_{ \smallfont\Pr} (x_0-y)}$, $\widehat R_0:={\Ur_\Ar^o}^{(-1)}( R_0)$, and assume that:
\begin{enumerate}[label=$(\roman*)$, ref=.$(\roman*)$,wide,  labelindent=0pt]
\item the maps $\sigma$,  $\lambda^\star$ and ${k^o}^\star$ do not depend on the $x$ variable$;$

\item for any $(t,\eta)\in[0,T]\times \R$, the map $G: \R\longrightarrow \R$ given by 
\[
G(z):= \lambda^{\star}_t(z) -\frac{g(t)}{g(T)}{k^o_t}^{\star}(z) -\frac{\gamma_A}{2g(T)}\Big|\frac{\eta_0\sigma_t }{\eta}\Big|^2 |  z |^2-\frac{\gamma_{\rm P}}{2}\sigma_r^2\bigg(1-\frac{\eta_0}{\eta g(T)} z \bigg)^2,
\]
 has a unique maximiser $z^\star(t,\eta)$, such that $[0,T]\ni t\longmapsto z^\star(t,\eta)$ is square-integrable.
\end{enumerate}
Then
\begin{align*}
\underline \Vr^\Pr=\sup_{Z^\smallfont{\eta}  \in \tilde \Hc } C_{\hat  R_\smallfont{0}} \E^{\P}\bigg[\exp\bigg(-\gamma_{\rm P}\int_0^T G_r(z^\star(r,\eta_r)) \d r \bigg)\bigg], \; \text{where}\; Z_t^{s,\eta}:=\frac{\eta_s}{\eta_t} z^\star(t,\eta_t).
\end{align*}
Moreover
\begin{enumerate}[label=$(\roman*)$, ref=.$(\roman*)$,wide,  labelindent=0pt]
\item let $\underline \Vr^{\Pr,o}$ denote the restriction of $\underline \Vr^\Pr$ to the subclass of $\widetilde \Hc$ with deterministic $\eta$. Then the optimal deterministic contract is given by the family 
\[
Z^s_t:=\frac{g(T-s)}{g(T-t)}z^\star(t,g(T-t)),\]
and
\begin{align*}
\xi &=\frac{\widehat R_0}{g(T)} -\frac{1}{g(T)}\int_0^t h^\star_r\bigg(0, \frac{g(T)}{g(T-t)}z^\star(t,g(T-t)), z^\star(r,g(T-r)) \bigg)\d r+ \int_0^t \frac{z^\star(t,g(T-t))}{g(T-t)}  \d X_r;
\end{align*}

\item in the case $(\gamma_\Ar,\gamma_\Pr)=(0,0)$, \emph{i.e.} the case of risk-neutral principal and agent, the solution to $\Vr^\Pr$, and consequently of $\underline \Vr_\Pr$, agrees with the value given by {\rm \Cref{prop:sol2ndbest0}} and the optimal family $Z$ is deterministic.
\end{enumerate}
\end{proposition}

\begin{remark}
We close this section with a few remarks.
\begin{enumerate}[label=$(\roman*)$, ref=.$(\roman*)$,wide,  labelindent=0pt]
\item The solution to the problem of the principal for the general class of restricted contracts induced by $\Hc^{2,2}$ escaped the analysis presented above. As detailed in {\rm \Cref{rmk:ra2intro}\ref{rmk:ra2intro:ii}}, this is due to subtle integrability issues when trying to identify an appropriate reduction of $\Hc^{2,2}$, and the quadratic nature of the generator when working in term of the certainty equivalent.We believe this echoes the intricacies of the non-standard class of control problem introduced in {\rm \Cref{thm:repcontractgeneral}}. 

\item If, as in {\rm \Cref{rmk:solra1HM}}, we bring ourselves back to the setting of {\rm \cite{holmstrom1987aggregation}}, \emph{i.e.} $b_t(x,a)= a/\sigma$, $\sigma_t(x)=\sigma$, $k_t(x,a)=k a^2/2$,  we have 
\[
Z_t^s=\frac{g(T) g(T-s)(g(T-t)+\gamma_\Pr\sigma^2 k)}{g(t)g^2(T-t)+\sigma^2 k g(T)(\gamma_\Ar g(T)+\gamma_\Pr)}.
\]
We highlight that: $(a)$ in contrast to {\rm \cite{holmstrom1987aggregation}}, for any type of discounting structure $($including exponential discounting$)$ the previous expression and consequently the optimal action is neither linear nor Markovian. This corroborates our comment in {\rm \Cref{rmk:ra2intro}\ref{rmk:ra2intro:i}}, in the sense that even in the case of exponential discounting the problem of the agent remains time-inconsistent; $(b)$ in the case of no discounting, \emph{i.e.} $g=1$,  when we bring ourselves back to the model of {\rm \Cref{rmk:solra1HM}}, the previous expression coincides with the linear contract result specified by {\rm \cite{holmstrom1987aggregation}}. This shows that, even if possibly not the best, the optimal contract in the class $\widetilde\Hc$ at least captures the optimal contract when the problem becomes time-consistent again.\medskip

\end{enumerate}
\end{remark}

\begin{appendix}
\section*{Appendix}

\section{Proofs of Section \ref{sec:firstbestexamples}}

\begin{proof}[Proof of {\rm \Cref{prop:solFBra}}]
Let $(\rho,\alpha)\in \R_+\times \Ac$ be fixed and optimise the mapping $\Cc \ni\xi \longmapsto \mathfrak{L}(\alpha, \xi, \rho)\in\R$. An upper bound of this problem is given by optimising $x$-by-$x$. This leads us to define, for any $(\alpha,\rho)\in \Ac\times \R_+$ fixed, the candidate
\begin{align*}
\xi^\star(\rho,\alpha)= \frac{1}{g(T)\gamma_\Ar+\gamma_\Pr}\Big(\gamma_\Pr \Gamma( X_T) +\gamma_\Ar K_{0,T}^{0,\alpha} +\log\big(\rho g(T) f(T)\big) \Big).
\end{align*}
To show the upper bound induced by $\xi^\star(\rho,\alpha)$ is attained it suffices to note that $\xi^\star(\rho,\alpha)\in \Cc$ by assumption. Replacing in $\mathfrak{L}(\rho,\alpha,\xi)$ we obtain
\[\Vr^{\Pr,{\rm d}}_{\rm FB}=\inf_{\rho\in \R_+} \bigg\{-  \rho R_0- \frac{1}{\bar \gamma}(\rho g(T)f(T))^{\frac{\bar \gamma}{g(T) \gamma_\smallfont{\Ar}}} 
\sup_{\alpha\in \Ac} \E^{\P^\alpha}\bigg[\exp\bigg(-\bar\gamma \Gamma(X_T)+\frac{\bar\gamma}{g(T)} K_{0,T}^{0,\alpha}\bigg)\bigg]\bigg\}.\]

If $\Vr_{\rm cont}<\infty$, as the above function is a strictly convex function of $\rho$, first order conditions gives $\rho^\star$ as in the statement.\medskip

We are only left to show that $\Vr^{\Pr}_{\rm FB}=\Vr^{\Pr,{\rm d}}_{\rm FB}$, \emph{i.e.} that there is no duality gap. For this, it suffices to verify that $(\xi^\star(\rho^\star,\alpha^\star),\alpha^\star)$ is primal feasible, \emph{i.e.} that it satisfy the participation constraint. Indeed
\begin{align*}
& \E^{\P^{\alpha^\smallfont{\star}}}\bigg[\Ur_\Ar^o\bigg(g(T)\xi-\int_0^Tg(r)k_r^o \big(X_{\cdot\wedge r},\alpha_r^\star\big)\d r\bigg)\bigg]\\
&= \frac{-1}{\gamma_\Ar} \E^{\P^{\alpha^\smallfont{\star}}}\bigg[\exp\bigg(-\bar \gamma \Gamma(X_T)-\frac{\gamma_\Ar}{\gamma_\Pr} \bar \gamma K_{0,T}^{0,\alpha^\smallfont{\star}} -\frac{\bar\gamma}{\gamma_\Pr} \log\bigg(\frac{\bar \gamma f(T)}{\gamma_\Ar R_0} V_{\rm cont}\bigg)^{1+\frac{\gamma_\smallfont{\Pr}}{\gamma_\smallfont{\Ar} g(T)}}-\gamma_\Ar K_{0,T}^{0,\alpha^\smallfont{\star}} \bigg)\bigg]\\
&= \frac{-1}{\gamma_\Ar} \E^{\P^{\alpha^\smallfont{\star}}}\bigg[\exp\bigg(-\bar \gamma \Gamma(X_T)-\frac{\gamma_\Ar}{\gamma_\Pr} \bar \gamma K_{0,T}^{0,\alpha^\smallfont{\star}} -\gamma_\Ar K_{0,T}^{0,\alpha^\smallfont{\star}} \bigg) \bigg]\bigg(\frac{\bar\gamma}{\gamma_\Ar} V_{\rm cont}\bigg)^{-1} \frac{R_0}{f(T)}=\frac{R_0}{f(T)}.
\end{align*}
\begin{comment}
We next note that under the probability measure $\widetilde \P^\alpha$, equivalent to $\P^\alpha$, given by
\begin{align*}
\frac{\d \widetilde \P^\alpha}{\d \P^\alpha}:=\exp\bigg(- \bar \gamma \int_0^T  \sigma_r(X_{\cdot\wedge r}) \cdot    \d B_r^{\alpha}-\frac{\bar \gamma^2}2 \int_0^T |\sigma_r^\t(X_{\cdot\wedge r})|^2\d r\bigg),
\end{align*}
we have that
\[
\E^{\P^\alpha}\bigg[\e^{-\bar\gamma\big(X_T-K_{0,T}^{0,\alpha}\big)}\bigg]=\e^{-\bar\gamma x_0}\E^{\tilde \P^\alpha}\bigg[\exp\bigg(-\bar\gamma\int_0^T\big(\sigma_r(X_{\cdot\wedge r})b_r(X_{\cdot\wedge r},\alpha_r)-g(r)k_r(X_{\cdot \wedge r},\alpha_r) -\frac{\bar \gamma}2|\sigma_r(X_{\cdot\wedge r})|^2 \big)\d r\bigg)\bigg].
\]
\end{comment}
\end{proof}

\begin{proof}[Proof of {\rm \Cref{prop:solFBseparable}}]
We argue $(i)$. Let $(\rho,\alpha)\in \R_+\times \Ac$ be fixed and optimise the mapping $\Cf \ni\xi \longmapsto \mathfrak{L}(\alpha, \xi, \rho)\in\R$. An upper bound of this problem is given by optimising $x$-by-$x$.  This defines the mapping $\xi^\star({\rm x},\rho)$. As before, the fact that $\xi^\star(\rho,\alpha)\in \Cf$ guarantees the upper bound is indeed attained. Replacing in $\mathfrak{L}(\rho,\alpha,\xi)$ we obtain ${\rm V_{cont}}(\rho)$ and the corresponding equality for ${\rm V^{P,d}_{FB}}$. Now, to obtain the absence of duality gap we must verify that there exists a solution to the dual problem that is primal feasible. This is exactly the additional assumption in the statement.\medskip

We now consider $(ii)$. In this case, we can solve $\Vr^\Pr_{\rm FB}$ directly. In light of $\Ur_\Pr({\rm x})=\Ur_\Ar({\rm x})={\rm x}$, 

\[
{\rm V^P_{FB}}=\sup_{(\alpha,\xi)\in\Ac\times\Cf}\E^{\P^\alpha}\big[\Gamma(X_T)-\xi\big],\; \text{\rm s.t.}\; \E^{\P^\alpha}\bigg[f(T)\xi-\int_0^Tf(r)c_r\big(X_{\cdot\wedge r},\alpha_r\big)\mathrm{d}r\bigg]\geq R_0.
\]

Let us note that for fixed $\alpha\in \Ac$ the principal's reward is linear and strictly decreasing in $\E^{\P^\alpha}[\xi]$ and therefore she is indifferent between contracts that have the same expectation. Therefore, she optimises over the feasible contracts that have the same expectation. Now, for fixed $\alpha\in \Ac$ any feasibility contract satisfies
\[
 \E^{\P^\alpha}[\xi^\star]\geq \E^{\P^\alpha}\bigg[\int_0^T\frac{f(r)}{f(T)}c_r\big(X_{\cdot\wedge r},\alpha_r\big)\mathrm{d}r\bigg]+ \frac{R_0}{f(T)}.
\]
Therefore, our previous comment implies that for given $\alpha$ the principal is indifferent between contracts in $\hat \Cf(\alpha)$.
Note that $\hat \Cf(\alpha)\neq \emptyset$. Indeed, take the deterministic contract $\xi^\star(\alpha):=\frac{R_0}{f(T)}+f(T)^{-1}\E^{\P^\alpha}\bigg[\displaystyle \int_0^Tf(r)c_r\big(X_{\cdot\wedge r},\alpha_r\big)\mathrm{d}r\bigg]$. \medskip

Plugging this back into the principal's utility, we get the expression for $\Vr^\Pr_{\rm FB}$ in the statement.
\end{proof}

\section{On time-inconsistency for BSVIE-type rewards}\label{sec:appendixdpp}

Let us start by mentioning that in the context of rewards given by \eqref{eq:bsvieval}, the methodology devised in \cite{hernandez2020me}, which builds on the approach in the Markovian framework of \cite{bjork2017time}, is based introducing the family of processes $(Y^s,Z^s)_{s\in [0,T]}$ solution to the backward stochastic Volterra integral equation {\rm (BSVIEs} for short$)$, which satisfies 
\begin{align}\label{eq:bsdevalapendix}
Y_t^{s,\alpha} = \eta(s, \xi) +\int_t^T  h_r\big(s,X,Y_r^{s,\alpha}  , Z_r^{s,\alpha},\pi_r,\alpha_r\big) \d r-\int_t^T  Z_r^{s,\alpha} \cdot \d X_r, \; t\in [0,T],\; \P\as,\; s\in [0,T].
\end{align}
Throughout this section we fix $\Cc=(\xi,\pi)\in\Cf$ and $\alpha^\star\in \Ec:=\Ec(\Cc)$. Thus, we identify the agent's reward under $\alpha\in \Ac$ via $\Vr^\Ar_t(\alpha):=\Vr^\Ar_t(\alpha,\Cc)= Y_t^{t,\alpha}$. We write $\Vr^\Ar_t$ for the associated value function under $\alpha^\star$.\medskip

To establish an extended dynamic programming principle, we need the following minimal set of assumptions.

\begin{assumption}\label{assump:datadpp}
\begin{enumerate}[label=$(\roman*)$, ref=.$(\roman*)$,wide,  labelindent=0pt]

\item  \label{assump:datadpp:i} $(s,y,z) \longmapsto  h_t(s,x,y,z,p,a)$ $($resp. $s\longmapsto \eta(s,x))$ is continuously differentiable. Moreover, the mapping $\nabla h:[0,T]^2\times \Omega \times (\R \times \R^{n} )^2\!\times \R\times A    \longrightarrow \R$ defined by \vspace{-1em}
\[
 \nabla h_t (s,x,u, v ,y,z,p,a):=\partial_s h_t(s,x,y,z,p,a)+\partial_y h_t(s,x,y,z,p,a){ u}+\sum_{i=1}^n \partial_{z_{i}} h_t(s,x,y,z,p, a){ v}_{i},
 \]\vspace{-1em}
 \\
 satisfies $\nabla h _\cdot(s,\cdot, u, v ,y,z,p,a)\in \Pc_{\rm prog}(\R,\F);$ for all $s\in [0,T];$

\item\label{assump:datadpp:ii} \textcolor{black}{for $\varphi \in \{h, \partial_s h \}$, $(y,z,a)\longmapsto  \varphi_t(s,x,y,z,a)$ is uniformly Lipschitz-continuous, \emph{i.e.} there exists some $C>0$ such that $\forall  (s,t,x,y,\tilde y,z,\tilde z,a,\tilde a)$,}
\begin{align*}
 |\varphi_t(s,x,y,z,a)-\varphi_t(s,x,\tilde y,\tilde z,\tilde a)|\leq C\big(|y-\tilde y|+|\sigma_t(x)^\t(z-\tilde z)|+|a-\tilde a|\big).
\end{align*}
\item \label{assump:datadpp:iii} Let $ \big(\tilde h_\cdot(s), \nabla \tilde h_\cdot(s)\big):=\big(h_\cdot(s,\cdot,0,0,\pi_\cdot,0),\nabla h_\cdot(s,\cdot,0,0,0,0,\pi_\cdot,0)\big)$, then the pair $ \big(\tilde h,  \nabla \tilde h):=\big(\tilde h(s) , \nabla \tilde h(s)\big)_{s\in [0,T]}$ belongs to $\mathbb{L}^{1,2,2}(\R)\times \mathbb{L}^{1,2,2}(\R)$.
\end{enumerate}
\end{assumption}

%\begin{remark}
%\textcolor{magenta}{If we leave the integrability condition on $\pi$ in $\Cf_0$ we drop this} We remark that {\rm \Cref{assump:datadpp}\ref{assump:datadpp:iii}} is satisfied if one assumes the classic linear growth condition, \emph{i.e.} 
%there exists some $C>0$ such that for all $(s,t,x,y,z,u,v,a)$
%\[ \big|h_t(s,x,y,z,a)\big|\leq C\big(1+|y|+|\sigma_t(x)^\t z|\big), \;\big|\partial_s  h_t(s,x,y,z,a)\big|\leq C\big(1+|y|+|\sigma_t(x)^\t z|\big).
%\]
%\end{remark}

Under {\rm \Cref{assump:datadpp}}, {\rm \cite[Lemma 6.1]{hernandez2020unified}} guarantees that for any $\alpha\in \Ac$ there exists $(\partial Y^\alpha, \partial Z^\alpha)\in \S^2\times \H^{2,2}$ such that for every $s\in [0,T]$
\begin{align}\label{eq:bsdeparapendix}
\partial Y_t^{s,\alpha}=\partial_s \eta(s,\xi)+ \int_t^T \nabla h_r(s,X,\partial Y_r^{s, \alpha},\partial Z_r^{s, \alpha}, Y_r^{s, \alpha} Z_r^{s, \alpha},\pi_r,\alpha_r)  \d r - \int_t^T\partial Z_r^{s, \alpha}\cdot \d X_r, \; t\in [0,T],\; \P\as,
\end{align}
which ultimately implies the absolute continuity of the mapping $([0,T],\Bc([0,T]))\longrightarrow (\H^{2},\|\cdot \|_{\H^2}): s\longmapsto Z^{s,\alpha}$. With this, the process $(Z_t^{t,\alpha})_{t\in [0,T]}$ is well-defined.  Moreover, see {\rm \cite[Lemma 6.2]{hernandez2020unified}}, for any $\alpha\in \Ac$
\begin{align}\label{eq:propertybsdeapendix}
Y_t^{t,\alpha}=\eta(T,\xi)+\int_t^T \Big( h_r(r,X , Y_r^{r,\alpha},Z_r^{r,\alpha},\pi_r,\alpha_r)-\partial Y_r^{r,\alpha}\Big)\d r-\int_t^T Z_r^{r,\alpha}\cdot \d X_r, \; t\in [0,T],\; \P\as
\end{align}

We begin stating the following auxiliary result.

\begin{lemma}\label{lemma:apendixdpp}
Let {\rm \Cref{assump:datadpp}} hold. For any $\{\gamma, \gamma^\prime\} \subseteq \Tc_{0,T}$,  $\gamma\leq \gamma^\prime$,  and $\alpha\in \Ac$
\[  \E^{\P}\bigg [Y_\gamma^{\gamma,\alpha}-Y_{\gamma^\prime}^{\gamma^\prime, \alpha }+\int_{\gamma}^{\gamma^\prime}\partial Y_r^{r,\alpha}\d r\Big|\Fc_\gamma\bigg],\; \text{\rm depends only on the value of } \alpha \text{ on } [\gamma,\gamma^\prime].\]
\begin{proof}
This property is clear for BSDEs whose generator does not depend on $(y,z)$.  Indeed,
\begin{align*}
 \E^{\P}\bigg [Y_\gamma^{\gamma,\alpha}-Y_{\gamma^\prime}^{\gamma^\prime, \alpha }+\int_{\gamma}^{\gamma^\prime}\partial Y_r^{r,\alpha}\d r\Big|\Fc_\gamma\bigg]= \E^{\P}\bigg [\int_\gamma^{\gamma^\prime} h_r(r,X_{\cdot\wedge r},\alpha_r)\d r\Big|\Fc_\gamma\bigg].
\end{align*}

 To extend this result to the BSDEs \eqref{eq:bsdevalapendix}--\eqref{eq:bsdeparapendix} we consider the Picard iteration procedure
\begin{align*}
Y_t^{s,\alpha,n+1}&= \eta(s, \xi) +\int_t^T  h_r\big(s,X_{\cdot\wedge r},Y_r^{s,\alpha ,n}  , Z_r^{s,\alpha,n},\pi_r,\alpha_r\big) \d r-\int_t^T  Z_r^{s,\alpha,n+1} \cdot \d X_r,\\
\partial Y_t^{s,\alpha ,n+1}&=\partial_s \eta(s,\xi)+ \int_t^T \nabla h_r(s,X_{\cdot\wedge r},\partial Y_r^{s,\alpha,n},\partial Z_r^{s,\alpha,n}, Y_r^{s,\alpha,n}, Z_r^{s,\alpha,n},\pi_r,\alpha_r)  \d r - \int_t^T\partial Z_r^{s, \alpha,n+1}\cdot \d X_r, 
\end{align*}
and note that, as in \eqref{eq:propertybsdeapendix}
\begin{align}\label{eq:recursiveppty}
\E^{\P}\bigg [ Y_\gamma^{\gamma,\alpha,n+1}-Y_{\gamma^\prime}^{\gamma^\prime,\alpha,n+1}+\int_{\gamma}^{\gamma^\prime}\partial Y_r^{r,\alpha,n+1}\d r\bigg|\Fc_\gamma\bigg]= \E^{\P}\bigg [\int_{\gamma}^{\gamma^\prime}   h_r(r,X_{\cdot\wedge r}, Y_r^{r,\alpha,n},Z_r^{r,\alpha,n},\pi_r,\alpha_r)\d r\bigg|\Fc_\gamma\bigg].
\end{align}
Then, from the fact that $Y^{\alpha,0}=Z^{\alpha,0}=\partial Y^{\alpha,0}=\partial Z^{\alpha,0}=0$ we see that \eqref{eq:recursiveppty} implies the result at the initial step. It is then also clear, again from \eqref{eq:recursiveppty}, that this property is preserved at every iteration and thus in the limit. 
\end{proof}
\end{lemma}

In the following, given $(\sigma, \tau) \in \Tc_{t,T}\times \Tc_{t,t+\ell}$, with $\sigma \leq \tau$, we denote by $\Pi^\ell:=(\tau_i^\ell)_{i=1,\dots,n_\ell}\subseteq \Tc_{t,T}$ a generic partition of $[ \sigma, \tau]$ with mesh smaller than $\ell$, \emph{i.e.} for $n_\ell := \ceil[\big]{( \tau-\sigma )/ \ell}$, $\sigma=:\tau^\ell_0 \leq \dots\leq  \tau^\ell_{n^\ell}:=\tau,$ $\forall \ell$, and $\sup_{1\leq i\leq  n_\ell} |\tau^\ell_i-\tau^\ell_{i-1}|\leq \ell$. We also let $\Delta \tau_i^\ell:= \tau_{i}^\ell-\tau_{i-1}^\ell$. The previous definitions hold $x$-by-$x$. \medskip

\begin{theorem}[Dynamic programming principle]\label{thm:dpp}
Let {\rm \Cref{assump:datadpp}} hold. Let $\alpha^\star \in \Ec(\Cc)$ and $\{\sigma, \tau\}\subset \Tc_{t,T}$, with $\sigma\leq \tau$. Then, $\P \as$
\begin{align*}
\Vr_\sigma^\Ar= \es_{\alpha\in \Ac}\;  \E^{\P}\bigg [\Vr_\tau^\Ar+ \int_{\sigma}^{\tau}\Big( h_r\big(r,X ,Y_r^{r,\alpha}  , Z_r^{r,\alpha},\pi_r,\alpha_r\big) -\partial Y_r^{r,\alpha^\star} \Big) \d r\bigg|\Fc_\sigma\bigg],
\end{align*}
where for every $s\in [0,T]$, $\partial Y^{s,\alpha^\star}$ denotes the solution to \eqref{eq:bsdeparapendix} with $\alpha^\star$. Moreover, $\alpha^\star$ attains the $\es$.

\begin{proof}
We first show the inequality $\geq$. We proceed in 3 steps.  Let $\eps>0$, $0<\ell<\ell_\eps$, and $\Pi^{\ell}$ be a partition of $[\sigma,\tau]$.\medskip

{\bf Step $1$:} From the definition of equilibria we have that for any $\alpha \in \Ac$
\begin{align*}
\Vr_\sigma^\Ar & \geq  \Vr_\sigma^\Ar(\alpha\otimes_{\tau_\smallfont{1}}\alpha^\star) -\eps \ell \geq   \E^{\P}\Big [ Y_\sigma^{\sigma,\alpha\otimes_{\smallfont{\tau}_\tinyfont{1}}\alpha^\smallfont{\star}}-Y_{\tau_\smallfont{1}}^{\tau_\smallfont{1},\alpha\otimes_{\smallfont{\tau}_{\tinyfont1}}\alpha^\smallfont{\star}}+Y_{\tau_\smallfont{1}}^{\tau_\smallfont{1},\alpha\otimes_{\smallfont{\tau}_{\tinyfont1}}\alpha^\smallfont{\star}}\Big|\Fc_\sigma\Big]-\eps \ell.
\end{align*}
Recall that for any $\rho\in \Tc_{0,T}$, $Y_\rho^{\rho,\alpha\otimes_{\smallfont\rho}\alpha^\smallfont{\star}}=Y_\rho^{\rho, \alpha^\smallfont{\star}}$. In light of the arbitrariness of $\alpha\in \Ac$ we obtain

\begin{align*}
\Vr_\sigma^\Ar \geq   \es_{\alpha\in \Ac } \; \E^{\P}\Big [ Y_\sigma^{\sigma,\alpha\otimes_{\smallfont{\tau}_{\tinyfont1}}\alpha^\smallfont{\star}}-Y_{\tau_\smallfont{1}}^{\tau_\smallfont{1},\alpha\otimes_{\smallfont{\tau}_{\tinyfont1}}\alpha^\smallfont{\star}}+\Vr_{\tau_\smallfont{1}}^\Ar \Big|\Fc_\sigma\Big]-\eps \ell,\; \P\as
\end{align*}

{\bf Step $2$:} Let us note that in light of {\bf Step 1}
\begin{align*}
\Vr_\sigma^\Ar & \geq   \es_{\alpha\in \Ac} \;   \E^{\P}\Big [ Y_\sigma^{\sigma,\alpha\otimes_{\smallfont{\tau}_{\tinyfont{1}}}\alpha^\smallfont{\star}}-Y_{\tau_\smallfont{1}}^{\tau_\smallfont{1},\alpha\otimes_{\smallfont{\tau}_{\tinyfont{1}}}\alpha^\smallfont{\star}}+\Vr_{\tau_\smallfont{1}}^\Ar  \Big|\Fc_\sigma\Big]-\eps \ell\\
&= \es_{\alpha\in \Ac } \;  \E^{\P}\Big [ Y_\sigma^{\sigma,\alpha\otimes_{\smallfont{\tau}_{\tinyfont1}}\alpha^\smallfont{\star}}-Y_{\tau_\smallfont{1}}^{\tau_\smallfont{1},\alpha\otimes_{\smallfont{\tau}_{\tinyfont1}}\alpha^\smallfont{\star}}+\es_{\tilde \alpha\in \Ac } \;  \E^{\P}\Big [ Y_{\tau_\smallfont{1}}^{\tau_\smallfont{1},\tilde \alpha\otimes_{\smallfont{\tau}_{\tinyfont2}}\alpha^\smallfont{\star}} -Y_{\tau_\smallfont{2}}^{\tau_\smallfont{2},\tilde \alpha\otimes_{\smallfont{\tau}_{\tinyfont2}}\alpha^\smallfont{\star}}+Y_{\tau_\smallfont{2}}^{\tau_\smallfont{2},\tilde \alpha\otimes_{\smallfont{\tau}_{\tinyfont2}}\alpha^\smallfont{\star}}\Big|\Fc_{\tau_\smallfont{1}}\Big]\Big|\Fc_\sigma\Big]-2\eps \ell\\
&= \es_{\alpha\in \Ac }\;  \E^{\P}\Big [   \Vr_{\tau_\smallfont{2}}^\Ar  +Y_\sigma^{\sigma,\alpha\otimes_{\smallfont{\tau}_{\tinyfont1}}\alpha^\smallfont{\star}}-Y_{\tau_\smallfont{1}}^{\tau_\smallfont{1},\alpha\otimes_{\smallfont{\tau}_{\tinyfont1}}\alpha^\smallfont{\star}} + Y_{\tau_\smallfont{1}}^{\tau_\smallfont{1},\alpha\otimes_{\smallfont{\tau}_{\tinyfont2}}\alpha^\smallfont{\star}} - Y_{\tau_\smallfont{2}}^{\tau_\smallfont{2},\alpha\otimes_{\smallfont{\tau}_{\tinyfont2}}\alpha^\smallfont{\star}}\Big|\Fc_\sigma\Big] -2\eps \ell,
\end{align*}
where the second equality holds in light of \cite[Lemma 3.5]{possamai2015stochastic} as \cite[Assumption 1.1]{possamai2015stochastic} holds under \Cref{assump:datareward}. Iterating the previous argument we obtain that $\P\as$
\begin{align*}
 \Vr_{\sigma}^\Ar  \geq \es_{\alpha\in \Ac}\; \E^{\P}\bigg [ \Vr_{\tau}^\Ar +\sum_{i=0}^{n_\smallfont{\ell}-1} Y_{\tau_\smallfont{i}}^{\tau_\smallfont{i},\alpha\otimes_{\smallfont{\tau}_{\tinyfont{i}\tinyfont{+}\tinyfont{1}}}\alpha^\smallfont{\star}}-Y_{\tau_{\smallfont{i}\smallfont{+}\smallfont{1}}}^{\tau_{\smallfont{i}\smallfont{+}\smallfont{1}},\alpha\otimes_{\smallfont{\tau}_{\tinyfont{i}\tinyfont{+}\tinyfont{1}}}\alpha^\smallfont{\star}}\bigg|\Fc_\sigma\bigg]-n_{\ell}\eps \ell.
\end{align*}

Now, we use the fact that for any $i\in \{0,\dots,n_\smallfont{\ell}-1\}$ and $\alpha\otimes_{\smallfont{\tau}_{\tinyfont{i}\tinyfont{+}\tinyfont{1}}}\alpha^\star$, \Cref{lemma:apendixdpp} implies 
\begin{align*}
 \E^{\P}\Big [ Y_{\tau_\smallfont{i}}^{\tau_\smallfont{i},\alpha\otimes_{\smallfont{\tau}_{\tinyfont{i}\tinyfont{+}\tinyfont{1}}}\alpha^\smallfont{\star}}-Y_{\tau_{\smallfont{i}\smallfont{+}\smallfont{1}}}^{\tau_{\smallfont{i}\smallfont{+}\smallfont{1}},\alpha\otimes_{\smallfont{\tau}_{\tinyfont{i}\tinyfont{+}\tinyfont{1}}}\alpha^\smallfont{\star}}+\int_{\tau_\smallfont{i}}^{\tau_{\smallfont{i}\smallfont{+}\smallfont{1}}}\partial Y_r^{r,\alpha\otimes_{\smallfont{\tau}_{\tinyfont{i}\tinyfont{+}\tinyfont{1}}}\alpha^\smallfont{\star}}\d r\Big|\Fc_\sigma\Big]=\E^{\P}\bigg [ \int_{\tau_\smallfont{i}}^{\tau_{\smallfont{i}\smallfont{+}\smallfont{1}}}  h_r(r,X_{\cdot\wedge r},Y_r^{r,\alpha}  , Z_r^{r,\alpha },\pi_r, \alpha_r) \d r\bigg|\Fc_\sigma\bigg].
\end{align*}

Replacing in the previous expression we obtain that $\P\as$
\begin{align}\label{eq:dppaux}
\Vr_\sigma^\Ar \geq  \es_{\alpha\in \Ac} \;  \E^{\P}\bigg [ \Vr_{\tau}^\Ar  + \int_{\sigma}^{\tau} h_r(r,X_{\cdot\wedge r},Y_r^{r,\alpha}  , Z_r^{r,\alpha }, \pi_r,\alpha_r)\d r -\sum_{i=0}^{n_\smallfont{\ell}-1} \int_{\tau_\smallfont{i}}^{\tau_{\smallfont{i}\smallfont{+}\smallfont{1}}}\partial Y_r^{r,\alpha{\otimes}_{\smallfont{\tau}_{\tinyfont{i}\tinyfont{+}\tinyfont{1}}}\alpha^\smallfont{\star}} \d r -n_{\ell}\eps \ell\bigg|\Fc_\sigma\bigg].
\end{align}

\medskip

{\bf Step $3$:} Let $i\in \{0,\dots,n_{\ell}-1\}$. In light of \Cref{assump:datareward}, the stability of the system of BSDE defined by \eqref{eq:bsdevalapendix} and \eqref{eq:bsdeparapendix},  see \cite[Proposition 6.4]{hernandez2020unified}, yields there exists a constant $C>0$ such that
\begin{align*}
\bigg\| \int_{\tau_\smallfont{i}}^{\tau_{\smallfont{i}\smallfont{+}\smallfont{1}}}  \partial Y_r^{r,\alpha\otimes_{\smallfont{\tau}_{\tinyfont{i}\tinyfont{+}\tinyfont{1}}}\alpha^\smallfont{\star}}  -  \partial Y_r^{r,\alpha^\smallfont{\star}}   \d r\bigg\|_{\Lc^\smallfont{2}} \leq \ell \Big\|\partial Y^{\alpha\otimes_{\smallfont{\tau}_{\tinyfont{i}\tinyfont{+}\tinyfont{1}}}\alpha^\smallfont{\star}}  -  \partial Y^{\alpha^\smallfont{\star}} \Big\|_{\S^{\smallfont{2}\smallfont{,}\smallfont{2}}} \leq \ell C \E\bigg[\bigg(\int_{\tau_\smallfont{i}}^{\tau_{\smallfont{i}\smallfont{+}\smallfont{1}}}|\alpha_r-\alpha^\star_r|\d r\bigg)^2\bigg],
\end{align*}
which leads to
\begin{align*}
\bigg\| \sum_{i=0}^{n_\smallfont{\ell}-1} \int_{\tau_\smallfont{i}}^{\tau_{\smallfont{i}\smallfont{+}\smallfont{1}}}  \partial Y_r^{r,\alpha\otimes_{\smallfont{\tau}_{\tinyfont{i}\tinyfont{+}\tinyfont{1}}}\alpha^\smallfont{\star}}  -  \partial Y_r^{r,\alpha^\smallfont{\star}}   \d r\bigg\|_{\Lc^\smallfont{2}}\leq \ell C\E\bigg[\bigg(\int_\sigma^\tau | \alpha_r-\alpha^\star_r|\d r\bigg)^2\bigg]\xrightarrow{\ell\to 0} 0.
\end{align*}
By choosing an appropriate partition $\Pi^\ell$ and applying the dominated convergence theorem we obtain that 
\begin{align*}
I(n_\smallfont{\ell}):=\sum_{i=0}^{n_\smallfont{\ell}-1} \int_{\tau_\smallfont{i}}^{\tau_{\smallfont{i}\smallfont{+}\smallfont{1}}}\partial Y_r^{r,\alpha\otimes_{\smallfont{\tau}_{\tinyfont{i}\tinyfont{+}\tinyfont{1}}}\alpha^\smallfont{\star}} \d r \xrightarrow{\ell_\smallfont{\eps}\rightarrow 0} & \int_{\sigma}^{\tau}  \partial Y_r^{r,\alpha^\star}   \d r, \; \P\as
\end{align*}

Back in \eqref{eq:dppaux} we obtain
\begin{align*}
\Vr_\sigma^\Ar= \es_{\alpha\in \Ac}\;  \E^{\P}\bigg [\Vr_\tau^\Ar+ \int_{\sigma}^{\tau}\Big( h_r\big(r,X ,Y_r^{r,\alpha}  , Z_r^{r,\alpha},\pi_r,\alpha_r\big) -\partial Y_r^{r,\alpha^\smallfont{\star}} \Big) \d r\bigg|\Fc_\sigma\bigg],\; \P\as
\end{align*}

Lastly, we show that the equality is attained by $\alpha^\star\in \Ec$. Indeed, note that \eqref{eq:propertybsdeapendix} implies
\begin{align*}
\Vr_\sigma^\Ar=  \E^{\P}\bigg [\Vr_\tau^\Ar+ \int_{\sigma}^{\tau}\Big( h_r\big(r,X ,Y_r^{r,\alpha^\smallfont{\star}}  , Z_r^{r,\alpha^\smallfont{\star}},\pi_r,\alpha^\star_r\big) -\partial Y_r^{r,\alpha^\smallfont{\star}} \Big) \d r\bigg|\Fc_\sigma\bigg],  \; \P\as
\end{align*}
\end{proof}
\end{theorem}

Let us recall that the Hamiltonian associated to $h$ is given by
\[
H_t(x,y,z,p)=\sup_{a\in A} h_t(t,x,y,z,p,a),\; (t,x,y,z,p)\in [0,T]\times \Omega\times \R\times \R^n\times\R.
\]
Our standing assumptions on $H$ are the following.
\begin{assumption}\label{assump:uniquemaxdpp}
\begin{enumerate}[label=$(\roman*)$, ref=.$(\roman*)$,wide,  labelindent=0pt]
\item \label{assump:uniquemaxdpp:i}For any $(t,x)\in[0,T]\times\Omega$, the map $\R\times \R^n \ni (y,z)\longmapsto H_t(x,y,z,p)$ is uniformly Lipschitz-continuous, \emph{i.e.} there is $C>0$ such that for any $ (t,x,p,{\rm y},{\rm \tilde y},{\rm z},{\rm \tilde z})\in [ 0,T]\times\Omega \times \R^3\times(\R^n)^2$
\[
\big|H_t(x,{\rm y},{\rm z},p)-H_t(x,{\rm \tilde y},{\rm \tilde z},p)\big|\leq C\big( |{\rm y}-{\rm \tilde y}|+ |\sigma_t(x)^\t({\rm z}- {\rm \tilde z})|\big);
\]
\item \label{assump:uniquemaxdpp:ii}there exists a unique Borel-measurable map $a^\star:[0,T]\times\Omega \times\R\times \R^n\longrightarrow A$ such that
\[
H_t(x,{\rm y},{\rm z},p)=h_t\big(t,x,{\rm y},{\rm z},p,a^\star(t,x,{\rm y},{\rm z},p)\big),\; \forall(t,x,{\rm y},{\rm z})\in[0,T]\times\Omega \times \R \times\R^n.
\]
\item \label{assump:uniquemaxdpp:iii}For any $(t,x)\in[0,T]\times\Omega$, the map $\R\times \R^n \ni ({\rm y},{\rm z})\longmapsto a^\star(t,x,{\rm y},{\rm z},p)$ is uniformly Lipschitz-continuous, \emph{i.e.} there is $C>0$ such that for any $ (t,x,p,{\rm y},{\rm \tilde y},{\rm z},{\rm \tilde z})\in [ 0,T]\times\Omega \times \R^3\times(\R^n)^2$.
\[
\big|a^\star(t,x,{\rm y},{\rm z},p)-a^\star(t,x,{\rm \tilde y},{\rm \tilde z},p)\big|\leq C\big( |{\rm y}-{\rm \tilde y}|+|\sigma_t(x)^\t({\rm z}-{\rm \tilde z})|\big);
\]
\item \label{assump:uniquemax:iv} $(\tilde H,\tilde a^\star)\in \mathbb{L}^{2}(\R)\times \mathbb{L}^{2}(\R)$, where $\big(\tilde H_\cdot,\tilde a^\star_\cdot \big):=\big(H_\cdot(\cdot,0,0,\pi_\cdot),a^\star_\cdot(\cdot,0,0,\pi_\cdot)\big)$.
\end{enumerate}
\end{assumption}

With this we introduce the system defined for any $s\in [0,T]$ by
\begin{align}\label{eq:systemapp}\tag{H}\begin{split}
Y_t&=\eta(T,\xi)+\int_t^T \big(H_r(X ,Y_r,Z_r,\pi_r)-\partial Y_r^{r}\big)\d r -\int_t^T Z_r\cdot \d X_r, \; t\in [0,T],\; \P\as\\
Y_t^{s} &= \eta(s, \xi) +\int_t^T  h_r^\star\big(s,X ,Y_r^{s}  , Z_r^{s}, Y_r,Z_r,\pi_r\big) \d r-\int_t^T  Z_r^{s} \cdot \d X_r, \; t\in [0,T],\; \P\as\\
\partial Y_t^{s}&=\partial_s \eta(s,\xi)+ \int_t^T \nabla h_r^\star(s,X ,\partial Y_r^{s},\partial Z_r^{s}, Y_r^{s} Z_r^{s}, Y_r,Z_r,\pi_r\big)  \d r - \int_t^T\partial Z_r^{s}\cdot \d X_r, \; t\in [0,T],\; \P\as
\end{split}
\end{align}
We will say $(\Yc,\Zc,Y,Z,\partial Y,\partial Z)\in \mathfrak{H}$ is a solution to the system whenever \eqref{eq:systemapp} is satisfied. In light of \Cref{thm:dpp}, given $\alpha^\star\in \Ec$ it is reasonable to associate the value along the equilibria with a BSDE whose generator is given, partially, by $H$. This is the purpose of the next result. 

\begin{theorem}[Necessity]\label{thm:necessity}
Let Assumptions {\rm \ref{assump:datadpp}} and {\rm \ref{assump:uniquemaxdpp}} hold and $\alpha^\star \in \Ec$. Then,  one can construct a solution to \eqref{eq:systemapp}.
\begin{proof}
Given $\alpha^\star\in \Ec$, \Cref{assump:datadpp} guarantees that the processes $(Y^{\alpha^\smallfont{\star}},Z^{\alpha^\smallfont{\star}})$ and $(\partial Y^{\alpha^\smallfont{\star}},\partial Z^{\alpha^\smallfont{\star}})$ solution to \eqref{eq:bsdevalapendix} and \eqref{eq:bsdeparapendix}, respectively, are well-defined. Moreover, the processes $\big((Y_t^{t,\alpha^\smallfont{\star}})_{t\in [0,T]},(Z_t^{t,\alpha^\smallfont{\star}})_{t\in [0,T]}\big)$ are well-defined as elements of $\S^2\times \H^2$, see \cite[Lemma 6.2]{hernandez2020unified}. Given $(\partial Y_t^{t,\alpha^\smallfont{\star}})_{t\in [0,T]}\in\S^2$, for any $\alpha\in \Ac$, we can define the processes $(\Yc^\alpha,\Zc^\alpha)\in \S^2\times \H^2$ solution to
\begin{align*}
\Yc_t^{\alpha}=\eta(T,\xi)+\int_t^T \Big( h_r(r,X , \Yc_r^{\alpha},\Zc_r^{\alpha},\pi_r,\alpha_r)-\partial Y_r^{r,\alpha^\star}\Big)\d r-\int_t^T \Zc_r^{\alpha}\cdot \d X_r, \; t\in [0,T],\; \P\as
\end{align*}
We now note that under \Cref{assump:datadpp} the classic comparison result for BSDEs holds, see for instance \cite[Theorem 4.4.1]{zhang2017backward}. Then, it follows from \Cref{thm:dpp} that the pair $\big((Y_t^{t,\alpha^\smallfont{\star}})_{t\in [0,T]},(Z_t^{t,\alpha^\smallfont{\star}})_{t\in [0,T]}\big)$ solves the {\rm BSDE}
\begin{align*}
Y_t=\eta(T,\xi)+\int_t^T \big(H_r(X ,Y_r,Z_r,\pi_r)-\partial Y_r^{r,\alpha^\smallfont{\star}}\big)\d r -\int_t^T Z_r\cdot \d X_r, \; t\in [0,T],\; \P\as
\end{align*}

Moreover, the second part of the statement of \Cref{thm:dpp} implies that $\alpha^\star=a^\star(\cdot,X,Y_\cdot,Z_\cdot,\pi_\cdot)$, $\d t\otimes \d \P\ae$ Consequently, $(Y^{\alpha^\smallfont{\star}},Z^{\alpha^\smallfont{\star}})$ and $(\partial Y^{\alpha^\smallfont{\star}},\partial Z^{\alpha^\smallfont{\star}})$ define a solution to the second and third equations in \eqref{eq:systemapp}, respectively.
\end{proof}
\end{theorem}

We close this section with a verification theorem for equilibria.
\begin{theorem}[Verification]\label{thm:verification}
Let Assumptions {\rm \ref{assump:datadpp}} and {\rm \ref{assump:uniquemaxdpp}} hold. Let $(\Yc,\Zc,Y,Z,\partial Y,\partial Z)\in \mathfrak{H}$ be a solution to \eqref{eq:systemapp} with $\alpha^\star:=a^\star(\cdot,X,\Yc_\cdot,\Zc_\cdot,\pi_\cdot)$. Then,  $\alpha^\star\in \Ec$ and
\[ \Vr^\Ar_t= \Yc_t,\; \P\as\]
\begin{proof}
We verify the definition of an equilibria. Let $\eps>0$, $(t,\ell)\in [0,T]\times(0,\ell_\eps)$ with $\ell_\eps$ to be chosen. Let $(\Yc^{\alpha\otimes_{\smallfont{t}\smallfont{+}\smallfont{\ell}} \alpha^\smallfont{\star}},\Zc^{\alpha\otimes_{\smallfont{t}\smallfont{+}\smallfont{\ell}} \alpha^\smallfont{\star}})\in \S^2\times \H^2$ be the solution, which exists in light of {\rm \Cref{assump:datadpp}}, to \eqref{eq:bsdevalapendix} with action $\alpha\otimes_{t+\ell} \alpha^\star$, that is to say, $\P$--a.s.
\begin{align*}
\Yc_t^{s,\alpha\otimes_{\smallfont{t}\smallfont{+}\smallfont{\ell}} \alpha^\smallfont{\star}}=\eta(s,\xi)+\int_t^T  h_r\big(s,X , \Yc_r^{s,\alpha\otimes_{\smallfont{t}\smallfont{+}\smallfont{\ell}} \alpha^\smallfont{\star}},\Zc_r^{s,\alpha\otimes_{\smallfont{t}\smallfont{+}\smallfont{\ell}} \alpha^\smallfont{\star}},\pi_r,(\alpha\otimes_{t+\ell} \alpha^\star)_r\big)\d r-\int_t^T \Zc_r^{s,\alpha\otimes_{\smallfont{t}\smallfont{+}\smallfont{\ell}} \alpha^\smallfont{\star}}\cdot \d X_r, \; t\in [0,T].
\end{align*}
It then follows that
\begin{align*}
\Yc_t^{t,\alpha\otimes_{\smallfont{t}\smallfont{+}\smallfont{\ell}} \alpha^\smallfont{\star}}&=\eta(T,\xi)+\int_t^T \Big( h_r(r,X , \Yc_r^{\alpha\otimes_{\smallfont{t}\smallfont{+}\smallfont{\ell}} \alpha^\smallfont{\star}},\Zc_r^{\alpha\otimes_{\smallfont{t}\smallfont{+}\smallfont{\ell}} \alpha^\smallfont{\star}},\pi_r,(\alpha\otimes_{t+\ell} \alpha^\star)_r\big)-\partial Y_r^{r,\alpha\otimes_{\smallfont{t}\smallfont{+}\smallfont{\ell}} \alpha^\smallfont{\star}}\Big)\d r\\
&\quad-\int_t^T \Zc_r^{r,\alpha\otimes_{\smallfont{t}\smallfont{+}\smallfont{\ell}} \alpha^\smallfont{\star}}\cdot \d X_r.
\end{align*}
Now
\begin{align*}
\Yc_t-\Yc_t^{t,\alpha\otimes_{\smallfont{t}\smallfont{+}\smallfont{\ell}} \alpha^\smallfont{\star}}&=\int_t^T \Big(H_r(X ,Y_r,Z_r,\pi_r)-h_r(r,X_{\cdot\wedge r}, \Yc_r^{\alpha\otimes_{\smallfont{t}\smallfont{+}\smallfont{\ell}} \alpha^\smallfont{\star}},\Zc_r^{\alpha\otimes_{\smallfont{t}\smallfont{+}\smallfont{\ell}} \alpha^\smallfont{\star}},\pi_r,(\alpha\otimes_{t+\ell} \alpha^\star)_r\big)\Big)\d r \\
&\quad-\int_t^T\partial Y_r^{r,\alpha^\smallfont{\star}}-\partial Y_r^{r,\alpha\otimes_{\smallfont{t}\smallfont{+}\smallfont{\ell}} \alpha^\smallfont{\star}}\d r -\int_t^T\big( Z_r-\Zc_r^{r,\alpha\otimes_{\smallfont{t}\smallfont{+}\smallfont{\ell}} \alpha^\smallfont{\star}}\big)\cdot \d X_r\\
&\geq \int_{t+\ell}^T \Big(H_r(X ,Y_r,Z_r,\pi_r)-\partial Y_r^{r,\alpha^\smallfont{\star}}-h_r(r,X_{\cdot\wedge r}, \Yc_r^{\alpha\otimes_{\smallfont{t}\smallfont{+}\smallfont{\ell}} \alpha^\smallfont{\star}},\Zc_r^{\alpha\otimes_{\smallfont{t}\smallfont{+}\smallfont{\ell}} \alpha^\smallfont{\star}},\pi_r,\alpha^\star_r\big)+\partial Y_r^{r,\alpha\otimes_{\smallfont{t}\smallfont{+}\smallfont{\ell}} \alpha^\smallfont{\star}}\Big)\d r \\
&\quad-\int_t^{t+\ell} \partial Y_r^{r,\alpha^\smallfont{\star}}-\partial Y_r^{r,\alpha\otimes_{\smallfont{t}\smallfont{+}\smallfont{\ell}} \alpha^\smallfont{\star}}\d r -\int_t^T\big( Z_r-\Zc_r^{r,\alpha\otimes_{\smallfont{t}\smallfont{+}\smallfont{\ell}} \alpha^\smallfont{\star}}\big)\cdot \d X_r\\
&= \int_t^{t+\ell} \partial Y_r^{r,\alpha\otimes_{\smallfont{t}\smallfont{+}\smallfont{\ell}} \alpha^\smallfont{\star}}-\partial Y_r^{r,\alpha^\smallfont{\star}}\d r -\int_t^T\big( Z_r-\Zc_r^{r,\alpha\otimes_{\smallfont{t}\smallfont{+}\smallfont{\ell}} \alpha^\smallfont{\star}}\big)\cdot \d X_r,
\end{align*}
where the inequality follows by definition of $H$ and $\alpha^\star$ and \Cref{assump:uniquemaxdpp}. The second equality follows from the fact that the first term cancels on $[t+\ell,T]$, see \Cref{lemma:apendixdpp}. Taking expectation we find
\begin{align*}
\Vr^\Ar_t-\Vr^\Ar_t(\alpha\otimes_{t+\ell}\alpha^\star)=\E\Big[\Yc_t-\Yc_t^{t,\alpha\otimes_{\smallfont{t}\smallfont{+}\smallfont{\ell}} \alpha^\smallfont{\star}}\Big|\Fc_t \Big]\geq \E\bigg[ \int_t^{t+\ell} \partial Y_r^{r,\alpha\otimes_{\smallfont{t}\smallfont{+}\smallfont{\ell}} \alpha^\smallfont{\star}}-\partial Y_r^{r,\alpha^\smallfont{\star}}\d r\bigg|\Fc_t\bigg].
\end{align*}
By \cite[Proposition 6.4]{hernandez2020unified} we find that
\begin{align*}
\bigg\| \int_{t}^{t+\ell}  \partial Y_r^{r,\alpha\otimes_{\smallfont{t}\smallfont{+}\smallfont{\ell}} \alpha^\smallfont{\star}}  -  \partial Y_r^{r,\alpha^\smallfont{\star}}   \d r\bigg\|_{\Lc^\smallfont{2}} 
\leq \ell \Big\|\partial Y^{\alpha\otimes_{{\smallfont{t}\smallfont{+}\smallfont{\ell}} \alpha^\smallfont{\star}}}  -  \partial Y^{\alpha^\smallfont{\star}} \Big\|_{\S^{\smallfont{2}\smallfont{,}\smallfont{2}}}
 \leq \ell C \E\bigg[\bigg( \int_{t}^{t+\ell}|\alpha_r-\alpha^\star_r|\d r\bigg)^2\bigg].
\end{align*}
By the boundedness of the action set, we may choose $\ell_\eps$ such that the last term above is smaller that $\ell\eps$. With this, we conclude $\alpha^\star\in \Ec$.  The second part of the statement follows from the fact that $\Yc_t=Y_t^{t,\alpha^\smallfont{\star}}$, $\P\as$.
\end{proof}
\end{theorem}

\section{On forward stochastic Volterra integral equations}\label{sec:appforwardVolterra}
We are given a jointly measurable mapping $h$, a processes $Z$, and a  family $(Y_0^s)_{s\in[0,T]}\in \Ic$ such that for any $(y,u)\in \R^2 \times(\R^{n})^2 $
\begin{align*} 
 &h:  [0,T]^2 \times \Omega\times \R  \times \R^{ n }  \longrightarrow\R ,\;   h_\cdot(\cdot ,y,u )\in \Pc_{{\rm prog}}(\R,\F).
\end{align*}
To ease the notation, we drop the dependence of $h$ on $((\pi_t)_{t\in [0,T]},(Z^s_t)_{(s,t)\in[0,T]},(Z_t^t)_{t\in[0,T]})$ since in the analysis in \Cref{sec:restrictedcontracts} these processes are given. Moreover, we work under the following set of assumptions.
\begin{assumption}\label{assump:fsvieappendix}

\begin{enumerate}[label=$(\roman*)$, ref=.$(\roman*)$,wide,  labelindent=0pt]
\item  \label{assump:fsvieappendix:i} $(s,y) \longmapsto  h_t(s,x,y,u)$ $($resp. $s\longmapsto Y^s_0(x))$ is continuously differentiable, uniformly in $(t,x,u)$ $($resp. in $x)$. Moreover, the mapping $\nabla h:[0,T]^2\times \Omega \times (\R \times \R^{n} )^2    \longrightarrow \R$ defined by 
\[ \nabla h_t (s,x,{\rm u}, y,u):=\partial_s h_t(s,x,y,u)+\partial_y h_t(s,x,y,u){\rm u},\; \nabla h _\cdot(s,\cdot, {\rm u} ,y,u)\in \Pc_{\rm prog}(\R,\F);\]
\item   \label{assump:fsvieappendix:ii} for $\varphi\in \{h,\partial_s h\}$, $(y,u)\longmapsto  \partial_s \varphi_t(s,x,y,u)$ is uniformly Lipschitz-continuous, \emph{i.e.} there exists some $C>0$ such that $\forall  (s,t,x,y,\tilde y,u,\tilde u)$,
\begin{align*}
 |\varphi_t(s,x,y,u)-\varphi_t(s,x,\tilde y,\tilde u)|\leq C\big(|y-\tilde y|+|u-\tilde u|\big).
\end{align*}
\item \label{assump:fsvieappendix:iii} $ (Y_0,\partial_s Y_0)\in\big( \Ic \big)^2$,  $Z\in \overline{\H}^{_{\raisebox{-1pt}{$ \scriptstyle 2,2$}}}$, $(\tilde h_\cdot(s), \nabla \tilde h_\cdot(s)):=(h_\cdot(s,\cdot,{\bf 0}),\partial_s h_\cdot(s,\cdot,{\bf 0})) \in  \big(\mathbb{L}^{2,2}\big)^2$, for ${\bf 0}:=(y,u)|_{(0,0)}$.
\end{enumerate}
\end{assumption}

%$\overline{\H}^{_{\raisebox{-1pt}{$ \scriptstyle 2,2$}}}(\R^n)$ 
We are interested in establishing the well-posedness of the FSVIE
\begin{align}\label{eq:fsvieapp}
Y_t^{s}=Y_0^s+\int_0^t h_r\big(s,X ,Y_r^{s},Y_r^{r} \big) \d r+\int_0^t  Z_r^s  \cdot \d X_r, \;t\in [0,T],\;\P\as,\; s\in [0,T].
\end{align}
To alleviate the notation we write $Y$ instead of $Y^Z$ in the previous equation. 

\begin{definition}\label{def:solfsvie}
We say $Y$ is a solution to the {\rm FSVIE} \eqref{eq:fsvieapp} if $Y$ satisfies equation \eqref{eq:fsvieapp} and $Y\in \S^{2,2}$.
\end{definition}

\begin{remark}
We remark that in light of the pathwise continuity of $Y^s$ for every $s\in [0,T]$ the process $(Y_t^t)_{t\in[0,T]}$ is well defined $\d t\otimes \d \P\ae$ on $[0,T]\times \Omega$.
\end{remark}
We begin presenting a priori estimates for solutions of \eqref{eq:fsvieapp}. These can be recover from the arguments in \cite{zhang2017backward}.
\begin{lemma}\label{lemma:apriorisfsvie}
Let $Y$ be a solution to \eqref{eq:fsvieapp}, there exists a constant $C>0$ such that
\[
\|Y\|_{\S^{\smallfont{2}\smallfont{,}\smallfont{2}}}^2\leq C\big( \|Y_0\|_{\Lc^{\smallfont{2}\smallfont{,}\smallfont{2}}}^2+\|\tilde h\|_{\mathbb{L}^{\smallfont{1}\smallfont{,}\smallfont{2}\smallfont{,}\smallfont{2}}}^2+\|Z\|_{\H^{\smallfont{2}\smallfont{,}\smallfont{2}}}^2\big).
\]
Moreover, for $Y^i$ solution to \eqref{eq:fsvieapp} with data $(Y_0^i,h^i,Z^i)$ satisfying \eqref{assump:fsvieappendix} for $i\in\{1,2\}$ there exists $C>0$ such that
\begin{align*}
\|Y^1-Y^2\|_{\S^{\smallfont{2}\smallfont{,}\smallfont{2}}}^2\leq C\big( \|Y_0^1-Y_0^2\|_{\Lc^{\smallfont{2}\smallfont{,}\smallfont{2}}}^2+\|\tilde h^1-\tilde h^2\|_{\mathbb{L}^{\smallfont{1}\smallfont{,}\smallfont{2}\smallfont{,}\smallfont{2}}}^2+\|Z^1-Z^2\|_{\H^{\smallfont{2}\smallfont{,}\smallfont{2}}}^2\big)
\end{align*}
\begin{proof}
Let us observe that the continuity of the application $s\longmapsto \|\Delta Y^{s}\|_{\S^{2}}$ implies
\begin{align}\label{eq:auxapriories}
\E\bigg[ \int_0^T \e^{-c r} |\Delta Y_r^{r}|^2\d r\bigg]\leq \int_0^T \E\bigg[\sup_{u\in [0,T]} \e^{-c u} |\Delta Y_u^{r}|^2\bigg] \d r\leq T\|\Delta Y\|^2_{\S^{\smallfont{2}\smallfont{,}\smallfont{2}}}.
\end{align}
With this, the proof of both statements can obtained following the line of \cite[Theorem 3.2.2 and Theorem 3.2.4]{zhang2017backward}.
\end{proof}
\end{lemma}

We are now ready to establish the well-posedness of \eqref{eq:fsvieapp}.

\begin{proposition}\label{prop:wpfsvie}
Let {\rm \Cref{assump:fsvieappendix}} hold. There is a unique solution to \eqref{eq:fsvieapp}.
\begin{proof}
Uniqueness follows from \Cref{lemma:apriorisfsvie}. We use a Picard iteration argument. Let $Y^{s,0}_\cdot=Y^s_0, s\in [0,T]$ and 
\begin{align*}
Y_t^{s,n+1}=Y^s_0+\int_0^t h_r\big(s,X ,Y_r^{s,n},Y_r^{r,n} \big) \d r+\int_0^t  Z_r^s  \cdot \d X_r, \;t\in [0,T]
\end{align*}
We note that $Y^n\in \S^{2,2}$ for $n\geq 0$. Indeed, the result holds for $Y^0$ and the process $(Y_t^{t,0})_{t\in [0,T]}\in \mathbb{L}^{1,2}$ is well-defined. Inductively, in light of \Cref{assump:fsvieappendix}, the fact that $Z\in \H^{2,2}$ and $(Y^{t,n}_t)_{t\in[0,T]}\in \mathbb{L}^{1,2}$, see \eqref{eq:auxapriories}, yields $Y^{s,n+1}\in \S^{2}$ for every $s\in [0,T]$.  The continuity of $s\longmapsto \|Y^{s,n}\|_{\S^{2}}$, \Cref{assump:fsvieappendix} together with \Cref{lemma:apriorisfsvie} guarantees $Y^{n+1}\in \S^{2,2}$. Moreover, the pathwise continuity $Y^{s,n+1}$ for any $s\in [0,T]$ guarantees $(Y_t^{t,n+1})_{t\in [0,T]}$ is well-defined $\d t\otimes \d \P\ae$\medskip

Let $\Delta Y^n:= Y^n-Y^{n-1}$. Then, for any $s\in [0,T]$
\begin{align*}
 \Delta Y^{s,n+1}_t = \int_0^t \big( h_r(s,X, Y_r^{s,n},Y_r^{r,n})- h_r(s,X_{\cdot \wedge r}, Y_r^{s,n-1},Y_r^{r,n-1})\big)\d r.
\end{align*}
The inequality $2ab\leq \eps^{-1}a^2+\eps b^2$ for any $\eps>0$ yields that for any $\eps>0$ there exists $C(\eps)>0$ such that 
\begin{align*}
\e^{-c t}  |\Delta Y^{s,n+1}_t|^2=& \int_0^t \e^{-c r}\big( 2 |\Delta Y^{s,n+1}_r| ( h_r(s,X , Y_r^{s,n},Y_r^{r,n})- h_r(s,X , Y_r^{s,n-1},Y_r^{r,n-1}))-c |\Delta Y^{s,n+1}_r|  \big)\d r \\
\leq &  \int_0^t \e^{-c r} \big( |\Delta Y^{s,n+1}_r|^2(C(\eps)-c)+ \eps |\Delta Y_r^{s,n}|^2+\eps |\Delta Y_r^{r,n}|^2\big) \d r,
\end{align*}

we then find that for $c>C(\eps)$
\begin{align}\label{eq:auxwpfsvie}
\E\bigg[ \sup_{t\in [0,T]} \e^{-c t}  |\Delta Y^{s,n+1}_t|^2\bigg] \leq \eps \E\bigg[ \int_0^T \e^{-c r} \big(  |\Delta Y_r^{s,n}|^2+ |\Delta Y_r^{r,n}|^2\big) \d r\bigg].
\end{align}

Now,  as $\Delta Y^{s,n}\in\S^{2,2}$ we may use \eqref{eq:auxapriories} back in \eqref{eq:auxwpfsvie} and obtain that for $\tilde \eps=\frac1{8T}$,  $c>C(\tilde \eps)$, we have$  \|\Delta Y^{n+1}\|^2_{\S^{\smallfont{2}\smallfont{,}\smallfont{2}\smallfont{,}\smallfont{c}}}\leq 4^{-1} \|\Delta Y^{n}\|^2_{\S^{\smallfont{2}\smallfont{,}\smallfont{2}\smallfont{,}\smallfont{c}}}$. Inductively, we find that for all $n\geq 1$,
$  \|\Delta Y^{n}\|^2_{\S^{\smallfont{2}\smallfont{,}\smallfont{2}\smallfont{,}\smallfont{c}}}\leq C 4^{-n}$. Thus, for $m>n$,
\begin{align*}
\| Y^{m}-Y^{n}\|_{\S^{\smallfont{2}\smallfont{,}\smallfont{2}\smallfont{,}\smallfont{c}}}\leq \sum_{k=n+1}^m \|\Delta Y^k\|_{\S^{\smallfont{2}\smallfont{,}\smallfont{2}\smallfont{,}\smallfont{c}}}\leq  \sum_{k=n+1}^m \frac{C}{ 2^{k}}\leq \frac{C}{2^n}.
\end{align*}
Hence there is $Y\in \S^{2,2}$ such that $Y^n\xrightarrow{n\longrightarrow 0}Y$.
\end{proof}
\end{proposition}

We now establish a result regarding the differentiability of \eqref{eq:fsvieapp}. Recall that for $Z\in \overline{\H}^{_{\raisebox{-1pt}{$ \scriptstyle 2,2$}}}$ there exists by definition, see \Cref{sec:spacesandassupm},  a process $\partial Z$ which can be interpreted as the derivative of the mapping $([0,T],\Bc([0,T]))\longrightarrow (\H^{2,2},\| \cdot \|_{\H^{\smallfont2}}):s\longmapsto Z^s$. 

\begin{proposition}\label{prop:wpfsviepartial}
Let {\rm \Cref{assump:fsvieappendix}} hold and $Y\in \S^2$ be the solution to \eqref{eq:fsvieapp}. There is a unique process $\partial Y\in \S^{2,2}$ that satisfies
\begin{align}\label{eq:partialfsvieapp}
\partial Y_t^{s}=\partial_s Y_0^s+\int_0^t \nabla h_r\big(s,X ,\partial Y_r^s, Y_r^{s},Y_r^{r} \big) \d r+\int_0^t  \partial Z_r^s  \cdot \d X_r, \;t\in [0,T],\;\P\as,\; s\in [0,T].
\end{align}
Moreover
\[ \int_0^s \partial Y^u \d u=Y^s-Y^0, \text{ in } \S^{2}(\R)\]
\begin{proof}
Note that given the pair $(Y,Z)\in \S^{2,2}\times \overline{\H}^{_{\raisebox{-1pt}{$ \scriptstyle 2,2$}}}$,  \Cref{assump:fsvieappendix} guarantees there is $C>0$ such that
\[ \sup_{s\in [0,T]} \bigg(\int_0^T \big|\nabla \tilde h_r(s,X,0,Y_r^s,Y_r^{r}) \big| \d r\bigg)^2\leq C \Big( \|\partial_s \tilde h\|_{\mathbb{L}^{\smallfont{1}\smallfont{,}\smallfont{2}\smallfont{,}\smallfont{2}}}^2 + \| Y\|_{\S^{\smallfont{2}\smallfont{,}\smallfont{2}}}^2\Big)<\infty.\]
We now note that \Cref{assump:fsvieappendix}\ref{assump:fsvieappendix:ii} guarantees ${\rm u}\longmapsto \nabla h_t(s,x,{\rm u},y,u)$ is Lipschitz uniformly in $(s,t,x,y,u)$. Therefore,  \Cref{prop:wpfsvie} guarantees there is a unique solution $\partial Y\in \S^{2,2}$. The second part of the statement, follows arguing as in \cite[Lemma 6.1]{hernandez2020unified} in light of the stability result in \Cref{lemma:apriorisfsvie} and the fact $Z\in \overline{\H}^{_{\raisebox{-1pt}{$ \scriptstyle 2,2$}}}$.
\end{proof}
\end{proposition}

\section{Proofs of Section \ref{sec:sbexamples}}\label{sec:apenexamples}

\subsection{Proof of Proposition {\rm \ref{prop:sol2Bra1}}}

We first note that, $ \P\as$
\begin{align*}
X_T- \widehat Y_T^{y_\smallfont{0},Z} = &\; x_0-\widehat Y_0+\int_0^T \bigg( \lambda^\star_r(X_{\cdot\wedge r},\widehat Z_r) (1-\widehat Z_r)+\widehat H_r(X_{\cdot \wedge r},\widehat Z_r)-\frac{\gamma_{\rm A}}2 |\sigma_r^\t(X_{\cdot\wedge r}) \widehat Z_r|^2-\frac{1}{\gamma_\Ar}\frac{f'(T-r)}{ f(T-r)}  \bigg)\d r \\
&+\int_0^T  \sigma_r^\t(X_{\cdot\wedge r})(1-\widehat Z_r)   \d B^{a^\star(\hat Z)}_r,
\end{align*}
so that
\begin{align*}
\Ur_\Pr(X_T- \widehat Y_T^{y_0,Z}) =C_{\widehat Y_0}M_T\exp\bigg(-\gamma_\Pr  \int_0^T G_r(X_{\cdot \wedge r},\widehat Z_r) \d r\bigg) ,\; \P\as
\end{align*}
where
\[ G_t(x,z):= \lambda^{\star}_r(x,z) -{k^{o}_r}^\star(x,z) -\frac{\gamma_A}2 |\sigma_r^\t(x) z|^2-\frac{\gamma_{\rm P}}{2}|\sigma_r^\t(x)(1-z) |^2-\frac{1}{\gamma_\Ar}\frac{f'(T-r)}{ f(T-r)},\]
and $M$ denotes the supermartingale
\[M_t:=\exp\bigg(- \gamma_\Pr\int_0^t  \sigma_r^\t(X_{\cdot\wedge r})(1-\widehat Z_r) \cdot    \d B_r^{a^\star(\hat Z)}-\frac{\gamma_\Pr^2}2 \int_0^t |\sigma_r^\t(X_{\cdot\wedge r})(1-\widehat Z_r)|^2\d r\bigg), \; t\in [0,T].\]
Indeed,  $M$ is a local martingale that is bounded from below and $M_0=1$. Consequently,
\begin{align*}
&\E^{\P^\star(Z)}\big [ {\rm U}_{\rm P} \big( X_T- \widehat Y_T^{y_\smallfont{0},Z}\big)\big] \leq C_{\widehat Y_0} \E^{\P}\bigg[\exp\bigg(-\gamma_{\rm P}\int_0^T G_r(X_{\cdot \wedge r},\widehat Z_r) \d r \bigg)\bigg].
\end{align*}

Now, under assumptions $(i)$ and $(ii)$ in the statement it is clear that
\begin{align*}
{\rm V}_{\rm P}\leq C_{\widehat R_0} {f(T)}^\frac{\gamma_\smallfont{\Pr}}{\gamma_\smallfont{\Ar}}\exp\bigg(-\gamma_\Pr \int_0^T g(z^\star(t))\d t\bigg)=:{\rm V}^{\rm P,\star},
\end{align*}
where the upper bound is given by
\begin{align*}
\sup_{ z} \Big\{ \lambda_r^\star(z)-{k_r^o}^\star(z)-\frac{\gamma_A}2 |\sigma_r^\t  z|^2-\frac{\gamma_{\rm P}}{2}| \sigma_r^\t(1-z) |^2 \Big\}.
\end{align*}
Let us now show the upper bound is attained.  Indeed, letting
\[ Z_t^\star:=-\gamma_\Ar \Ur_\Ar^o( \widehat Y^{\hat R_\smallfont{0},z^\star}_t) z^\star(t),\;  \widehat Y_t^{\hat R_0,z^\star}:=\widehat R_0 -\int_0^t \bigg(\widehat H_r(z_r^\star) -\frac{\gamma_{\rm A}}2 |\sigma_r^\t  z^\star_r|^2-\frac{1}{\gamma_\Ar}\frac{f'(T-r)}{ f(T-r)}  \bigg)\d r +\int_0^t z_r^\star \d X_r,\; t\in [0,T],\]
it is easy to verify that the integrability assumption on $z^\star$ guarantees that $Z^\star\in\Hc^2$. We conclude that $\xi^\star\in \overline \Xi$, where $\xi^\star$ denotes the contract induced by $\widehat R_0$ and $Z^\star$, is optimal as it attains $\Vr^{\Pr,\star}$. this concludes the proof.

\subsection{Proof of Lemma {\rm \ref{lemma:sepagentreduction}}}
We argue $(i)$. Let $Z\in \Hc^{2,2}$. Recall
\begin{align*}
Y_t^{s,y_\smallfont{0},Z}=y_0^s-\int_0^t \big ( \lambda^\star_r(X_{\cdot\wedge r},Z_r^r)Z_r^s- f(r-s) c_r^\star(X_{\cdot\wedge r},Z_r^r) \big) \d r+\int_0^t  Z_r^s  \cdot \d X_r,
\end{align*}
and, in light of \eqref{eq:constraintgen}, we have
\[
{\Ur_\Ar^o}^{(-1)}\bigg(\frac{Y^{s,y_\smallfont{0},Z}_T}{f(T-s)}\bigg)={\Ur_\Ar^o}^{(-1)}\bigg(\frac{Y^{u,y_\smallfont{0},Z}_T}{f(T-u)}\bigg), \; (s,u)\in [0,T].
\]
Therefore, for any $s\in [0,T]$
\begin{align*}
0&=\frac{Y_0^{s,y_\smallfont{0},Z}}{f(T-s)}-\frac{Y_0^{0,y_\smallfont{0},Z}}{f(T)}- \int_0^T  \lambda^\star_r(X_{\cdot\wedge r},Z_r^r)\bigg( \frac{Z_r^s}{f(T-s)}- \frac{Z_r^0}{f(T)}\bigg)- c^\star_r\big(X_{\cdot\wedge r},Z_r^r\big)\bigg(\frac{f(r-s)}{f(T-s)}-\frac{f(r)}{f(T)}\bigg) \d r\\
&\quad+\int_0^T \bigg( \frac{Z_r^s}{f(T-s)}-\frac{Z_r^0}{f(T)} \bigg) \cdot \d X_r\\
&=\frac{Y_0^{s,y_\smallfont{0},Z}}{f(T-s)}-\frac{Y_0^{0,y_\smallfont{0},Z}}{f(T-u)}+\frac{1}{f(T-s)}\int_0^T \delta^\star_r(s,X_{\cdot\wedge r},Z_r^r ) \d r+\int_0^T \bigg( \frac{Z_r^s}{f(T-s)}-\frac{Z_r^0}{f(T)} \bigg) \cdot\big( \d X_r-\lambda^\star_r(X_{\cdot\wedge r},Z_r^r)\d r\big)\\
&=\frac{Y_t^{s,y_\smallfont{0},Z}}{f(T-s)}-\frac{Y_t^{0,y_\smallfont{0},Z}}{f(T)}+\frac{1}{f(T-s)}\int_t^T \delta^\star_r(s,X_{\cdot\wedge r},Z_r^r ) \d r+\int_t^T \bigg( \frac{Z_r^s}{f(T-s)}-\frac{Z_r^0}{f(T)} \bigg) \cdot\big( \d X_r-\lambda^\star_r(X_{\cdot\wedge r},Z_r^r)\d r\big).
\end{align*}
The result then follows taking conditional expectation thanks to the integrability of $Z^s$ and $Z^0$.\medskip

We now argue $(ii)$. Let $s\in[0,T]$ be fixed. Note that
\begin{align}\label{eq:dynN}
N_t^{s,Z}:= \E^{\P^\star(Z)}\bigg[ \int_t^T\delta^\star_r(s,X_{\cdot\wedge r},Z_r^r )\d r \bigg|\Fc_t\bigg]\nonumber & = M^{s,Z}_t -\int_0^t\delta^\star_r(s, X_{\cdot\wedge r},Z_r^r )\d r\nonumber \\
&= M_0^{s,Z}+\int_0^t\widetilde Z_r^{s,Z}\cdot \big( \mathrm{d} X_r- \lambda_r^\star(X_{\cdot \wedge r}, Z_r^r)\d r \big) - \int_0^t \delta^\star_r(s,X_{\cdot\wedge r},Z_r^r )\d r .
\end{align}
Therefore, in light of $(i)$, we have that there exists a finite variation process $A$ such that
\begin{align*}
Y_t^{s,y_\smallfont{0},Z}&= \frac{f(T-s)}{f(T)}Y_t^{0,y_\smallfont{0},Z}-\E^{\P^\star(Z)}\bigg[ \int_t^T\delta^\star_r(s,X_{\cdot\wedge r},Z_r^r )\d r \bigg|\Fc_t\bigg] \\
& = A_t -\int_0^t  \bigg(\frac{f(T-s)}{f(T)} Z_t^0 -\widetilde Z_t^{s,Z}\bigg)\cdot \big( \mathrm{d} X_r- \lambda_r^\star(X_{\cdot \wedge r}, Z_r^r)\d r \big), \; \P\as
\end{align*}
The result then follows from the uniqueness of the It\^o decomposition of $Y_t^{s,y_0,Z}$. \medskip

We are only left to argue $(iii)$ as $(iv)$ is a direct consequence. The inclusion $\Hc^{2,2}\subseteq \Hc^\bullet$ follows from $(ii)$ and taking $Y_t^Z:=Y_t^{0,y_\smallfont{0},Z}/f(T)$. Conversely, given $Z\in \Hc^\bullet$ we define for any $s\in [0,T]$
\[ Y_t^{0,y_\smallfont{0},Z}:= f(T)Y_t^{y_\smallfont{0},Z},\; Y_t^{s,y_\smallfont{0},Z}:=\frac{f(T-s)}{f(T)}Y_t^{0,y_\smallfont{0},Z}-N_t^{s,Z}, \; t\in [0,T], \;\P\as
\]
Let us note $Y_0^{s,y_\smallfont{0},Z}$ is clearly differentiable and $Y^{s,y_\smallfont{0},Z}$ satisfies \eqref{eq:constraintgen}. Indeed, as $N_T^{s,Z}=0$, $s\in [0,T]$, we have
\[
\frac{Y_T^{s,y_\smallfont{0},Z}}{f(T-s)}=Y_T^{y_\smallfont{0},Z}=\frac{Y_T^{u,y_\smallfont{0},Z}}{f(T-u)}, \; (s,u)\in [0,T]^2.
\]
We now verify $Y^{y_\smallfont{0},Z}\in \S^{2,2}$.  Let us first note that $\|\widetilde Z\|_{\H^{\smallfont{2}\smallfont{,}\smallfont{2}}}^2<\infty$. Indeed
\[
\|\widetilde Z^s\|_{\H^\smallfont{2}}^2=\E^\P\bigg[ \int_0^t |\sigma_r \sigma^\t_r \widetilde Z^{s,Z}_r|^2\mathrm{d}r\bigg]=\E^{\P^\smallfont{\star}(Z)}\bigg[M_t \int_0^t |\sigma_r \sigma^\t_r \widetilde Z^{s,Z}_r|^2\mathrm{d}r\bigg]\leq \E^{\P^\smallfont{\star}(Z)}\bigg [ \int_0^t |\sigma_r \sigma^\t_r \widetilde Z^{s,Z}_r|^2\mathrm{d}r\bigg]<\infty,
\]
where $M$ denotes the supermartingale given by
\[M_t:=\exp\bigg(- \int_0^t  b_r^\star(X_{\cdot\wedge r}, Z_r^r) \cdot    \d B_r^{\star}-\frac{1}2 \int_0^t | b_r^\star(X_{\cdot\wedge r}, Z_r^r)|^2\d r\bigg), \; t\in [0,T].\]

From this, it follows by \Cref{assump:datareward} that
\begin{align*}
\|N^Z\|_{\S^{\smallfont{2}\smallfont{,}\smallfont{2}}}^2\leq C\bigg( \E^{\P}\bigg[  \int_0^T\big| \lambda^\star_r(X,Z_r^r)\big|^2\d r\bigg]+\sup_{s\in [0,T]} \E^{\P}\bigg[  \int_0^T\big| \delta_r^\star(s,X,Z_r^r)\big|^2\d r\bigg]+\|\widetilde Z\|_{\H^{\smallfont{2}\smallfont{,}\smallfont{2}}}^2 \bigg)<\infty.
\end{align*}
We also note that the continuity of $s\longmapsto f(t-s)$ implies the continuity of $s\longmapsto \|N^{s,Z}\|_{\S^2}$. Moreover, $ \|Y^{0,y_\smallfont{0},Z}\|_{\S^\smallfont{2}}^2<\infty$ guarantees, by definition, that $\|Y^{y_\smallfont{0},Z}\|_{\S^{\smallfont{2}\smallfont{,}\smallfont{2}}}^2<\infty$. Moreover, by definition
\begin{align*}
Y_t^{s,y_\smallfont{0},Z}&=\frac{f(T-s)}{f(T)}\bigg( y_0+\int_0^t c_r^\star(X_{\cdot\wedge r}Z_r^r)f(r) \d r +\int_0^t Z_r^0 \cdot \big( \mathrm{d} X_r- \lambda_r^\star(X_{\cdot \wedge r}, Z_r^r)\d r \big) \bigg)\\
&\quad -M_0^{s,Z}-\int_0^t\widetilde Z_r^{s,Z}\cdot \big( \mathrm{d} X_r- \lambda_r^\star(X_{\cdot \wedge r}, Z_r^r)\d r \big) + \int_0^t \delta^\star_r(s,X_{\cdot\wedge r},Z_r^r )\d r \\
&= \frac{f(T-s)}{f(T)}  y_0-M_0^{s,Z} +\int_0^t c_r^\star(X_{\cdot\wedge r}Z_r^r)f(r-s) \d r  + \int_0^t  \bigg(\frac{f(T-s)}{f(T)} Z_r^0- \widetilde Z_r^{s,Z}\bigg) \cdot \big( \mathrm{d} X_r- \lambda_r^\star(X_{\cdot \wedge r}, Z_r^r)\d r \big)  \\
&= \frac{f(T-s)}{f(T)}  y_0-M_0^{s,Z} +\int_0^t  c_r^\star(X_{\cdot\wedge r}Z_r^r)f(r-s) \d r  + \int_0^t Z_r^s  \cdot \big( \mathrm{d} X_r- \lambda_r^\star(X_{\cdot \wedge r}, Z_r^r)\d r \big)  \\
&=y_0^{s} -\int_0^t  h_r^\star\big(s,X_{\cdot\wedge r}, Z_r^s,Z_r^r\big) \d r  + \int_0^t Z_r^s  \cdot  \d  X_r,
\end{align*}
where the third inequality follows from the fact $Z\in \Hc^\bullet$. We conclude $Z\in \Hc^{2,2}$.

\subsection{Proof of Lemma \ref{lemma:ra2agentreduction}}
Let us note that $(ii)$ and $(iii)$ are argued as in \Cref{lemma:sepagentreduction}. We now argue $(i)$. Let $(y_0,Z)\in \Ic\times \Hc^{2,2}$. Note that given $Y^{s,y_0,Z}$, in light of the regularity of $y^s_0$ and the generator, it is possible to define $\partial Y^{s,y_\smallfont{0},Z}$ such that 
\begin{align*}
Y_t^{t,y_\smallfont{0},Z}=& y_0^0-\int_0^t \Big( H_r\big(X_{\cdot\wedge r},Y_r^{r,y_\smallfont{0},Z},Z_r^r\big)-\partial Y_r^{r,y_{\smallfont0},Z} \Big) \d r+\int_0^t  Z_r^r \cdot  \mathrm{d} X_r,\\
Y_t^{s,y_\smallfont{0},Z}=&y_0^s-\int_0^t h_r^\star\big(s,X_{\cdot\wedge r},Y_r^{s,y_\smallfont{0},Z} , Z_r^s,Y_r^{r,y_\smallfont{0},Z} , Z_r^r\big) \d r+\int_0^t  Z_r^s  \cdot \d X_r,\\
\partial Y _t^{s,y_\smallfont{0},Z}=&\partial y_0^s-\int_0^t \nabla h_r^\star\big(s,X_{\cdot\wedge r},\partial Y_r^{s,y_\smallfont{0},Z} , \partial Z_r^s, Y_r^{s,y_\smallfont{0},Z} ,Y_r^{r,y_\smallfont{0},Z} , Z_r^r\big) \d r+\int_0^t  \partial Z_r^s  \cdot \d X_r.
\end{align*}
Letting
\[ \widehat Z_t:=\frac1{-\gamma_\Ar}\frac{ Z_t^t}{ Y_t^{t,y_\smallfont{0},Z}},  \text{ and, } \widehat Z^s_t:=\frac1{-\gamma_\Ar}\frac{ Z_t^s}{ Y_t^{s,y_\smallfont{0},Z} },\]
we obtain that
\begin{align*}
Y_t^{t,y_\smallfont{0},Z}=& y_0^0-\int_0^t \Big( -\gamma_\Ar Y_r^{r,y_\smallfont{0},Z} \widehat H_r\big(X_{\cdot\wedge r},Z_r^r\big)-\partial Y_r^{r,y_\smallfont{0},Z} \Big) \d r+\int_0^t  Z_r^r \cdot  \mathrm{d} X_r,\\
Y_t^{s,y_0,Z}=&y_0^s-\int_0^t -\gamma_\Ar Y_r^{s.y_\smallfont{0},Z} \widehat h_r^\star\big(s,X_{\cdot\wedge r}, \widehat Z_r^s,\widehat Z_r^r\big) \d r+\int_0^t  Z_r^s  \cdot \d X_r,\\
\partial Y _t^{s,y_{\smallfont0},Z}=&\partial y_0^s-\int_0^t \nabla h_r^\star\big(s,X_{\cdot\wedge r},\partial Y_r^{s,y_\smallfont{0},Z} , \partial Z_r^s, Y_r^{s,y_\smallfont{0},Z} , \widehat  Z_r^r\big) \d r+\int_0^t  \partial Z_r^s  \cdot \d X_r.
\end{align*}

The result then follows by It\^o's formula introducing
\begin{align*}
\widehat Y^{s,y_\smallfont{0},Z}_t:=-\frac{1}{\gamma_\Ar}\ln(-\gamma_\Ar Y^{s,y_\smallfont{0},Z}_t),\;  \partial \widehat Y^{s,y_\smallfont{0},Z}_t:=\frac1{-\gamma_\Ar} \frac{\partial Y_t^{s,y_\smallfont{0},Z}}{Y_t^{s,y_\smallfont{0},Z}}, \text{ and }\partial \widehat Z^s_t:= \frac1{-\gamma_\Ar}\bigg(\frac{\partial Z_t^s}{ Y_t^{s,y_\smallfont{0},Z}}+\gamma_\Ar^2\partial \widehat Y_t^{s,y_0,Z}  \widehat Z_t^s \bigg).
\end{align*}
\begin{comment}
\begin{align*}
\widehat Y_t^{t,Z}&=-\frac{1}{\gamma_\Ar} \ln\big(-\gamma_\Ar  Y_0^0\big)-\int_0^t \Big( \widehat H_r\big(X_{\cdot\wedge r},\widehat Z_r^r\big)-\frac{\gamma_\Ar}2 |\sigma_r^\t \widehat Z_r^r|^2-\partial \widehat Y_r^{r,Z} \Big) \d r+\int_0^t  \widehat Z_r^r \cdot  \d  X_r, \\
 \widehat Y_t^{s,Z} &=-\frac{1}{\gamma_\Ar} \ln\big(-\gamma_\Ar  Y_0^s\big)-\int_0^t \Big( \widehat h_r^\star\big(s,X_{\cdot\wedge r} , \widehat Z_r^s, \widehat Z_r^r\big)- \frac{\gamma_\Ar}2 |\sigma_r^\t \widehat Z_r^s |^2\Big)\d r+\int_0^t  \widehat Z_r^s \cdot \d X_r, \\
\partial \widehat Y_t^{s,Z}&=- \frac{\partial Y_0^s}{\gamma_\Ar Y_0^s} -\int_0^t  \nabla \widehat h_r^\star\big(s,X_{\cdot\wedge r},\partial \widehat Z_r^s, \widehat Z_r^s,\widehat Z_r^r \big) \d  r+ \int_0^t \partial \widehat Z_r^s   \cdot \d  X_r.
\end{align*}
\end{comment}

\subsection{Proof of Proposition {\rm \ref{prop:sol2Bra2}}}
Note that it always holds that, $\P^\star(Z)\as$
\begin{align*}
X_T-  {\Ur_\Ar^o}^{(-1)}(Y_T^{0,Z})/g(T) & =x_0-\frac{\Ur_\Ar^{(-1)}(Y_0^0)}{g(T)}+\int_0^T\! \bigg( \lambda^\star_r(X_{\cdot\wedge r},\widehat Z_r^r) -\frac{g(r)}{g(T)} {k^o_r}^\star(X_{\cdot\wedge r},\widehat Z_r^r)  -\frac{\gamma_{\rm A}}{2g(T)} |\sigma_r^\t(X_{\cdot\wedge r}) \widehat Z_r^0 |^2 \bigg)\d r \\
&\quad +\int_0^T \bigg(1-\frac{\widehat Z_r^0}{g(T)}\bigg) \cdot   \big(\d X_r-\lambda^\star_r(X_{\cdot\wedge r},\widehat Z_r^r)\d r\big),
\end{align*}

so that
\begin{align*}
{\Ur_\Pr^o}\Big(X_T- {\Ur_\Ar^o}^{(-1)}(Y_T^{0,Z})/g(T) \Big) =C_{ \hat{Y}_\smallfont{0}^\smallfont{0}}M_T\exp\bigg(-\gamma_\Pr  \int_0^T G_r(X_{\cdot \wedge r},\widehat Z_r,\widehat Z_r^0) \d r\bigg) ,\; \P\as,
\end{align*}
where
\[ G_t(x,z,v):= \lambda^{\star}_t(x,z) -\frac{g(t)}{g(T)}{k^o_t}^\star(x,z) -\frac{\gamma_A}{2g(T)} |\sigma_t^\t(x) v|^2-\frac{\gamma_{\rm P}}{2}\bigg|\sigma_r^\t(x)\bigg(1-\frac{v}{g(T)} \bigg) \bigg|^2,\]
and $M$ denotes the supermartingale
\[M_t:=\exp\bigg(- \gamma_\Pr\int_0^t  \sigma_r^\t(X_{\cdot\wedge r}) \bigg(1-\frac{\widehat Z_r^0}{g(T)}\bigg) \cdot    \d B_r^{a^\star(Z)}-\frac{\gamma_\Pr^2}2 \int_0^t \bigg|\sigma_r^\t(X_{\cdot\wedge r})\bigg(1-\frac{\widehat Z_r^0}{g(T)}\bigg)\bigg|^2\d r\bigg), \; t\in [0,T].\]
Consequently
\begin{align}\label{eq:3rdexample:1}
\E^{\P^\smallfont{\star}(Z)}\big [ {\rm U}_{\rm P}^o \big( X_T- {\Ur_\Ar^o}^{(-1)}(Y_T^{0,Z})/g(T) \big)\big]  \leq C_{ \hat{Y}_\smallfont{0}^\smallfont{0}} \E^{\P}\bigg[\exp\bigg(-\gamma_{\rm P}\int_0^T G_r(X_{\cdot \wedge r},\widehat Z_r^r,\widehat Z_r^0) \d r \bigg)\bigg].
\end{align}

Now, under the additional assumptions, we have that for any $Z\in \widetilde \Hc$
\begin{align*}
\E^{\P^\smallfont{\star}(Z)}\big [ {\rm U}_{\rm P}^o \big( X_T- {\Ur_\Ar^o}^{(-1)}(Y_T^{0,Z})/g(T) \big)\big]  \leq C_{\hat  R_\smallfont{0}} \E^{\P}\bigg[\exp\bigg(-\gamma_{\rm P}\int_0^T G_r\Big(\zeta_r, \eta_{0,r} \zeta_r\Big) \d r \bigg)\bigg]
\end{align*}

Therefore, as
\[
z^\star(t,\eta)\in \argmax_{z\in\R} \bigg\{ \lambda^{\star}_t(z) -\frac{g(t)}{g(T)}{k^o_t}^\star(z) -\frac{\gamma_A}{2g(T)} |\sigma_t^\t  z |^2|\eta|^2-\frac{\gamma_{\rm P}}{2}\bigg|\sigma_r^\t\bigg(1-\frac{z}{g(T)}\eta  \bigg) \bigg|^2\bigg\},
\]
as longs as $\eta$ is chosen so that $Z^{\eta}\in \widetilde \Hc$ the upper bound is attained. \medskip

Let us argue the second part of the statement.  Since we are now constrained to deterministic choices of $\eta$ the integrability of $z^\star$ and the boundedness of $[0,T]\ni t\longmapsto g(t)$ guarantee that the constant process $M^{s,Z}$ in \Cref{lemma:ra1agentreduction} is finite and thus square integrable. Therefore, as the contract induced by the family
\[Z^s_t:=\frac{g(T-s)}{g(T-t)}z^\star(t,\eta_{0,t}^\star),\; \eta_{0,t}^\star=\frac{g(T)}{g(T-t)},\]
attains the upper bound in \eqref{eq:3rdexample:1}, the result follows. The last statement follows letting $(\gamma_\Ar, \gamma_\Pr)\longrightarrow (0,0)$ and noticing the terms involving $\widehat Z^0$ in \eqref{eq:3rdexample:1} vanish. Therefore the upper bound is attained by the maximiser of $G(z)= \lambda^{\star}_t(z) - g(t){k^o_t}^\star(z)/g(T)$, \emph{i.e.} the deterministic contract given by $Z^s_t=f(T-s)z^\star(t)/f(T-t)$.

\end{appendix}

{\small
\bibliography{bibliography}}
\end{document}